\documentclass[12pt]{article}
\usepackage{geometry}
\usepackage{amsmath}
\usepackage{amssymb}
\usepackage{amsthm}

\usepackage{mathtools}	      
\usepackage{multirow}		
\usepackage[inline]{enumitem}		

\usepackage{graphicx}
\usepackage{placeins}

\usepackage{appendix}
\usepackage{prettyref}
\newrefformat{hypo}{H\ref{#1}}
\newcounter{hypothesis}
\renewcommand{\thehypothesis}{\arabic{hypothesis}}

\usepackage[numbers]{natbib}
\usepackage{color}
\usepackage[hang,flushmargin]{footmisc}
\usepackage[hidelinks]{hyperref}

\newtheorem{thm}{\sc Theorem}
\newtheorem{cor}{\sc Corollary}
\newtheorem{lem}{\sc Lemma}

\definecolor{asparagus}{rgb}{0.53, 0.66, 0.42} 
\definecolor{arsenic}{rgb}{0.23, 0.27, 0.29}
\newcommand{\graytext}[1]{\textbf{\textcolor{arsenic}{#1}}}
\definecolor{light}{gray}{.50}
 
\DeclareMathOperator{\E}{E}					
\DeclareMathOperator{\var}{Var}			
\newcommand{\EX}{\E_{\bX}}					
\newcommand{\varX}{\var_{\bX}}			
\newcommand{\dd}{\,\mathrm{d}\,}		
\newcommand{\tr}{^\mathsf{T}}				
\newcommand{\as}{\textrm{a.s.}}			
\newcommand{\PP}{P}					
\DeclareMathOperator{\cov}{Cov}			

\def\m1{\boldsymbol{\mathbb{1}}}
\def\R{\mathbb{R}}

\def\bzero{\boldsymbol{0}}
\def\b1{\boldsymbol{1}}

\def\bb{\boldsymbol{b}}

\def\bdelta{\boldsymbol{\delta}}

\def\bD{\boldsymbol{D}}
\def\bDelta{\boldsymbol{\Delta}}
\def\be{\boldsymbol{e}}
\def\bgamma{\boldsymbol{\gamma}}
\def\bGamma{\boldsymbol{\Gamma}}
\def\bI{\boldsymbol{I}}
\def\bK{\boldsymbol{K}}
\def\bLambda{\boldsymbol{\Lambda}}
\def\bM{\boldsymbol{M}}
\def\bphi{\boldsymbol{\phi}}
\def\bPhi{\boldsymbol{\Phi}}
\def\bR{\boldsymbol{R}}
\def\btheta{\boldsymbol{\theta}}
\def\bu{\boldsymbol{u}}
\def\bv{\boldsymbol{v}}

\def\bX{\boldsymbol{X}}
\def\bY{\boldsymbol{Y}}
\def\bZ{\boldsymbol{Z}}
\def\bz{\boldsymbol{z}}
\def\bU{\boldsymbol{U}}

\def\cP{\mathcal P}
\def\cS{\mathcal S}

\def\ds{\displaystyle}

\begin{document}

\centerline{\large\bf SCALAR-ON-FUNCTION LOCAL LINEAR} ~ \\ 
\centerline{\large\bf REGRESSION AND BEYOND}

\bigskip
 
\centerline{\large\sc Fr\'ed\'eric Ferraty\footnote{University of Toulouse, Toulouse Mathematics Institute, France. \newline e-mail: \url{ferraty@math.univ-toulouse.fr}} and Stanislav Nagy\footnote{Charles University, Prague, Faculty of Mathematics and Physics, Department of Probability and Math. Statistics, Czech Republic. \newline e-mail: \url{nagy@karlin.mff.cuni.cz}}}

\begin{abstract} 
Regressing a scalar response on a random function is nowadays a common situation. In the nonparametric setting, this paper paves the way for making the local linear regression based on a projection approach a prominent method for solving this regression problem. Our asymptotic results demonstrate that the functional local linear regression outperforms its functional local constant counterpart. Beyond the estimation of the regression operator itself, the local linear regression is also a useful tool for predicting the functional derivative of the regression operator, a promising mathematical object on its own. The local linear estimator of the functional derivative is shown to be consistent. On simulated datasets we illustrate good finite sample properties of both proposed methods. On a real data example of a single-functional index model we indicate how the functional derivative of the regression operator provides an original and fast, widely applicable estimating method. 
\end{abstract}

\noindent {\bf keywords}: {Asymptotics, functional data, functional derivative of regression operator, functional index model, local linear regression, scalar-on-function regression}


\section{Introduction}
Functional data analysis is a toolbox of statistical techniques for dealing with datasets of random functions \cite{RamsS02, RamsS05, HsinE15, KokoR17}. When regressing a scalar response $Y$ on an explanatory random function $X$ using the model $Y = m(X) + error$, with $m$ unknown, common terminology refers to the scalar-on-function regression. In this setting, linear modeling \cite{CardFS99, CaiH06, CramKS09} or its generalized versions \cite{Jame02, MullS05} have been intensively studied in the literature. Although the scalar-on-function linear regression is a powerful tool, its lack of flexibility when nonlinearities occur led the statistical community to develop a nonparametric approach as it was in the multivariate case. For an overview on the functional version of the Nadaraya-Watson kernel estimator see \cite{FerrV06}. Since the development of local linear estimation for multivariate data in the 1990's (see for instance \cite{Fan92, FanG92, Fan93, RuppW94, ChenFM97, HallM97} and the monograph \cite{FanG96}), it is well known that local linear regression outperforms the usual Nadaraya-Watson kernel estimator (i.e. the local constant regressor). Therefore, it became perhaps the most popular nonparametric regression technique. Surprisingly, in the framework of functional data, there are only two papers that focus on the estimation of the regression operator in the scalar-on-function local linear regression model. In \cite{BailG09}, a projection approach to the problem similar to the study presented here is proposed, but the asymptotics derived in that paper appears to suffer from a lack of rigorousness. Paper \cite{BerlEM11} is a pure theoretical work providing an alternative estimating procedure by regularizing a non-bounded linear operator. The latter method does not directly relate to the approach taken here, as it does not rely on the use of projections. All in all, the theory and practice of scalar-on-function local linear regression is severely underdeveloped, and thus the method is far from being as popular as it is in the multivariate case. 

This paper introduces the local linear regression as an indispensable tool in the setting of scalar-on-function nonparametric regression. It turns out that the \emph{functional local linear regression} (that is, local linear regression when the regressor is a random function) is not only a smart method of estimating the regression operator. As an exciting by-product we obtain an easy and fast method for estimating the functional derivative $m'_x$ of the regression operator $m$ at any function $x$. The functional derivative is a linear functional that represents a local linear approximation to the regression operator $m$ around $x$ (for a precise statement see (\prettyref{hypo:taylor}) below). In what follows, we use the Riesz representation theorem, and identify the functional derivative $m'_x$ with its unique representing function. What makes the estimation of functional derivatives of such great interest? A first motivation is given in the pioneering works \cite{HallMY09, MullY10} where estimating procedures are developed without considering the local linear regression setting. There, it was convincingly demonstrated that the concept of functional derivative greatly facilitates the interpretation of results. As a further step in the pursuit for understanding how one can use the functional derivative in a natural way, let us consider the functional Taylor expansion of the regression operator. For a small positive real $\eta$ and a direction $u$ (i.e. a function $u$ such that $\|u\|=1$), Taylor's expansion and the Riesz representation theorem allow us to write $m(x + \eta \, u) - m(x) = \eta \langle m_x', u \rangle + O(\eta^2)$. A first order approximation of the magnitude of the difference $m(x + \eta \, u) - m(x)$ is therefore the interval $[-\eta \|m'_x\|, \eta \|m'_x\| ]$ --- the smaller $\|m'_x\|$ is, the less sensitive to small perturbations in $x$ is $m$. In a sense, $\|m'_x\|$ can be seen as a measure of reliability for the prediction of $m$ at $x$. For another example where the functional derivative appears as a successful tool in interesting statistical problems, consider, for instance, the single-functional index model \cite{AmatAD06, AitFKV08, ChenHM11, JianW11}. That model takes the form $m(x) = \mu + g\left( \langle \beta, x   \rangle \right)$, where $\mu$ is an unknown scalar and the scalar response interacts with the functional covariate only through an unknown functional direction $\beta$ combined with an unknown real-valued link function $g$. Extending the \emph{average derivative estimation} method introduced in \cite{HardS89} to the functional setting, it is easy to show that $\E\left(m'_X \right)$ is proportional to the functional direction $\beta$. Thus, given a sample $(X_1, Y_1), \ldots,(X_n, Y_n)$, as soon as one is able to obtain estimates $\widehat{m_{X_1}'}, \ldots, \widehat{m_{X_n}'}$ of  the functional derivatives $m_{X_1}', \ldots, m_{X_n}'$, one can compute $\widehat{\E\left(m'_X \right)} \coloneqq n^{-1}\sum_i \widehat{m_{X_i}'}$. The quantity $\widehat{\E\left(m'_X \right)}  / \| \widehat{\E\left(m'_X \right)} \|$ is a reasonable estimator of the functional index $\beta$. 

All the examples above emphasize the major role that the functional derivative of the regression operator plays in important aspects of statistics: interpretation, reliability and methodology. This is why we propose to revisit the functional local linear regression by focusing not only on the  regression operator, but also on its functional derivative. In this work, a projection approach to functional local linear estimation is adopted. Theoretical properties of the estimator of the regression operator $m$ are stated. Simultaneously, a local linear estimator of the functional derivative $m'_x$ at $x$ is proposed. For both estimators of the regression operator and its functional derivative, original technical tools are developed in order to study their theoretical behavior. If the implementation of the estimator of the regression operator is straightforward, selection of the smoothing parameter for estimating its functional derivative poses a major challenge. This is why an ad hoc bootstrap procedure is introduced to pilot the bandwidth choice. Ease of implementation as well as nice finite sample properties of our local linear estimators are highlighted in a simulation study. The paper is concluded with a benchmark real dataset that is used to illustrate the important role of the functional derivatives in functional data analysis. Most of the proofs given in the paper are deferred to Appendices~\ref{app:A} and~\ref{app:B}. The paper is complemented by extensive supplementary material that includes an efficient \texttt{R} \cite{R18} implementation of the newly proposed methods\footnote{Available at \url{https://bitbucket.org/StanislavNagy/fllr}} and some additional examples in Appendix~\ref{app:C}, and theoretical results in Appendix~\ref{app:D}.
\section{Functional local linear estimation}
Let $X$ be an $H$-valued random function where $H$ is the separable Hilbert space of square integrable functions defined on $[0, 1]$ equipped with the inner product $\langle \cdot , \cdot \rangle$, and let $\| \cdot \|$ be the associated norm. This work focuses on the relationship between $X$ and a scalar response $Y$ by considering the nonparametric regression model $Y \, = \, m(X) \, + \, \varepsilon$ where $\E(\varepsilon|X)=0$. The regression operator $m$ mapping $H$ into $\R$ is unknown, and is assumed to be smooth enough in a neighborhood $\mathcal{N}_x$ of a given $x \in H$. 
\begin{itemize}
\item[\refstepcounter{hypothesis} \label{hypo:taylor} (H\thehypothesis)] For any $u \in \mathcal{N}_x$, there exists $\zeta=x + t\,u$ with $t\in (0, 1)$ such that
$$
m(x+u) \, = \, m(x) + \langle m_x', u \rangle + \frac{1}{2} \langle m_\zeta'' u, u \rangle,
$$  
where  $m'_x \in H$, $m''_\zeta$ is a Hilbert-Schmidt linear operator mapping $H$ into $H$, and $v \mapsto m_{v}''$ is Lipschitz in $v \in \mathcal{N}_x$.
\end{itemize}
In other words, one focuses on those regression operators $m$ for which the second order Taylor expansion is valid. Note that condition (\prettyref{hypo:taylor}) is the functional counterpart of what is standardly required in the finite-dimensional local linear regression setting.

Based on an $n$-sample $(X_i, \, Y_i)_{i=1,\ldots,n}$ of independent identically distributed (iid) copies of $(X,\,Y)$, our main task is to estimate the regression operator $m$, as well as the functional derivative $m'_x$. To this end, one extends the sum of weighted squared errors (SWSE) to the functional setting
	\[	SWSE(a; \, \beta)\, \coloneqq \, \sum_{i=1}^n \left( Y_i - a - \langle \beta, X_i-x \rangle \right)^2 K\left( h^{-1} \| X_i - x\| \right),	\]
where $K(\cdot)$ is a kernel function defined on $[0,1]$ and $h$ is a positive smoothing parameter (a bandwidth). The principle of the local linear criterion is to linearize the regression operator in a neighborhood of $x$. In other words, for any $X_i$ close to $x$, one considers that $\E(Y_i|X_i)=a + \langle \beta, X_i-x \rangle$. Then, the real $a$ (resp. the square integrable function $\beta$) can be interpreted as the regression operator (resp. the functional derivative) at $x$. Consequently, the estimation of the functional derivative $m'_x$ will be based on the estimation of the function $\beta$. Instead of focusing on $\beta$ itself, it is computationally advantageous to consider its projection $\sum_{j\leq J} \langle \phi_j, \beta \rangle \phi_j$ onto the $J$-dimensional subspace $\cS_J$ of $H$ spanned by the orthonormal sequence $\phi_1,\ldots, \phi_J$ that is completed by $ \phi_{J+1}, \phi_{J+2},\ldots$ in order to get an orthonormal basis of $H$. In this notation the SWSE criterion is approximated by 
$$
SWSE_J(a; \, b_1, \ldots, b_J)\, \coloneqq \, \sum_{i=1}^n \left( Y_i - a - \sum_{j=1}^J b_j \langle \phi_j,\, X_i -x \rangle \right)^2 K\left( h^{-1} \| X_i - x\| \right), 
$$
where for any $j$, $b_j \coloneqq \langle \phi_j, \beta \rangle$. 
Then, $\widehat{m}(x) \coloneqq \widehat{a}$ and $\widehat{m'_x} \coloneqq \sum_{j=1}^J\widehat{b}_j \phi_j$ where $ (\widehat{a}; \, \widehat{b}_1, \ldots, \widehat{b}_J)\, \coloneqq \, \arg \inf_{(a; \, b_1, \ldots, b_J)} SWSE_J(a; \, b_1, \ldots, b_J).$
By using vector and matrix notations, one is able to express the local linear estimators. Let $\bY\coloneqq[Y_1, \dots, Y_n]\tr$, $\bK\coloneqq\mbox{diag}\left\{  K\left( h^{-1} \| X_1 - x\| \right), \dots, K\left( h^{-1} \| X_n - x\| \right)  \right\}$ the diagonal $n\times n$ matrix and $\bPhi$ the following $n\times(J+1)$ matrix
$$
\bPhi \coloneqq \left[  
\begin{array}{cccc}
1 & \langle \phi_1, \, X_1-x \rangle & \cdots & \langle \phi_J, \, X_1-x \rangle \\ 
\vdots & \vdots & \ddots & \vdots \\
1 & \langle \phi_1, \, X_n-x \rangle & \cdots & \langle \phi_J, \, X_n-x \rangle 
\end{array}
\right].
$$
Define $\widehat{\bb} \coloneqq \left[\widehat{b}_1, \dots,  \widehat{b}_J\right]\tr$, $\bphi \coloneqq \left[\phi_1, \dots, \phi_J\right]\tr$, $\bzero$ the $J\times 1$ null vector, $\be$ the $(J+1)$-dimensional vector $\left[1, \bzero\tr\right]\tr$ and $\bI$ the $J\times J$ identity matrix. Then, it is easy to see that $\left[ \widehat{a} \, | \, \widehat{\bb}\tr \right]\tr \, = \, \left( \bPhi\tr \bK \bPhi \right)^{-1} \bPhi\tr \bK \bY$, $\widehat{m}(x) \, = \,  \be\tr \,\left( \bPhi\tr \bK \bPhi  \right)^{-1} \bPhi\tr \bK \bY$ and $\widehat{m'_x}\, = \,  \bphi\tr \, \left[ \bzero | \bI \right] \left( \bPhi\tr \bK \bPhi \right)^{-1} \bPhi\tr \bK \bY$. The local linear approach has the nice property that both $\widehat{m}(x)$ and $\widehat{m'_x}$ are based on the common terms $\left( \bPhi\tr \bK \bPhi \right)^{-1} \bPhi\tr \bK \bY$ and this is why the raw computational cost of $\widehat{m'_x}$ is not much higher than the one for $\widehat{m}(x)$. 
\section{Asymptotic study}
Let us first focus on the assumptions we need to derive the asymptotic behavior of $\widehat{m}(x)$ and $\widehat{m'_x}$.

\begin{itemize}
\item[\refstepcounter{hypothesis} \label{hypo:kernel} (H\thehypothesis)] The kernel function $K$ is continuously differentiable on its support $(0, 1)$ with $K'(s)\leq 0$ for all $s \in (0,1)$ and $K(1)>0$.
\item[\refstepcounter{hypothesis} \label{hypo:gammafunction} (H\thehypothesis)] For any integers $j_1,\ldots, j_M, p_1,\ldots, p_M \geq 0$ with $M\geq 1$, let us define  \linebreak  $\gamma_{j_1,\ldots,j_M}^{p_1,\ldots,p_M}(t) \coloneqq \E \left( \langle \phi_{j_1}, X_1-x \rangle^{p_1}  \cdots  \langle \phi_{j_M}, X_1-x \rangle^{p_M} | \| X_1 - x \|^{p_1 + \cdots + p_M} = t \right)$ \linebreak and let $ {\gamma_{j_1,\ldots,j_M}^{p_1,\ldots,p_M}}^{'}(t)$ be its derivative at $t$. The functions $\gamma_{j_1}^1$, $\gamma_{j_1, j_2}^{1,1}$, $\ldots$, $\gamma_{j_1,\ldots, j_4}^{1,\ldots,1}$, $\gamma_{j_1}^2$, $\gamma_{j_1,j_2,j_3}^{2,1,1}$, $\gamma_{j_1, \dots, j_5}^{2,1,1,1,1}$ and $\gamma_{j_1, j_2}^{2,2}$ are assumed to be continuously differentiable around zero and the smallest eigenvalue $\lambda_J$ of the $J\times J$ matrix $\boldsymbol{ \Gamma }$, whose $(j,k)$-th element is defined by $\left[\boldsymbol{ \Gamma }\right]_{jk}\coloneqq{\gamma_{j,k}^{1,1}}^{'}(0)$, is strictly positive.
\item[\refstepcounter{hypothesis} \label{hypo:asymptotics} (H\thehypothesis)] $h=h_n$ tends to 0 with $n$,  $J=J_n$ and $n \, \pi_x(h)$ grow to infinity with $n$ so that $h \, J^{1/2} = o(1)$,  $h^{-1} \lambda_J^{-1}  \{n \, \pi_x(h) \}^{-1/2} = o(1)$ and $h^{-1} (\lambda_J/J)^{-1/2} \{n \, \pi_x(h) \}^{-1/2} = o(1)$, where $ \pi_x(h)\coloneqq P\left( \| X_1 - x \| < h \right)$.
\item[\refstepcounter{hypothesis} \label{hypo:sbp1} (H\thehypothesis)] For all $s$ in $[0,1]$, the ratio $\ds \tau_{x,h}(s)\coloneqq\frac{\pi_x(hs)}{\pi_x(h)}$ tends to $\tau_{x}(s)$ as $h$ goes to 0.
\item[\refstepcounter{hypothesis} \label{hypo:condvar} (H\thehypothesis)] The conditional variance of the error $\sigma^2(x) = \var(Y | X = x)$ mapping $H$ into $\R$ is a uniformly continuous operator.
\end{itemize}

The original hypothesis (\prettyref{hypo:gammafunction}) introduces a particular family of functions. It is shown in {\sc \prettyref{lem:gammafunction} } (postponed to Appendix~\ref{app:B}) that the functions $\gamma_{j_1,\ldots,j_M}^{p_1,\ldots,p_M}$ have interesting properties, in particular the matrix $\boldsymbol{ \Gamma }$ is positive semi-definite. Of course, since the size $J$ of the square matrix $\boldsymbol{ \Gamma }$ tends to infinity with $n$, it is clear that its smallest eigenvalue $\lambda_J$ tends to zero with $n$.  Nevertheless, it is not too much restrictive to require $\lambda_J$ strictly positive for any fixed $J$. It is worth noting that the differentiability assumptions imposed on $\gamma_{j_1,\ldots,j_M}^{p_1,\ldots,p_M}$ make this function very useful to approximate $\E\left\{\langle \phi_{j_1}, X_1-x \rangle^{p_1}  \cdots  \langle \phi_{j_M}, X_1-x \rangle^{p_M}\,  K^q \left( h^{-1} \| X_1 - x \|  \right)\right\}$ that plays a major role in the asymptotic behavior of the estimators $\widehat{m}_x$ and $\widehat{m}'_x$. The reader will find in Appendix~\ref{app:D} (see {\sc Lemma} \ref{lem:aboutH3}) an interesting result providing a general situation where (\prettyref{hypo:gammafunction}) is fulfilled. Condition (\prettyref{hypo:sbp1}) is a more classical assumption. For standard families of processes an explicit form of the function $\tau_x(s)$ is available (for more details, see \cite{FerrMV07} and references therein). 

Let $\bX$ stand for the sample $X_1, \ldots, X_n$ and let $\EX$ (resp. $\varX$) be the conditional expectation (resp. variance) with respect to $\bX$. The asymptotic conditional bias and variance of $ \widehat{m}(x)$ are provided in the following theorem.
\begin{thm} \label{thm:biasvar}  Under conditions  (\prettyref{hypo:taylor})--(\prettyref{hypo:condvar}), 
\begin{enumerate}[label=(\roman*)]
\item $\EX \left\{ \widehat{m}(x) \right\} \, = \, m(x) \, + \, O_P\left( \| \cP_{\cS_J^\perp} m'_x \| \, h \right)  \, + \, O_P\left( h^2 \right),$ 
\item $\varX \left\{ \widehat{m}(x) \right\} \, = \, O_P\left( \{n \, \pi_x(h)  \}^{-1} \right)$,
\end{enumerate}
where $\cS_J^\perp$ is the orthogonal complement to the space $\cS_J$ in $H$, and $\cP_{\cS_J^\perp} \colon H \to \cS_J^{\perp}$ is the orthogonal projection onto $\cS_J^\perp$.
\end{thm}
In the same situation (see for instance \cite{FerrV06}), the conditional bias of the estimated functional local constant regression is of order $h$ with the same conditional variance. However, the quantity $\| \cP_{\cS_J^\perp} m'_x \|$ tends to zero when $n$ grows to infinity since $m'_x$ is a square integrable function. Therefore $\widehat{m}(x)$ outperforms the asymptotic behavior of the kernel estimator of the functional local constant regression. Note that the conditional variance in {\sc \prettyref{thm:biasvar}} involves the small ball probability $\pi_x(h)$, which is standard in the functional nonparametric setting \cite{FerrV06}.

The functional setting involves implicitly the dimension $J$ of the approximating subspace  $\cS_J$ in the rate of convergence. Indeed, the asymptotic behavior of the conditional bias depends on the quantity $\| \cP_{\cS_J^\perp} m'_x \|$, which assesses the approximation error of the functional derivative $m'_x$ in $\cS_J$. From a theoretical point of view, $ \widehat{m}(x)$ involves $J\times J$ matrices and $J$-dimensional vectors with $J$ converging to infinity. This makes the asymptotic study much harder in comparison with the multivariate (finite-dimensional) setting. The issue of infinite dimension is overcome by deriving accurately, in an element-wise sense, the asymptotic behavior of the matrices and vectors involved in both local linear estimators.

\begin{proof}[Proof of the conditional bias of $\widehat{m}(x)$] Here only the main guidelines of the proof are given. Details and technical lemmas are postponed to the appendix. 

For any $X_i$, the Taylor expansion of $m$ at $x$ can be expressed as 
\begin{equation}\label{eq:taylor}
m(X_i) \, = \, m(x) + \langle m_x', X_i-x \rangle + \frac{1}{2}\langle m_\zeta''(X_i-x), X_i-x \rangle \quad \as,
\end{equation}
where $\zeta=x + t(X_i-x)$ with $t\in (0, 1)$. By using a basis expansion of $m_x'$ in terms of $\phi_1, \phi_2, \ldots$, we get $m(X_i) = m(x) +  \sum_{j\leq J} \langle m_x', \phi_j \rangle \langle \phi_j, X_i-x \rangle  + \langle\cP_{\cS_J^\perp} m_x', X_i-x \rangle + R_{\zeta, x, i}/2$ almost surely with $R_{\zeta, x, i} \coloneqq \langle m_\zeta''(X_i-x), X_i-x \rangle$. Let $\nabla m_x \coloneqq \left[\langle m_x', \phi_1 \rangle, \ldots, \langle m_x', \phi_J \rangle  \right]\tr$ be the first $J$ coordinates of the gradient of $m$ at $x$ and $\bR_{\zeta, x} \coloneqq [R_{\zeta, x, 1}, \ldots,R_{\zeta, x, n}]\tr$. Then we can write
\begin{equation}	\label{eq:expansion}
\left[ \begin{array}{c} m(X_1) \\ \vdots \\ m(X_n) \end{array} \right] \, = \, \bPhi \, \left[ \begin{array}{c} m(x) \\  \nabla m_x \end{array} \right] +  \left[ \begin{array}{c} \langle \cP_{\cS_J^\perp} m_x', X_1-x \rangle \\ \vdots \\ \langle \cP_{\cS_J^\perp} m_x', X_n-x \rangle  \end{array} \right]  +  \frac{1}{2} \,  \bR_{\zeta, x} \quad \as 
\end{equation}
According to the definition of $\widehat{m}(x)$, 
\begin{equation} \label{eq:bias}
\begin{aligned} 
\EX \left\{ \widehat{m}(x) \right\} & = \be\tr \,\left( \bPhi\tr \bK \bPhi  \right)^{-1} \bPhi\tr \bK \left[m(X_1), \dots, m(X_n)\right]\tr\\ 
 & = m(x) \, + \, T_1 \, + \, \frac{1}{2} \, T_2
\end{aligned}
\end{equation}
where $\ds T_1 \coloneqq  \be\tr \,\left( \bPhi\tr \bK \bPhi  \right)^{-1} \bPhi\tr \bK \left[\langle \cP_{\cS_J^\perp} m_x', X_1-x \rangle, \dots, \langle \cP_{\cS_J^\perp} m_x', X_n-x \rangle  \right]\tr$ and $\ds T_2 \coloneqq  \be\tr \,\left( \bPhi\tr \bK \bPhi  \right)^{-1} \bPhi\tr \bK \bR_{\zeta, x}\tr$. The next lemma focuses on the asymptotic behavior of $T_1$ and $T_2$. Its proof can be found in Appendix~\ref{app:A}.

\begin{lem}\label{lem:bias} As soon as conditions  (\prettyref{hypo:taylor})--(\prettyref{hypo:sbp1}) are fulfilled,
	\begin{enumerate}[label=(\roman*)]
	\item $T_1 \, = \, O_P\left( \|  \cP_{\cS_J^\perp} m_x'  \| \, h \right)$,
	\item $T_2 \, = \, O_P\left( h^2 \right)$.
	\end{enumerate}
\end{lem}
Now, it is enough to plug in {\sc \prettyref{lem:bias}} with \prettyref{eq:bias} to get the claimed conditional bias.
\end{proof}

\begin{proof}[Proof of the conditional variance of $\widehat{m}(x)$] By the assumptions, the covariance matrix of $\bY$ given $X_1, \ldots,X_n$ is the diagonal matrix $\mbox{diag} \{ \sigma^2(X_1), \ldots, \sigma^2(X_n) \}$ and $\varX \{\widehat{m}(x) \}   =   \be\tr  \left( \bPhi\tr \bK \bPhi  \right)^{-1}  \bPhi\tr  \bD   \bPhi  \left( \bPhi\tr \bK \bPhi  \right)^{-1}  \be$. Here, $\bD$ is also a diagonal matrix such that $[\bD]_{ii} \, = \, \sigma^2(X_i) \, [\bK]_{ii}^2$. According to (\prettyref{hypo:kernel}), (\prettyref{hypo:asymptotics}) and (\prettyref{hypo:condvar}), $[\bD]_{ii} \, = \, \{ \sigma^2(x) \, + \, o(1) \} \, [\bK]_{ii}^2$ and $\varX \{\widehat{m}(x) \} \,  = \, \{ \sigma^2(x) \, + \, o(1) \} \,  \be\tr \, \left( \bPhi\tr \bK \bPhi  \right)^{-1} \, \bPhi\tr \, \bK^2 \,  \bPhi$ $\left( \bPhi\tr \bK \bPhi  \right)^{-1} \be.$ The remainder of the proof is based on the technical {\sc Lemma~\ref{lem:variance}}, which is postponed to Appendix~\ref{app:A}. 
\end{proof}
Once the theoretical properties of $\widehat{m}(x)$ have been given, a natural and interesting issue concerns the asymptotic behavior of the functional derivative $\widehat{m}_x'$ of $m$ at $x$. The next result details the conditional bias and variance of $\widehat{m}_x'$. 
\begin{thm} \label{thm:biasvarderiv}  As soon as  (\prettyref{hypo:taylor})--(\prettyref{hypo:condvar}) are fulfilled, conditionally to $X_1, \ldots, X_n$, 
\begin{eqnarray*}
\|  \widehat{m}_x'  - m_x' \| & = & O_P\left( \lambda_J^{-1} \| \cP_{\cS_J^\perp} m'_x \|  \right)  \, + \, O_P\left( \lambda_J^{-1} h \right)\\ & & + \, O_P\left( h^{-1} \{ \lambda_J^2 \, n \, \pi_x(h) \}^{-1/2} \right) \, + \, O_P\left(  h^{-1} \{\lambda_J \, n \, \pi_x(h) \}^{-1/2} \sqrt{J}\right).
\end{eqnarray*}
\end{thm}
The first two summands in the formula above correspond to the conditional bias of $\widehat{m}_x'$, while the remaining terms come from its conditional variance. As $n$ tends to infinity, $h$ tends to zero, the dimension $J$ goes to infinity, and the smallest eigenvalue $\lambda_J$ to zero. Therefore, the rate of convergence of $\widehat{m}_x' $ is slower than that of $\widehat{m}(x)$. Compared to {\sc \prettyref{thm:biasvar}}, one $h$ is removed in the conditional bias, and $h^{-1}$ is added to the terms that relate to conditional variance, which corresponds to the standard degradation of the convergence rate observed in the multivariate case. But the functional setting adds specific terms like $J$ and $\lambda_J$ which deteriorate the asymptotic behavior. Nevertheless, as pointed out below, the finite sample properties of $\widehat{m}_x' $ are surprisingly good. 

\begin{proof}[Proof of {\sc \prettyref{thm:biasvarderiv}}]  The asymptotic behavior of $\|  \widehat{m}_x'  - m_x' \|$ is a direct by-product of the next results which detail the conditional bias and variance of the functional derivative $\widehat{m}_x'$.
\begin{lem} \label{lem:biasvarderiv}  As soon as  (\prettyref{hypo:taylor})--(\prettyref{hypo:condvar}) are fulfilled, 
	\begin{enumerate}[label=(\roman*)]
	\item $\| \EX \left( \widehat{m}_x'  \right) - m_x' \| \, = \, O_P\left( \lambda_J^{-1} \| \cP_{\cS_J^\perp} m'_x \|  \right)  \, + \, O_P\left( \lambda_J^{-1} h \right),$ 
	\item and conditionally on $X_1, \ldots, X_n$,
$$\| \widehat{m}_x' - \EX \left( \widehat{m}_x' \right)  \|  \, = \, O_P\left( h^{-1} \{ \lambda_J^2 n \pi_x(h) \}^{-1/2} \right) \, + \, O_P\left(  h^{-1} \{\lambda_J n \pi_x(h) \}^{-1/2} \sqrt{J}\right).$$
	\end{enumerate}
\end{lem}
\begin{proof}[Proof of {\sc \prettyref{lem:biasvarderiv}}-$(i)$] Let us first focus on the conditional bias of $\widehat{m}_x'$. By the definition of $\widehat{m}_x'$, $\EX\left( \widehat{m}_x'  \right) \, = \,   \bphi\tr \, [\bzero | \bI]  \,\left( \bPhi\tr \bK \bPhi  \right)^{-1} \bPhi\tr \bK \left[m(X_1), \dots, m(X_n)  \right]\tr$ so that, by \eqref{eq:taylor} and \eqref{eq:expansion},
\begin{equation}
\label{eq:biasderiv}
\EX\left( \widehat{m}_x'  \right) - m_x'\, = \, - \cP_{\cS_J^\perp} m_x' \, + \, Q_1 \, + \, \frac{1}{2} \, Q_2,
\end{equation}
for $\ds Q_1 =  \bphi\tr \, [\bzero | \bI] \,\left( \bPhi\tr \bK \bPhi  \right)^{-1} \bPhi\tr \bK \left[\langle \cP_{\cS_J^\perp} m_x', X_1-x \rangle, \dots, \langle \cP_{\cS_J^\perp} m_x', X_n-x \rangle  \right]\tr$, and $\ds Q_2 =  \bphi\tr \, [\bzero | \bI] \,\left( \bPhi\tr \bK \bPhi  \right)^{-1} \bPhi\tr \bK \bR_{\zeta, x}\tr$. {\sc \prettyref{lem:biasderiv}} given in Appendix~\ref{app:A} together with \prettyref{eq:biasderiv} provide the claimed expression for the conditional bias.
\end{proof}

\begin{proof}[Proof of {\sc \prettyref{lem:biasvarderiv}}-$(ii)$] Following the proof of the conditional variance of $\widehat{m}(x)$
	\[
	\begin{aligned}
	\EX \left( \| \widehat{m}_x' - \EX\left( \widehat{m}_x' \right)  \|^2 \right) = \, \int_0^1 \varX \left\{ \widehat{m}_x'(t) \right\} \dd t = \, \sigma^2(x) \, Q \,  \{ 1 \, + \, o(1) \},
	\end{aligned}
	\]
where $\ds Q = \int_0^1 \,  \bphi(t)\tr \, [\bzero | \bI] \left( \bPhi\tr \bK \bPhi  \right)^{-1} \, \bPhi\tr \, \bK^2 \,  \bPhi \, \left( \bPhi\tr \bK \bPhi  \right)^{-1} [\bzero | \bI]\tr \bphi(t) \dd t$. The remainder of the proof consists in decomposing this conditional variance according to formula \prettyref{eq:inverse} given in Appendix~\ref{app:A} with by-products of {\sc \prettyref{lem:variance}}. All details are postponed to the appendix. 
\end{proof}

The statement of {\sc \prettyref{thm:biasvarderiv}} now follows immediately.
\end{proof}
\section{Examples of approximating bases}	\label{sec:interpolation}
The functional local linear estimator depends on the basis $\phi_1, \phi_2,\ldots$ In this section we specify the asymptotic behavior of our estimator when considering particular bases. We start with usual approximating function spaces (cases~1 and~2 below). A more challenging issue consists in replacing the deterministic basis $\phi_1, \phi_2,\ldots$ with a data driven one. This important question is investigated in case~3.

\bigskip

{\sc Case~1: Orthogonal B-spline basis}. B-spline basis is a well-known and useful tool for approximating smooth functions. Consider the set of functions that are polynomials of degree $q$ on each interval $[(t-1)/k,\, t/k]$ for $t=1,\ldots,k$ and are $(q-1)$ times continuously differentiable on $[0,\,1]$. This $(k+q)$-dimensional subspace of $H$ defines the well-known space of splines. One can derive an orthogonal basis of B-splines of that space $\{B_{k,1}, \ldots,  B_{k,k+q}  \}$ (see \cite{Boor78} for an overview on spline functions and \cite{Redd12} for the orthogonalization of B-spline basis functions). In this situation, one sets $J=k+q$ and for any $j\leq J$, $\phi_j = B_{k,j} / \| B_{k,j} \|$. In this setting $k=k_n$ is a sequence that grows to infinity with $n$. Now, let the functional $m_x'$ be smooth enough so that its $p$th derivative is a H\"older function:
\begin{itemize}
\item[\refstepcounter{hypothesis} \label{hypo:mod1} (H\thehypothesis)] $| m_x'^{(p)}(u) - m_x'^{(p)}(v) | \leq C | u - v |^\nu$ with $\nu \in [0,1]$.
\end{itemize}
From Theorem XII.1 in \cite{Boor78} and (\prettyref{hypo:mod1}),  $ \| \cP_{\cS_J^\perp} m_x'\| \leq C\, (J-q)^{-(p+\nu)}$. Because $\| P_{\cS_J} m_x' - m_x'\| = \| P_{\cS_J^\perp} m_x'\|$, the rate of convergence of the conditional bias becomes
\begin{equation}\label{eq:bias-spline}
\EX\left\{  \widehat{m}(x) \right\}  \, = \, m(x) \, + \,  O_P(h \, J^{-p-\nu}) \, + \,  O_P(h^2).
\end{equation}

\bigskip

{\sc Case~2: Fourier basis}. When one suspects some periodic features for the functional $m_x'$, it can be advantageous to expand the functions by means of the Fourier basis $\phi_1(t) = 1$, $\phi_{2j}(t) = \sqrt{2} \sin(2\pi j t)$, and $\phi_{2j+1}(t) = \sqrt{2} \cos(2\pi j t)$ for $j=1,2,\ldots$ If one assumes that 
\begin{itemize}
\item[\refstepcounter{hypothesis} \label{hypo:mod2} (H\thehypothesis)] $m_x'$ is a periodic function,
\end{itemize}
then, according to \cite{Zygm02} and assumptions (\prettyref{hypo:mod1}) and (\prettyref{hypo:mod2}), one has $ \| \cP_{\cS_J^\perp} m_x'\| \, = \, O\left(J^{-p-\nu}\right)$ which leads to the same rate of convergence (\ref{eq:bias-spline}).

\bigskip

{\sc Case~3: Data driven basis}. The functional principal components analysis (FPCA) allows to expand a random function $X$ into a basis $\phi_1, \phi_2,\ldots$ in an optimal way (see \cite{Karh46, Loev46, Rao58, DauxPR82} for precursor works and \cite{Bosq00, YaoMW05, HallH06, HallMW06} for more recent statistical developments). In this setting, functions $\phi_j$ are the eigenfunctions of the covariance operator of $X$ and the eigenanalysis of the empirical covariance operator provides a data driven basis $\widehat{\phi}_1, \widehat{\phi}_2, \ldots$ (by convention, we assume that $\langle \phi_j, \widehat{\phi}_j \rangle >0$). We propose to investigate the asymptotic properties of the functional local linear estimator when replacing $\phi_1, \ldots, \phi_J$ with the data driven basis $\widehat{\phi}_1, \ldots, \widehat{\phi}_J$. This results in the new estimator $\widehat{\widehat{m}}(x) \, = \, \be\tr \,\left( \widehat{\bPhi} \tr \bK \widehat{\bPhi}  \right)^{-1} \widehat{\bPhi} \tr \bK \bY$ where, for $i=1,\ldots,n$, $[\widehat{\bPhi}]_{i1} = 1$ and for $j=2,\ldots,J$, $[\widehat{\bPhi}]_{ij} = \langle \widehat{\phi}_j, \, X_i-x \rangle$. Thanks to \cite{Bosq00} and \cite{CardFS99}, as soon as $\E \|X\|^4 < \infty$, one has for any $j=1,2,\ldots$ that $\| \widehat{\phi}_j - \phi_j \| = O_P(a_J^{-1} n^{-1/2})$ with $a_J \coloneqq \min_{j\leq J} \left\{\rho_j-\rho_{j+1}, \rho_{j-1}-\rho_j \right\}$. Here, $\rho_j$ are the eigenvalues of the covariance operator of $X$ (placed in descending order). 
\begin{thm} \label{thm:ddb} If $\E \|X\|^4 < \infty$, $a_J^{-1} n^{-1/2} = o(1)$ and (\prettyref{hypo:taylor})--(\prettyref{hypo:condvar}) hold, 
	\begin{enumerate}[label=(\roman*)]
	\item $\EX \left\{ \widehat{\widehat{m}}(x) \right\} =  m(x) +  O_P\left( J^{1/2} \,  \| \cP_{\cS_J^\perp} m_x'\| \, h  \right)  +  O_P\left( J^{1/2} \, h^2 \right)$\\ \hspace*{2.7cm}$+ \, O_P\left(a_J^{-1} \, n^{-1/2} \, J^{1/2} \,  h \right)$,
	\item $\varX \left\{ \widehat{\widehat{m}}(x) \right\} \, = \, O_P\left( \{n \, \pi_x(h)  \}^{-1} \right)$.
	\end{enumerate}
\end{thm}
Consideration of the data driven basis degrades slightly the conditional bias by introducing in both original terms the quantity $J^{1/2}$ and by adding a third term $O_P\left(a_J^{-1} \, n^{-1/2} \, J^{1/2} \,  h \right)$. However, the conditional variance is not sensitive to the introduction of the data driven basis.

\section{Implementation}
In this section we discuss the practical aspects, and assess the finite sample performance of our local linear estimators for functional data. Beyond standard issues such as the choice of the tuning parameters in the estimation of the regression operator, a novel heuristic is developed for selection of the bandwidth for the functional derivative. In a comparative simulation study, finite sample performance of our estimator is compared with its competitors available in the literature. We conclude this section with a real data example and an application of the functional local linear estimator. On a benchmark growth dataset we demonstrate a strong link between the considered scalar-on-function local linear regression, and the important single-functional index model, widely considered in the literature.

\subsection{Selection of tuning parameters \label{subsec:ParameterSelection}}
{\sc Methodology}. According to the definition of $\widehat{m}(x)$ and $\widehat{m'_x}$, two parameters have to be selected: the dimension $J$ of the approximating subspace and the bandwidth $h$. Theoretical results of {\sc Theorems} \ref{thm:biasvar} and \ref{thm:biasvarderiv} emphasize different asymptotic behaviors for the estimators of regression and functional derivative. This is why the optimal parameters for estimating the regression operator do not match necessarily those designed for the functional derivative. 

Choosing optimal  parameters for the functional derivative is quite challenging because no quantity is directly available to compare with. The following ad hoc methodology is proposed. Firstly, optimal parameters $h_{reg}$ and $J_{reg}$ for the estimator of the regression operator $\widehat{m}$ are selected. The \texttt{R} implementation of the estimating procedure that can be found in the online supplementary material allows for the use of two standard criteria --- (leave-one-out) cross-validation \cite{Alle74, Ston74, Ston77} and an adaptation of the corrected Akaike information criterion \cite{Akai73, HurvST98} to functional data; a general overview of these criteria can be found in \cite{FrieHT01}. Secondly, the parameters $h_{reg}$ and $J_{reg}$ are used to build, for each random function $X_i$ in the sample, a pilot estimator $\widehat{m_{X_i}'}^{boot}$ of the functional derivative at $X_i$ by means of a wild bootstrap procedure \cite{Wu86, Mamm93, FerrKV10}. Then, the optimal parameters $h_{deriv}$ and $J_{deriv}$ are those minimizing the mean squared error 
	\begin{equation}	\label{bootstrap}
	n^{-1} \sum_{i=1}^n \left\Vert \widehat{m_{X_i}'}^{boot} - \widehat{m_{{X_i},-i}'} \right\Vert^2,
	\end{equation}
where $\widehat{m'_{{X_i},-i}}$ is the estimator of the functional derivative at $X_i$ from a dataset with the $i$th observation removed.
An important methodological point consists in translating the bandwidths into the number of nearest neighbors. Given a function $x$ in $H$ and an integer $k$, a $k$-nearest neighbors bandwidth at $x$ corresponds to the smallest bandwidth $h$ such that the number of functional regressors $X_i$ belonging to ball of radius $h$ centered at $x$ is equal to $k$. An advantage of such a local approach is that it provides more flexible estimators while reducing a continuous set of candidates to a discrete one. From now on, all bandwidths in our results are expressed in terms of nearest neighbors. 

\bigskip

{\sc Focus on bandwidth choice.} In order to assess the quality of the bandwidth choice for both estimators, a first model (M1) is simulated so that the regression operator as well as the functional derivatives can be expressed analytically. 

\medskip 

{\em Simulated model} (M1). Let $X_1, \dots, X_n, X_{n+1}, \dots, X_{n+500}$ be iid copies of a functional predictor $X$. The estimators are based on the training set $X_1,\ldots, X_n$; the 500 remaining functions $X_{n+1},\ldots,X_{n+500}$ are used to assess the quality of the estimators. The model takes the form $Y \coloneqq m(X) + \varepsilon$, where $\varepsilon$ is an independent, centered and normally distributed error term with variance $\sigma_{\varepsilon}^2$. Let $\phi_1, \ldots, \phi_4$ be the first four elements of the Fourier basis; the random function $X$ is equal to the linear combination $\sum_{j=1}^4 U_j \, \phi_j$ where $U_j$ are iid uniform random variables on $[-1,1]$. The regression operator is given by $m(X) \coloneqq \sum_{j=1}^4 \exp(-U_j^2)$. From the expression for $m$, one can derive its functional derivative at $x$ that takes the form $m_x'(t) = -2 \sum_{j=1}^4 U_j \, \exp(-U_j^2) \, \phi_j(t)$. Here, the size of the approximating subspace is $J=4$. We consider $n=100, 150, 200,\ldots,500$. At last, the noise-to-signal ratio ($nsr$) is also controlled by setting $\sigma_{\varepsilon}^2 \coloneqq nsr \times \mbox{Var}\{m(X)\}$, with $nsr=0.05,0.2,0.4$. Because we focus on the bandwidth selection, in the present model in the estimation procedure we deliberately use the true approximating subspace (i.e. $J=4$ and the first four Fourier basis elements $\phi_1,\ldots, \phi_4$). 

\medskip

{\em Regression operator estimation}. The first step is to assess the quality of the bandwidth selection for estimating the regression operator. To this end, the optimal bandwidth $h_{reg}$ is the one minimizing the CV (cross-validation) or AIC$_C$ (corrected Akaike information criterion). Because in (M1) the true regression operator $m$ is known, one can compute the oracle relative mean squared error of prediction
	\begin{equation}	\label{ORMSEP reg}
	ORMSEP_{reg} \coloneqq \frac{1/500 \sum_{i=n+1}^{n+500} \left\{m(X_i) - \widehat{m}(X_i) \right\}^2}{1/500 \sum_{i=n+1}^{n+500} \left\{ m(X_i) - \overline{m(X)} \right\}^2 }	
	\end{equation}
where $\overline{m(X)} \coloneqq 1/500 \sum_{i=n+1}^{n+500} m(X_i)$. Repeating the simulation scheme 100 times in various situations, Figure~\ref{fig:reg_rmsep} assesses the quality of the local linear estimator of the regression operator. The consistency of the estimator appears clearly even for a large noise-to-signal ratio. Concerning the bandwidth selection, both methods CV and AIC$_C$ provide similar results, although CV seems to outperform slightly AIC$_C$. 

\begin{figure}[htpb] 
 \centerline{\includegraphics[scale = 0.5]{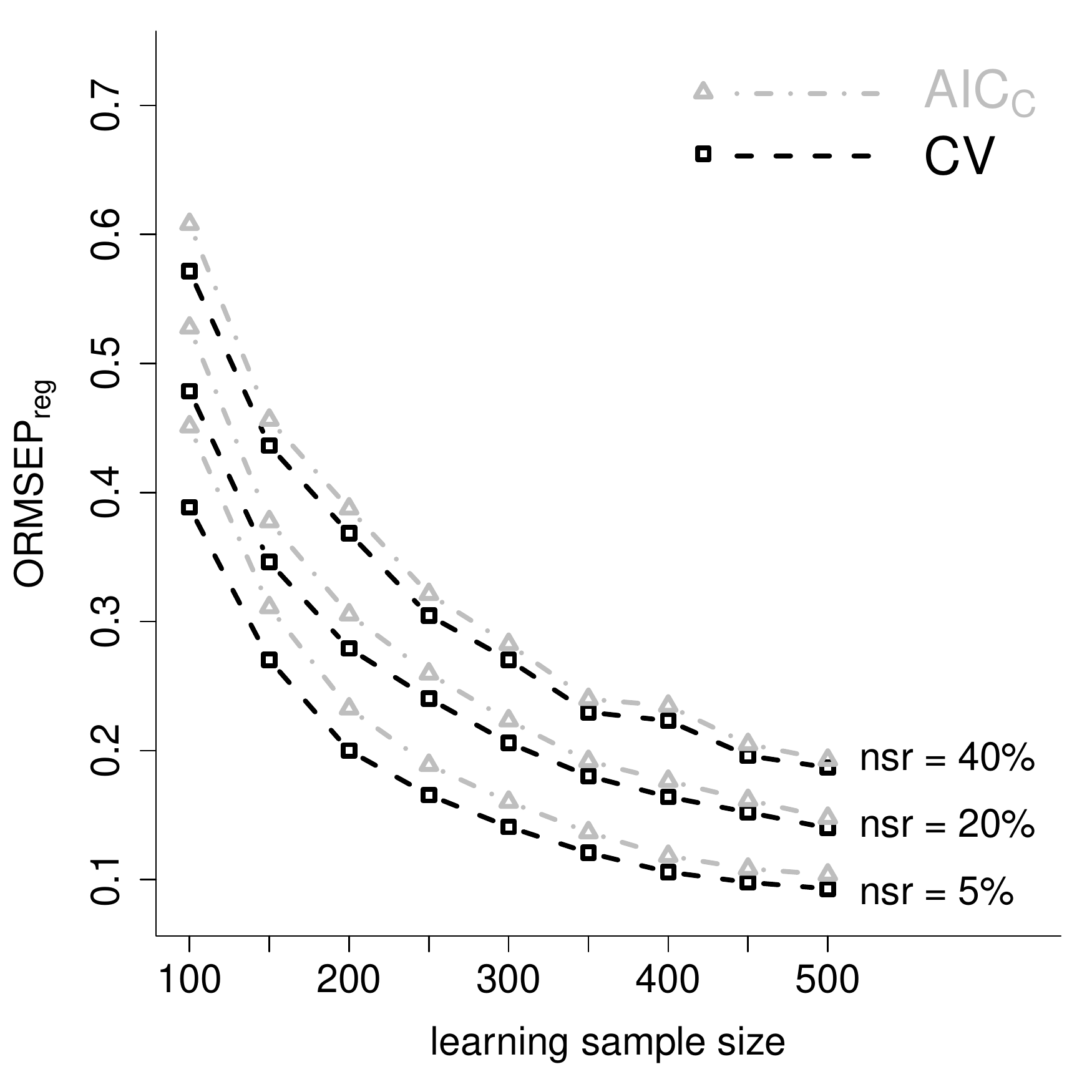}}
\caption{Mean (squares or triangles) of $ORMSEP_{reg}$ (each time over 100 runs) according to different noise-to-signal ratios ($nsr$), learning sample sizes and bandwidth selection methods (CV/AIC$_C$).} \label{fig:reg_rmsep}
\end{figure}
Figure~\ref{fig:reg_scatterplot} displays the true values $m(X_i)$ against their predictions $\widehat{m}(X_i)$ when cross-validation is used. Even in the worst case (small learning sample size and high noise-to-signal ratio), the local linear estimator provides reliable results.
\begin{figure}[htpb] 
\begin{tabular}{cc}
(a) worst case & (b) most favorable case\\
\includegraphics[width = 0.45\linewidth]{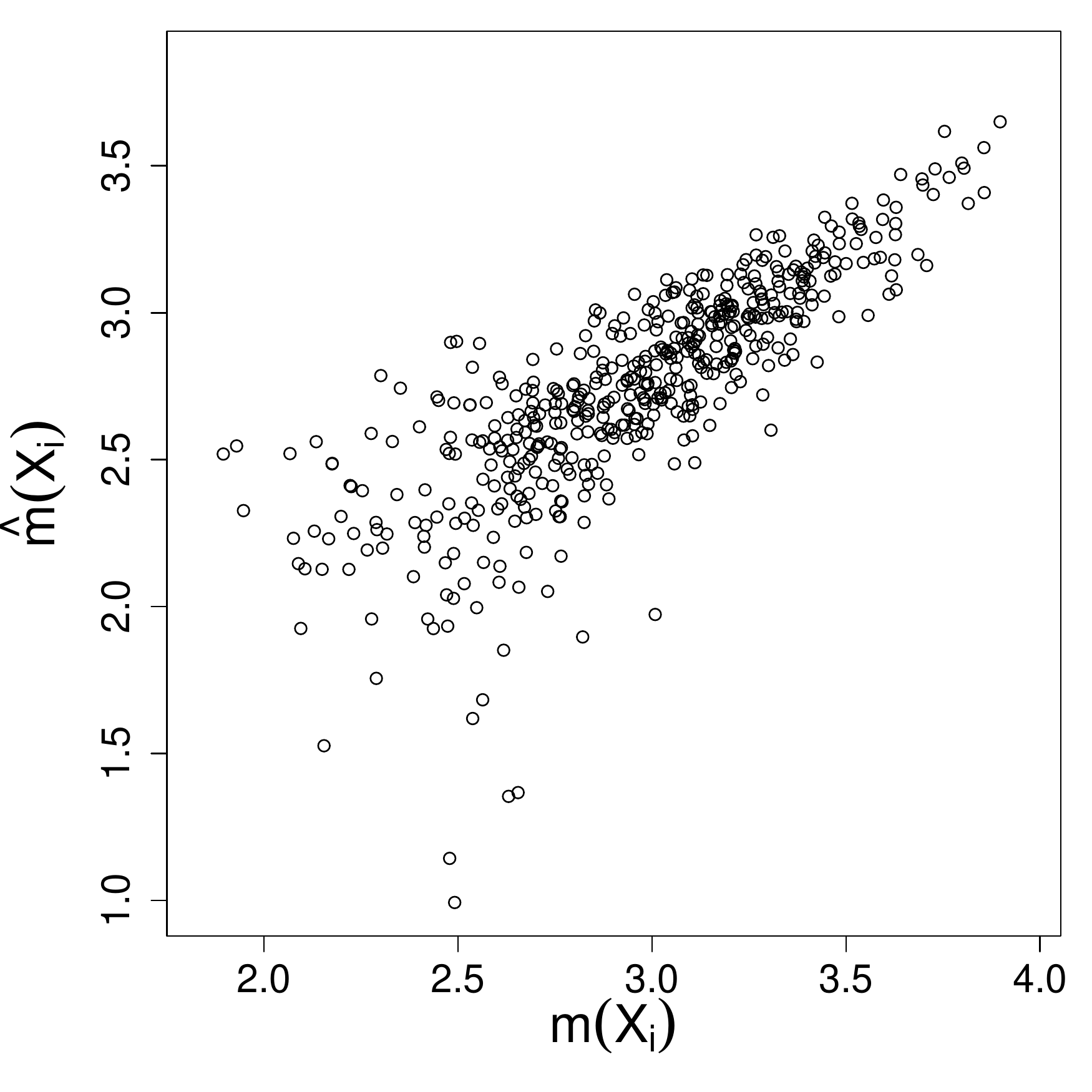}
& 
\includegraphics[width = 0.45\linewidth]{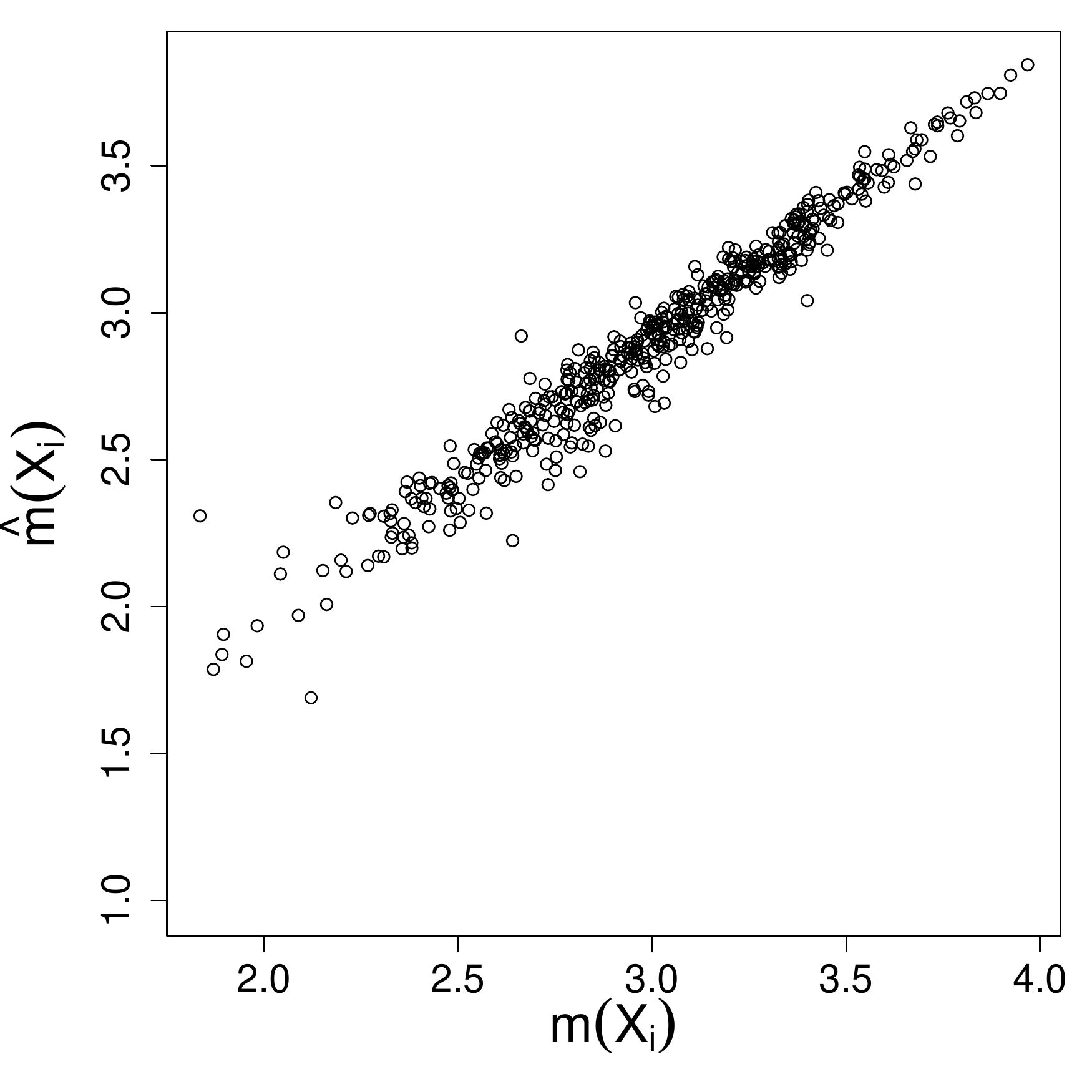}
\end{tabular}
\caption{(a) $n=100$ and $nsr=0.4$, (b) $n=500$ and $nsr=0.05$.}
\label{fig:reg_scatterplot}
\end{figure}

\medskip

{\em Functional derivative estimation}. The selection of the bandwidth for estimating the functional derivative is much more challenging because there is no standard criterion to minimize. An original bandwidth selection based on the wild bootstrap procedure is proposed. It aims to build a pilot estimator of the functional derivative:
	\begin{enumerate}[label=(\roman*), ref=(\roman*)]
	\item \label{step1} use $h_{reg}$ for estimating the model error $\widehat{\varepsilon}_i \coloneqq Y_i - \widehat{m}(X_i)$ for $i=1,\ldots,n$,
	\item given iid centered random variables $V_1,\ldots, V_n$ independent of $\widehat{\varepsilon}_i$ such that their first moments equal $1$, compute the bootstrapped errors $\varepsilon_i^{(b)} \coloneqq \widehat{\varepsilon}_i  \times V_i$,
	\item \label{step3} derive a bootstrapped sample $\mathcal{S}^{(b)} \coloneqq \big(X_i,\ Y_i^{(b)} \coloneqq  \widehat{m}(X_i) + \varepsilon_i^{(b)}\big)_{i = 1, \ldots, n}$ and compute the bootstrapped estimator $\widehat{m'_{X_i}}^{(b)}$ from $\mathcal{S}^{(b)}$.
	\end{enumerate}
Repeat steps \ref{step1}--\ref{step3} independently $B$ times and denote $\widehat{m'_{X_i}}^{boot} \coloneqq 1/B \sum_{b=1}^B \widehat{m'_{X_i}}^{(b)}$ what we name the pilot estimator of the functional derivative at $X_i$. The optimal bandwidth $h_{deriv}$ is defined as the one minimizing \eqref{bootstrap}. In order to assess the relevance of this bandwidth choice, a simulation study with model (M1) is conducted with $B=100$. Similarly as in the study of the estimator of the regression operator, Figure~\ref{fig:deriv_rmsep}(a) displays the means of the oracle relative mean squared error
	\begin{equation}	\label{ORMSEP}
	ORMSEP_{deriv} \coloneqq \frac{1/500 \sum_{i=n+1}^{n+500} \left\Vert m'_{X_i} -  \widehat{m'_{X_i}} \right\Vert^2}{1/500 \sum_{i=n+1}^{n+500} \left\Vert m'_{X_i} -  \overline{m'_{X}} \right\Vert^2}	
	\end{equation}
where $ \overline{m'_{X}} \coloneqq 1/500 \sum_{i=n+1}^{n+500} m'_{X_i}$. Firstly, selecting $h_{reg}$ with CV or AIC$_C$ has no significant impact on the prediction quality of the estimator of functional derivatives. Secondly, this plot demonstrates the consistency of the local linear estimator of the functional derivative. Its rate of convergence seems to be slightly slower than the one observed for the local linear estimator of the regression operator, as supported by the asymptotic results. Figure~\ref{fig:deriv_rmsep}(b) reproduces the same plot but $\widehat{m_{X_i}'}$ is built with $h_{reg}$ (the bandwidth used for estimating the regression operator with cross-validation) instead of $h_{deriv}$ (the specific bandwidth computed for estimating the functional derivative). In this situation, it is the AIC$_C$ method used for selecting $h_{reg}$ which provides better predictions. Nevertheless, consistency is far less obvious than in part (a). Figures~\ref{fig:deriv_rmsep}(a) and~\ref{fig:deriv_rmsep}(b) illustrate the importance of selecting a specific bandwidth for the estimation of functional derivatives.
\begin{figure}[htpb] 
\begin{tabular}{cc}
(a) bandwidth $h_{deriv}$ & (b) bandwidth $h_{reg}$ \\ 
\includegraphics[scale = 0.4]{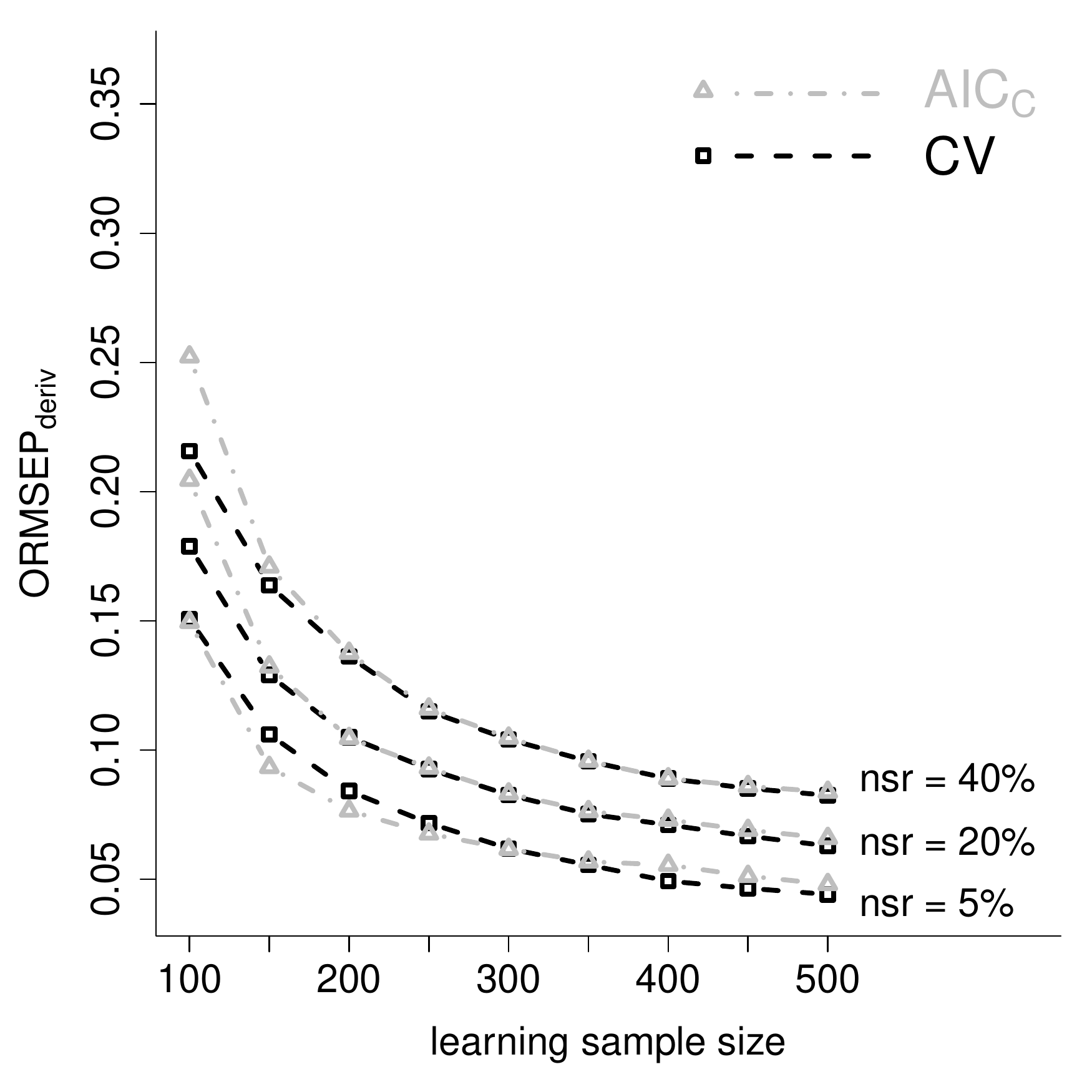}
& 
\includegraphics[scale = 0.4]{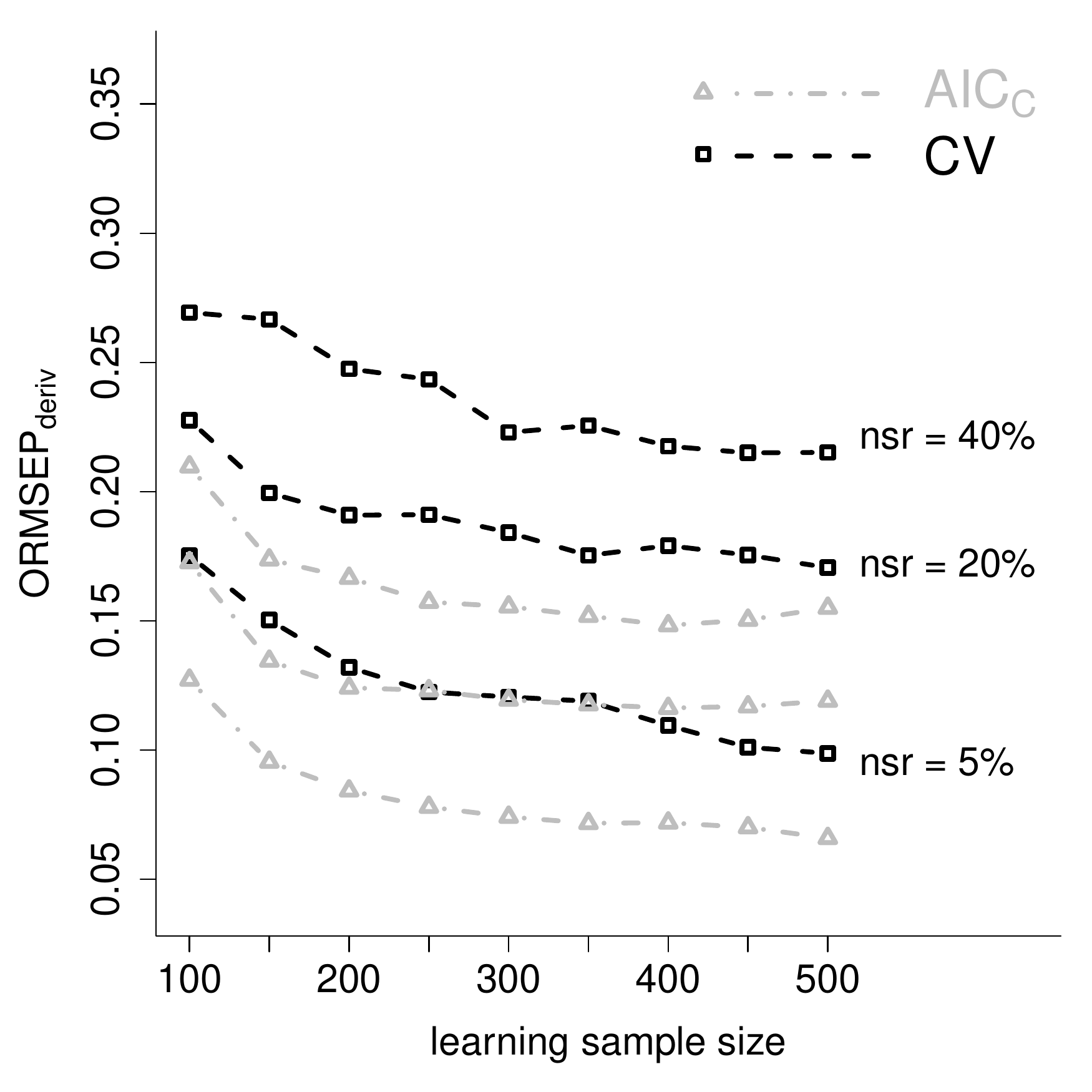}
\end{tabular}
\caption{Mean (squares or triangles) of $ORMSEP_{deriv}$ (each time over 100 runs) according to different noise-to-signal ratios ($nsr$), learning sample sizes and bandwidth selection methods (CV or AIC$_C$) used for computing $\widehat{m}$; (a) $\widehat{m_{X_i}'}$ is built with $h_{deriv}$, (b) $\widehat{m_{X_i}'}$ is built with $h_{reg}$. Smaller values of $nsr$ correspond to lower $ORMSEP_{deriv}$.}
\label{fig:deriv_rmsep}
\end{figure}
Figure~\ref{fig:deriv_comparison} compares a sample of true functional derivatives $m_{X_1}', \ldots, m_{X_n}'$ with corresponding predictions $\widehat{m_{X_1}'}, \ldots, \widehat{m_{X_n}'}$ (a) in the worst case ($n=100$ and $nsr = 0.4$) and (b) in the most favorable situation ($n=500$ and $nsr = 0.05$). Even in the worst situation, the predictions remain adequate. 
\begin{figure}[htpb] 
\begin{tabular}{cc}
(a) worst case & (b) most favorable case \\
\includegraphics[scale = 0.4]{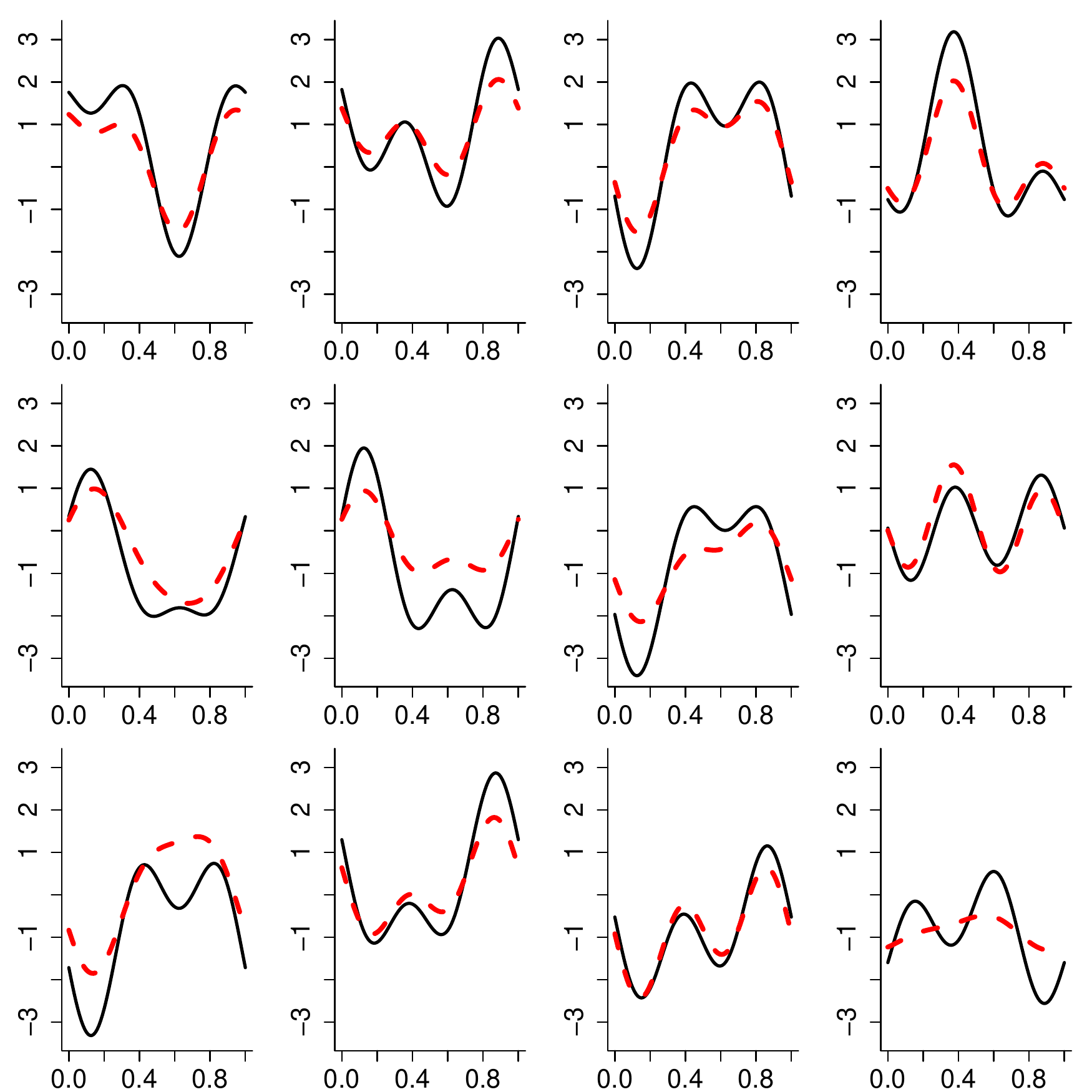}
&
\includegraphics[scale = 0.4]{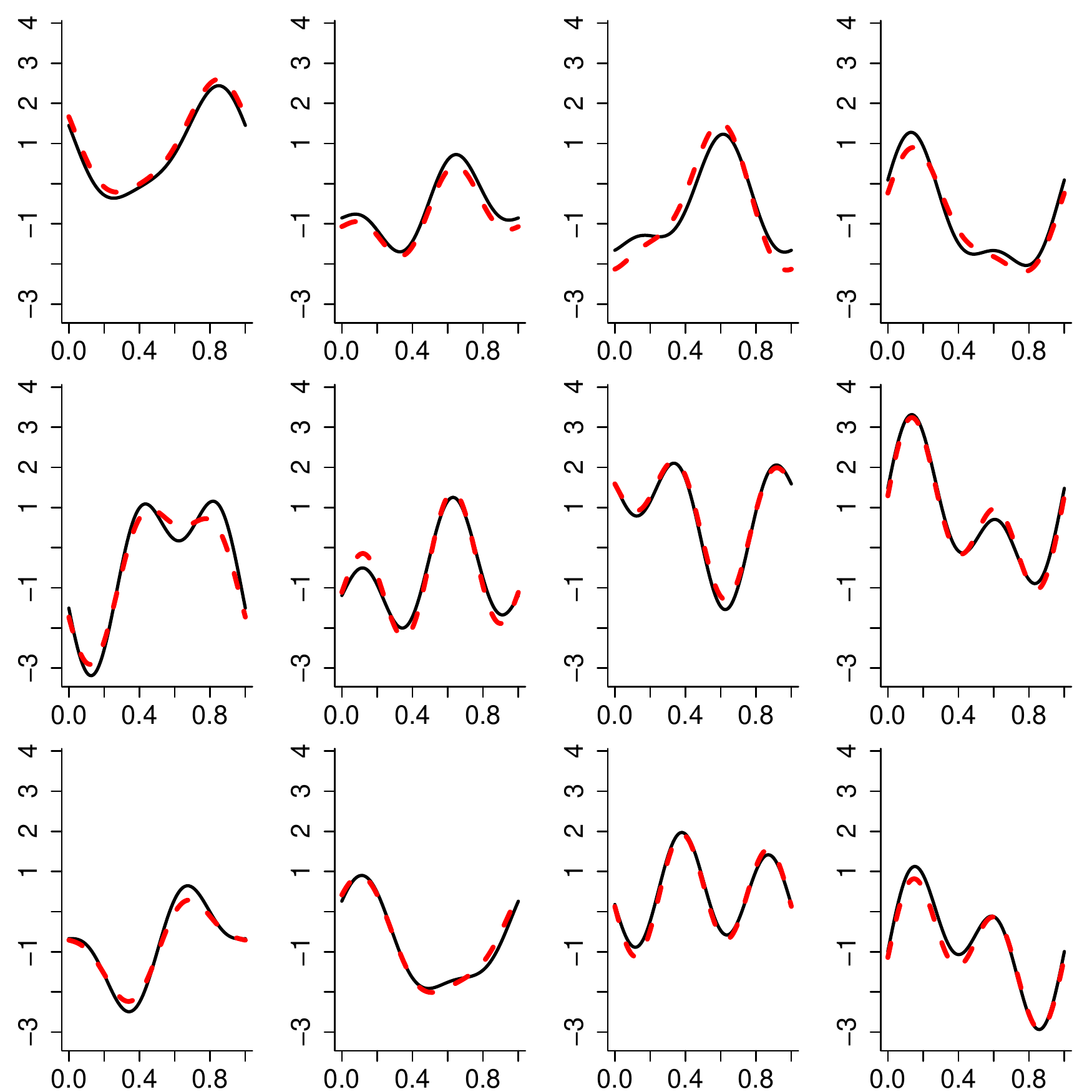}
\end{tabular}
\caption{Functional derivatives $m_{X_i}'$ (solid lines) and their predictions $\widehat{m_{X_i}'}$ with $h_{deriv}$ (dashed lines).}
\label{fig:deriv_comparison}
\end{figure}

Another way to assess the bandwidth choice for estimating the functional derivative is to compare the selected bandwidth $h_{deriv}$ itself with the oracle one $h_{deriv}^{oracle}$ minimizing the oracle relative mean squared error $ORMSE_{deriv}$ from \eqref{ORMSEP} computed from the random sample functions $X_1, \dots, X_n$. Given a learning sample size $n$ and a noise-to-signal ratio ($nsr$), 100 simulated datasets are drawn from (M1). Our local linear estimating procedure provides 100 triples of bandwidths $h_{reg}$, $h_{deriv}$ and $h_{deriv}^{oracle}$. Setting $n=100, 150, 200,\ldots,500$ and $nsr=0.05,0.2,0.4$, Figure~\ref{fig:der_bandwidth_choice} displays simultaneously $h_{deriv}^{oracle}$ versus $h_{deriv}$ and $h_{deriv}^{oracle}$ versus  $h_{reg}$. For each type of bandwidth, the mean (solid circle for $h_{deriv}$ and square for $h_{reg}$) and standard deviation (whiskers) over the 100 runs are displayed for each of the $9\times3$ pairs $(n, \, nsr)$. The bandwidth $h_{deriv}$ works quite well even if it underestimates slightly the oracle version $h_{deriv}^{oracle}$. On the other hand, $h_{reg}$ fails drastically, especially for larger sample sizes $n$. Summary statistics of the raw data are provided in Table~\ref{tab:der_bandwidth_choice} in Appendix~\ref{app:C}.
\begin{figure}[htpb] 
 \centerline{\includegraphics[scale = 0.45]{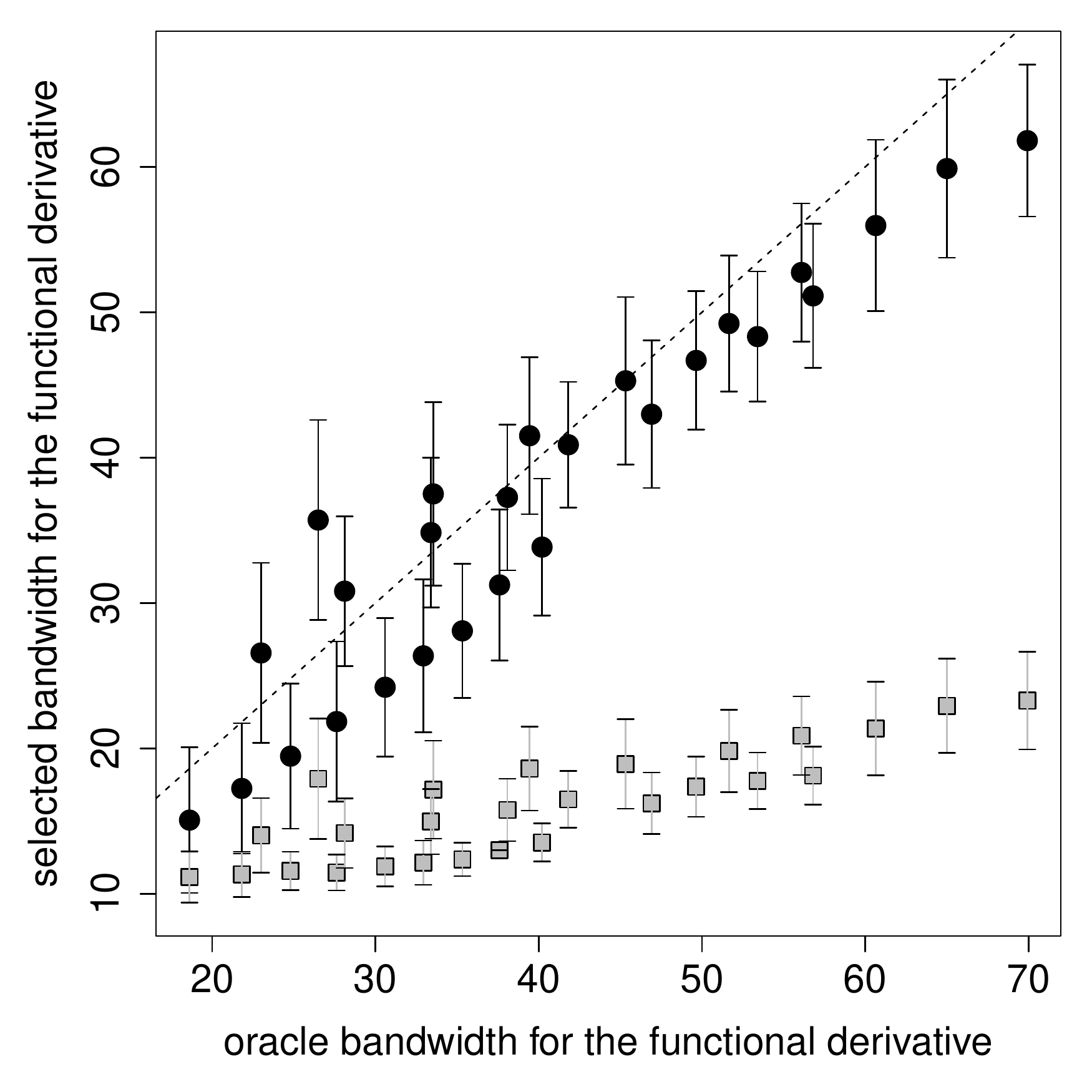}}
\caption{The oracle bandwidths $h_{deriv}^{oracle}$ versus $h_{deriv}$ (solid circles), and versus $h_{reg}$ (squares). On the horizontal axis, the averages of the oracle bandwidths $h_{deriv}^{oracle}$ are displayed.}
 \label{fig:der_bandwidth_choice}
\end{figure}
It is worth noting that the bootstrap bandwidth selection introduces additional randomness into the local linear estimation of the functional derivatives. However, Table~\ref{tab:stability} given in Appendix~\ref{app:C} indicates clearly that our procedure gives very stable results.

\medskip

{\em Conclusion.} To summarize this section devoted to bandwidth selection, we may conclude that: \begin{enumerate*}[label=(\roman*)] \item cross-validation is a useful method for determining $h_{reg}$, the bandwidth for estimating the local linear regression operator, and \item the bootstrap procedure which provides the bandwidth for estimating the functional derivative works well, even for small learning samples size and high noise-to-signal ratios. \end{enumerate*}

\bigskip

{\sc Automatic choice of the approximating subspace.} So far we focused only on the bandwidth selection. To make our method fully automatic, one has to determine also the approximating subspace spanned by $\phi_1,\ldots,\phi_J$ and its dimension $J$. As explained in Section~\ref{sec:interpolation} (case~3), functional principal component analysis is a very useful tool for expanding a random function onto the eigenfunctions of the covariance operator. Let $\phi_1,\ldots,\phi_J$ be the first $J$ eigenfunctions of the covariance operator of the functional predictor $X$ associated to the $J$ largest eigenvalues, and let $\widehat{\phi}_1,\ldots,\widehat{\phi}_J$ be their estimates from the empirical covariance operator. To make estimating procedure fully automatic, we may proceed in three steps: \begin{enumerate*}[label=(\roman*)] \item compute the first $J$ eigenfunctions $\widehat{\phi}_1,\ldots,\widehat{\phi}_J$ of the empirical covariance operator with $J$ large enough, \item carry out the local linear estimator of the regression operator with $h_{reg}$ and $J_{opt}$ obtained by minimizing the CV criterion with respect to $J_{opt} \in \left\{0, 1, \dots, J \right\}$ and $h_{reg}$, \item determine $h_{deriv}$ using the bootstrap procedure by minimizing \eqref{bootstrap} and compute the corresponding local linear estimator of the functional derivative. \end{enumerate*} 

\medskip

{\em Robustness of the selection procedure for $J_{opt}$}. To make the choice of the dimension $J$ more challenging, we add structural perturbation to the functional predictors $X$. Given the first eight Fourier basis elements $\phi_1,\ldots,\phi_4,\phi_5,\ldots,\phi_8$, set $X \coloneqq \sum_{j=1}^4 U_j \, \phi_j + \eta$ where $\eta \coloneqq \sum_{j=5}^8 V_j \, \phi_j$ with $U_j$ (resp. $V_j$) iid uniform random variables defined on $[-1,\, 1]$ (resp. $[-b,\, b]$). The second part $\eta$ provides a structural noise that is controlled by the ratio $\rho \coloneqq \E\left(\| \eta \|^2 \right) / \E\left(\| X \|^2 \right) = b^2 / (1 + b^2)$. We refer to this model as (M2). Given any $\rho \in (0,1)$, one can always find a corresponding bound $b$ for simulating the functional predictors. Table~\ref{tab:dimension_choice} shows how our estimating procedure is robust according to 4 structural noise ratios ($\rho=0.05, \,0.1, \,0.2, \,0.4$) and different learning sample sizes ($n=100, 150, \ldots, 500$). In each situation, model (M2) is simulated 100 times (with $nsr = 0.05$) resulting in 100 estimates $J_{opt}$. The larger $n$ is, the more often the dimension is correctly detected. The results degrade as the ratio $\rho$ increases. The detection is almost perfect for large sample sizes and ratios $\rho$ up to 0.2. Table~\ref{tab:ormsep_J4_nsr005} gives respectively the corresponding $ORMSEP_{reg}$ and $ORMSEP_{deriv}$, each time averaged over 100 runs with standard deviations in brackets. This additional table confirms that our estimating procedure is not too sensitive to perturbations. Complete results of this simulation study which include various choices of the true dimension $J$, $nsr$, and $\rho$, can be found in Appendix~\ref{app:C}.

\bigskip

{\sc Running time}. One could think that the introduction of the bootstrap procedure cumulates with the selection of two bandwidths and one dimension, and requires a computation that is quite intensive. Nevertheless, the running time of our \texttt{R} procedure is surprisingly short --- at most about 3 seconds (according to the previous simulation scheme with a processor Intel Core i7 2.7 GHz with 16 GB RAM) are necessary to carry out the estimation/prediction for both the regression operator and the functional derivative, including FPCA for the basis expansion $\widehat{\phi}_1, \ldots, \widehat{\phi}_{J_{opt}}$ with automatic computation of $J_{opt}$, and automatic bandwidths ($h_{reg}$ and $h_{deriv}$) selection, see the last column of Table~\ref{tab:dimension_choice}. 

\begin{table}[htpb]
\centering
\caption{Number of times, out of $100$, that the dimension is correctly selected, and running times (in s.).} 
\label{tab:dimension_choice}
\begin{tabular}{c|cccc|c}
$n$  & $\rho = 0.05$ & $\rho = 0.1$ & $\rho = 0.2$ & $\rho = 0.4$ & Timing \\ 
  \hline
100 & 27 & 20 & 21 & 11 & 0.60 \\ 
  150 & 46 & 48 & 30 & 11 & 0.69 \\ 
  200 & 56 & 59 & 46 & 13 & 0.51 \\ 
  250 & 79 & 67 & 54 & 12 & 0.89 \\ 
  300 & 85 & 86 & 56 & 20 & 1.03 \\ 
  350 & 86 & 87 & 67 & 20 & 1.62 \\ 
  400 & 96 & 93 & 71 & 23 & 1.92 \\ 
  450 & 94 & 90 & 83 & 18 & 2.39 \\ 
  500 & 98 & 97 & 86 & 29 & 3.29 \\ 
  \end{tabular}
\end{table}

\begin{table}[ht]
\centering
\begin{tabular}{cc|cccc}
                                                                                  & $n$            & \multicolumn{1}{l}{   $\rho = 0.05$} & \multicolumn{1}{l}{    $\rho = 0.1$} & \multicolumn{1}{l}{    $\rho = 0.2$} & \multicolumn{1}{l}{    $\rho = 0.4$} \\ 
   \hline
\parbox[t]{2mm}{\multirow{9}{*}{\rotatebox[origin=c]{90}{$ORMSEP_{reg}$}}}   & $100$ & 0.388 \,(0.058) & 0.391 \,(0.055) & 0.438 \,(0.060) & 0.593 \,(0.065) \\ 
                                                                                  & $150$ & 0.271 \,(0.042) & 0.288 \,(0.039) & 0.354 \,(0.056) & 0.518 \,(0.049) \\ 
                                                                                  & $200$ & 0.216 \,(0.032) & 0.231 \,(0.030) & 0.294 \,(0.033) & 0.459 \,(0.040) \\ 
                                                                                  & $250$ & 0.176 \,(0.026) & 0.195 \,(0.022) & 0.260 \,(0.029) & 0.431 \,(0.043) \\ 
                                                                                  & $300$ & 0.151 \,(0.018) & 0.172 \,(0.022) & 0.233 \,(0.022) & 0.397 \,(0.036) \\ 
                                                                                  & $350$ & 0.135 \,(0.022) & 0.151 \,(0.019) & 0.218 \,(0.025) & 0.369 \,(0.033) \\ 
                                                                                  & $400$ & 0.117 \,(0.014) & 0.137 \,(0.017) & 0.201 \,(0.022) & 0.350 \,(0.030) \\ 
                                                                                  & $450$ & 0.107 \,(0.014) & 0.125 \,(0.014) & 0.184 \,(0.018) & 0.333 \,(0.027) \\ 
                                                                                  & $500$ & 0.098 \,(0.011) & 0.116 \,(0.012) & 0.174 \,(0.018) & 0.323 \,(0.029) \\ 
   \hline
\parbox[t]{2mm}{\multirow{9}{*}{\rotatebox[origin=c]{90}{$ORMSEP_{deriv}$}}} & $100$ & 0.374 \,(0.181) & 0.399 \,(0.171) & 0.420 \,(0.162) & 0.604 \,(0.119) \\ 
                                                                                  & $150$ & 0.256 \,(0.155) & 0.256 \,(0.150) & 0.348 \,(0.166) & 0.528 \,(0.116) \\ 
                                                                                  & $200$ & 0.207 \,(0.145) & 0.202 \,(0.130) & 0.257 \,(0.126) & 0.476 \,(0.107) \\ 
                                                                                  & $250$ & 0.134 \,(0.128) & 0.170 \,(0.129) & 0.220 \,(0.119) & 0.446 \,(0.095) \\ 
                                                                                  & $300$ & 0.104 \,(0.095) & 0.113 \,(0.093) & 0.205 \,(0.115) & 0.404 \,(0.112) \\ 
                                                                                  & $350$ & 0.104 \,(0.126) & 0.103 \,(0.083) & 0.173 \,(0.107) & 0.378 \,(0.101) \\ 
                                                                                  & $400$ & 0.064 \,(0.053) & 0.086 \,(0.097) & 0.159 \,(0.104) & 0.352 \,(0.095) \\ 
                                                                                  & $450$ & 0.070 \,(0.096) & 0.084 \,(0.077) & 0.124 \,(0.088) & 0.353 \,(0.095) \\ 
                                                                                  & $500$ & 0.052 \,(0.038) & 0.061 \,(0.041) & 0.116 \,(0.083) & 0.326 \,(0.111) \\ 
  \end{tabular}
\caption{Average and standard deviation (in brackets) of $ORMSEP$ with $J=4$ and $nsr=0.05$.} 
\label{tab:ormsep_J4_nsr005}
\end{table}


\subsection{A comparative study} \label{subsec:ComparativeStudy}
We now conduct a simulation study, in which the finite sample performance of the local linear estimator is compared to its competitors from the literature. In the simulated datasets, we extend models (M1) and (M2) to consider a whole spectrum of scenarios, from a linear to a nonlinear additive one. 

\bigskip

{\sc Simulated model (M3).} A perturbed functional predictor $X \coloneqq \sum_{j=1}^4 U_j \, \phi_j + \eta$ is built according to the scheme given in (M2). We now consider $Y \coloneqq m_a(X)  + \varepsilon$ where $m_a(X) \coloneqq (1 - a) \langle \beta, X \rangle + a \sum_{j=1}^4 \exp(-U_j^2)$ and $\beta \coloneqq \sum_{j=1}^4 \phi_j$. Note that the choice $a = 0$ corresponds to a standard functional linear model, whereas $a = 1$ represents the nonlinear regression model (M2). For all choices of $a$, the model (M3) is additive \cite{MullY08, MullY10}, which means that the regression operator can be expressed as a sum of components where each component is a function that depends only on a single principal score of the regressor $X$. Additivity allows direct computation of the functional derivative of $m_a$. The derivative takes the form $m_{a,x}'(t) = (1 - a) \beta(t) - 2 a \sum_{j=1}^4 U_j \exp(-U_j^2) \, \phi_j(t)$.

For any $X_j$ in the testing sample, the predictive performance of our local linear estimators $m(X_j)$ and $m'_{X_j}$ is compared with:
	\begin{itemize}
	\item[(L)] \emph{Functional linear regression estimator}: the standard linear regression model applied to the projections of all the involved (centered) functional data into the first $J$ basis functions \cite{ReisO07}. The estimator of $m(x)$ is the intercept estimated by this model. A sensible estimator of (the Riesz representation of) the functional derivative can be obtained as $\bphi(t)\tr \widehat{\bb_L}$, where $\widehat{\bb_L}$ is the estimate of the non-intercept terms in the linear model with $J$ regressors and the intercept. Expansion into the eigenbasis estimated from the random sample functions is considered.
	\item[(LC)] \emph{Functional local constant Nadaraya-Watson kernel estimator}:
		\[	\widehat{m}(X_j) \coloneqq \sum_{i=1}^n Y_i \, K\left(\left\Vert X_j - X_i \right\Vert/h\right) / \sum_{i=1}^n K\left(\left\Vert X_j - X_i \right\Vert/h\right)	\]
	for $K$ a kernel function, and $h$ a bandwidth \cite{FerrV06}. The local constant estimator does not allow direct estimation of the functional derivative $m'_{X_j}$.
	\item[(LL)] \emph{Functional local linear regression estimator}: the estimators of $m(X_j)$ and $m'_{X_j}$ proposed in this paper (with expansions into the eigenbasis given by the empirical covariance operator of the random sample curves). All parameters (bandwidths and approximating subspace dimension) are automatically selected.	
	\item[(MY)]  \emph{M\"uller-Yao functional additive model estimator}: the (centered) functional data are first projected into the univariate spaces given by their first $J$ estimated eigenfunctions $\widehat{\phi}_1, \dots, \widehat{\phi}_J$ to obtain their principal component scores. For all $j = 1, \dots, J$, local polynomial estimates ${\widehat{f_j}}$ and ${\widehat{f_j}}^\prime$ of the regression and its derivative, respectively, in the model of (centered) responses against the univariate scores are obtained. The additive regression operator $m$ is estimated by the sum of the functional values ${\widehat{f_j}}$ evaluated at the principal scores of $x$, plus the average of the responses $Y_i$. The final estimator of the functional derivative $m_x$ is the sum of functions $\widehat{\phi}_j$ weighted by the corresponding estimated derivatives ${\widehat{f_j}}^\prime$ evaluated at the principal scores of $x$. According to the guidelines in \cite{MullY08, MullY10}, $J$ is chosen so that the first $J$ estimated eigenfunctions explain $90~\%$ of the variability in the data. 
	\end{itemize}
Note that both the kernel estimator (LC) and the linear regression estimator (L) are special cases of the local linear estimator (LL) --- for $J = 0$ we recover the kernel estimator, and for a kernel $K(t)$ continuous at $t = 0$ from the right, the local linear smoother approaches the standard functional linear regression estimator as the bandwidth $h$ tends to infinity. 

For all competitors, the asymmetric Epanechnikov kernel $K(t) = 0.5(1 - t^2)$ for $t \in [0,\, 1]$ is used. The bandwidths, as well as the dimension $J$ in the functional (local) linear regression, are chosen by a leave-one-out cross-validation procedure. 

\bigskip

{\sc Assessing performances}. The learning and testing sample sizes are set to 500. Two perturbations are considered as in Section~\ref{subsec:ParameterSelection}: the noise-to-signal ratio $nsr$ of the regression model and the structural perturbation $\rho$ acting on regressors. Parameters $(nsr,\, \rho)$ are set to $(0.05,\, 0.05)$ and $(0.4, \, 0.4)$, corresponding to a low/high perturbation level and $a\, = \, 0, \, 0.25, \, 0.5$, $0.75, \, 1$ successively. Results for other combinations of perturbation levels are provided in Appendix~\ref{app:C}. 100 runs are performed in each case. To assess the prediction quality, we use the $ORMSEP$ criteria from \eqref{ORMSEP reg} and \eqref{ORMSEP}, except for $a = 0$ and the derivative, where the denominator of $ORMSEP_{deriv}$ is null since $m'_{0,x} \equiv \beta$ for any $x$. In the latter case, we report only the numerator of $ORMSEP_{deriv}$ from \eqref{ORMSEP}. Mean and standard deviation (in brackets) can be found in Tables~\ref{tab:R-M2-1} and \ref{tab:R-M2-4}, each table corresponding to a particular noise level.
\begin{table}[htpb]
\centering
\caption{Model (M3) with $nsr = 0.05$ and $\rho=0.05$.} 
\label{tab:R-M2-1}
\resizebox{\textwidth}{!}{\begin{tabular}{c|c|ccccc}
&   & $a = 0$ & $a = 0.25$ & $a = 0.5$ & $a = 0.75$ & $a = 1$ \\ 
  \hline
\parbox[t]{2mm}{\multirow{4}{*}{\rotatebox[origin=c]{90}{Reg.}}}& L & 0.001 \,(0.000) & 0.014 \,(0.001) & 0.111 \,(0.009) & 0.531 \,(0.032) & 1.008 \,(0.010) \\ 
&   LC & 0.048 \,(0.006) & 0.051 \,(0.007) & 0.070 \,(0.011) & 0.159 \,(0.017) & 0.265 \,(0.027) \\ 
&   LL & 0.026 \,(0.021) & 0.026 \,(0.018) & 0.035 \,(0.020) & 0.061 \,(0.011) & 0.099 \,(0.012) \\ 
&   MY & 0.028 \,(0.029) & 0.033 \,(0.021) & 0.071 \,(0.022) & 0.165 \,(0.027) & 0.273 \,(0.071) \\ 
   \hline
\parbox[t]{2mm}{\multirow{3}{*}{\rotatebox[origin=c]{90}{Deriv.}}}& L & \graytext{1.852 \,(0.024)} & 1.058 \,(0.079) & 1.016 \,(0.025) & 1.011 \,(0.016) & 1.004 \,(0.011) \\ 
&   LL & \graytext{0.102 \,(0.089)} & 8.679 \,(8.054) & 0.989 \,(1.169) & 0.092 \,(0.157) & 0.055 \,(0.048) \\ 
&   MY & \graytext{0.397 \,(0.126)} & 3.266 \,(2.804) & 0.573 \,(0.433) & 0.257 \,(0.127) & 0.216 \,(0.069) \\ 
  \end{tabular}}
\end{table}

\begin{table}[htpb]
\centering
\caption{Model (M3) with $nsr = 0.4$ and $\rho=0.4$.} 
\label{tab:R-M2-4}
\resizebox{\textwidth}{!}{\begin{tabular}{c|c|ccccc}
&   & $a = 0$ & $a = 0.25$ & $a = 0.5$ & $a = 0.75$ & $a = 1$ \\ 
  \hline
\parbox[t]{2mm}{\multirow{4}{*}{\rotatebox[origin=c]{90}{Reg.}}}& L & 0.008 \,(0.004) & 0.021 \,(0.004) & 0.119 \,(0.011) & 0.542 \,(0.037) & 1.014 \,(0.020) \\ 
&   LC & 0.153 \,(0.023) & 0.158 \,(0.024) & 0.206 \,(0.028) & 0.403 \,(0.036) & 0.629 \,(0.051) \\ 
&   LL & 0.086 \,(0.063) & 0.081 \,(0.060) & 0.120 \,(0.051) & 0.252 \,(0.040) & 0.393 \,(0.039) \\ 
&   MY & 0.055 \,(0.044) & 0.058 \,(0.025) & 0.103 \,(0.027) & 0.245 \,(0.077) & 0.369 \,(0.079) \\ 
   \hline
\parbox[t]{2mm}{\multirow{3}{*}{\rotatebox[origin=c]{90}{Deriv.}}}& L & \graytext{1.851 \,(0.028)} & 1.201 \,(0.090) & 1.034 \,(0.017) & 1.012 \,(0.007) & 1.005 \,(0.006) \\ 
&   LL & \graytext{0.849 \,(0.138)} & 7.561 \,(7.474) & 1.419 \,(0.937) & 0.572 \,(0.172) & 0.461 \,(0.074) \\ 
&   MY & \graytext{0.565 \,(0.204)} & 7.702 \,(10.505) & 1.022 \,(0.527) & 3.782 \,(33.740) & 0.307 \,(0.111) \\ 
  \end{tabular}}
\end{table}

From the results of the simulation study we conclude the following: \begin{enumerate*}[label=(\roman*)] \item the functional linear estimator is the best method when the model is linear, or close to linear. In the situation when the model is strongly nonlinear, the estimator fails as expected. \item The local linear estimator of the regression operator convincingly outperforms the local constant estimator in all considered scenarios. \item For higher values of $a$ which correspond to models far from linear, (MY) performs worse than the local linear estimator, for both regression operator and functional derivatives. This corroborates the good finite sample properties of the local linear estimator observed before, as all models considered in (M3) in the simulation study satisfy the additivity condition, under which (MY) was designed. Note that for (LL) additivity of the regression operator is not required. \item Another practical issue regarding the behavior of the estimators is their numerical stability. In the complete results of this simulation study, given in Appendix~\ref{app:C}, we observed that (MY) tends to be numerically unstable, especially for small learning sample sizes. The instabilities occur mostly when the principal scores of a predictor lie outside the range of the scores of the data, in which case either the functional values $f_j$, or the derivatives $f_j^\prime$ have to be extrapolated. Remarkably, (LL) does not appear to suffer from such drawbacks. \end{enumerate*}

\subsection{Benchmark growth data analysis} \label{subsec:GrowthDataset}
Berkeley growth data trace back to the pioneering work \cite{TuddS54} and were reconsidered in \cite{GassKMKLMP84, KneiG92, RamsL98, GervG05}. Recently, in \cite{HallMY09} an interesting analysis of this dataset involving estimated functional derivatives was performed. To better understand the growth mechanism, the relationship between the growth velocity profile up to 10 years of age (the functional predictor $X$) and the adult height (observed at 18 years, scalar response $Y$) of the boys (39 individuals) was investigated. Here, we consider the same problem using the local linear methodology. Our approach allows to estimate the functional derivatives corresponding to individual regressors $X$ directly, which greatly facilitates the interpretation of the results. In Figure~\ref{fig:growth_dataset}(a) we see the growth velocity profiles obtained via standard univariate local linear regression and in~Figure \ref{fig:growth_dataset}(b) the estimated functional derivatives $\widehat{m'_{X_1}},\ldots,\widehat{m'_{X_{39}}}$ are displayed. Focusing on the estimated functional derivatives, a sharp increase at around 6 years is observed for all boys. A possible interpretation is that the growth velocity profile prior to the age of 6 years has little impact on the adult height of an individual. Compared to the previous analyses of the growth dataset, this finding appears to be original.

To better understand the shape of the estimated functional derivatives, we propose to model the relationship between the adult height at 18 and the growth velocity up to 10 using a single-functional index model
$$
\mbox{height at 18} \ = \ g\left( \langle \mbox{growth velocity up to 10}, \, \beta \rangle \right) \ + \ error,
$$
where the link function $g$ and the functional direction $\beta$ are unknown. The single-functional index model is well suited for the studied problem, as all the estimated functional derivatives share a common shape. Therefore, the average functional derivative is a good representative of the collection of the estimated derivatives. Using the average derivative estimation method described in the introduction, we can estimate the functional parameter $\beta$ and the link function $g$ as follows: 1) $\widehat{\E m'_X} \coloneqq 1/39 \sum_{i=1}^{39} \widehat{m'_{X_i}}$ and $\widehat{\beta} = \widehat{\E m'_X} / \| \widehat{\E m'_X}  \|$, 2) based on the sample $(Z_1, Y_1),\ldots,(Z_{39}, Y_{39})$ where $Z_i \coloneqq \langle X_i, \widehat{\beta} \rangle$, one gets an estimator $\widehat{g}$ of the link function $g$ by any standard univariate nonparametric regression method. 
\begin{figure}[htpb] 
\centering
\begin{tabular}{cc}
 (a) growth velocity & (b) estimated derivatives \\
\includegraphics[scale = 0.4]{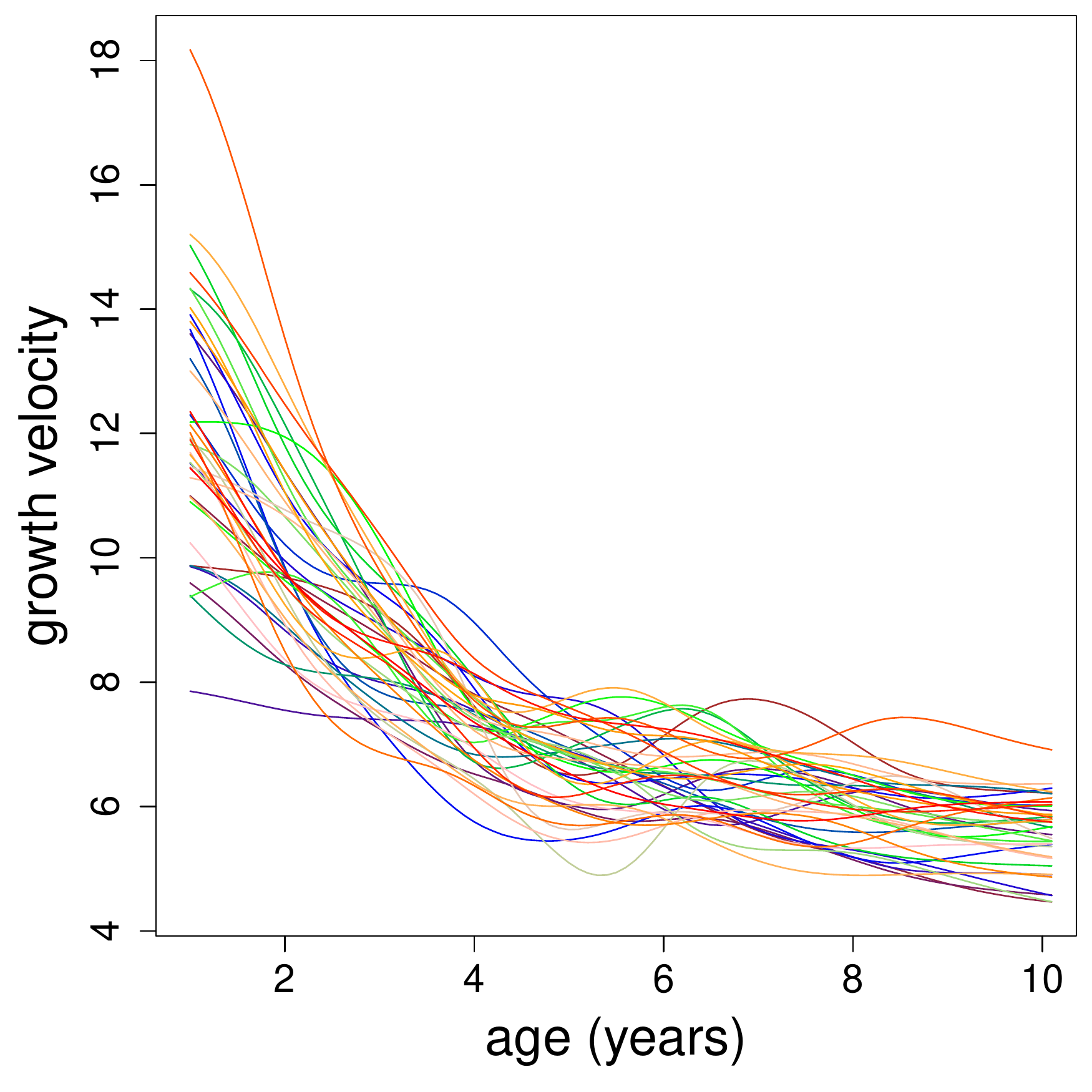}
&
\includegraphics[scale = 0.4]{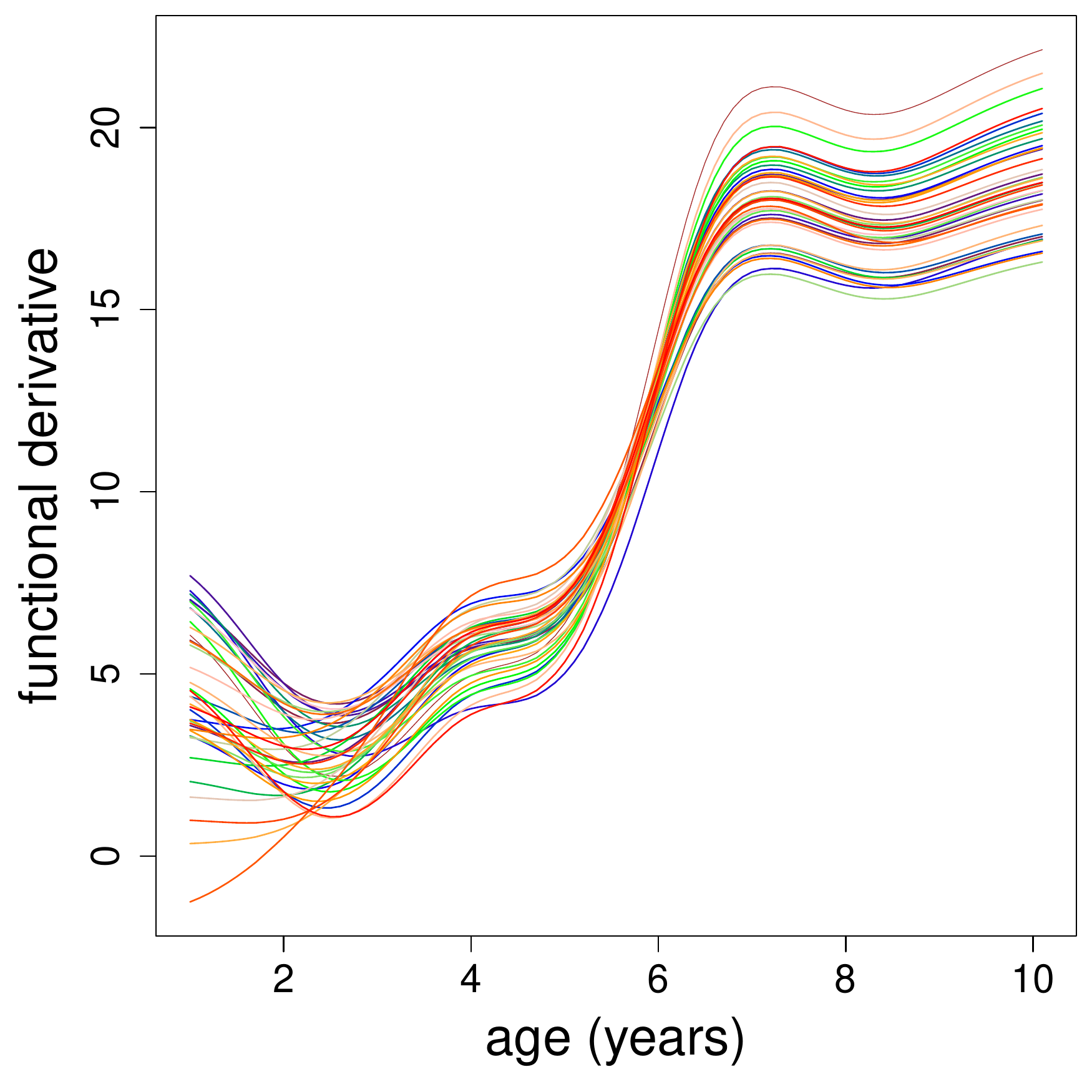}
\end{tabular}
\caption{Berkeley growth dataset: (a) growth velocity profiles; (b) functional derivatives estimated using the local linear approach.}
\label{fig:growth_dataset}
\end{figure}
Figures~\ref{fig:growth_dataset_sfim}(a) and (b) display respectively $\widehat{g}$ and $\widehat{\beta}$. The shape of the estimated functional index $\widehat{\beta}$ reflects the significant jump at around 6 years. The results are quite positive; the estimated heights are strongly correlated with the observed adult heights (Pearson's correlation $\simeq 0.82$). Since $\widehat{g}$ is a positive and non-decreasing function, the growth velocity after 6 years of age plays a major role in the prediction of the adult height. In order to confirm this interpretation, the single-functional index model was estimated again, but by restricting the growth velocity to the range of 6--10 years of age. As expected, the quality of the estimation is comparable; similar correlation ($\simeq 0.82$) between the estimated adult heights and the observed ones is obtained. For additional details we refer to Appendix~\ref{app:C}. 
\begin{figure}[htpb] 
\centering
\begin{tabular}{ccc}
 (a) function $\widehat{g}$ & (b) function $\widehat{\beta}$ & (c) estimates \\
\includegraphics[scale = 0.25]{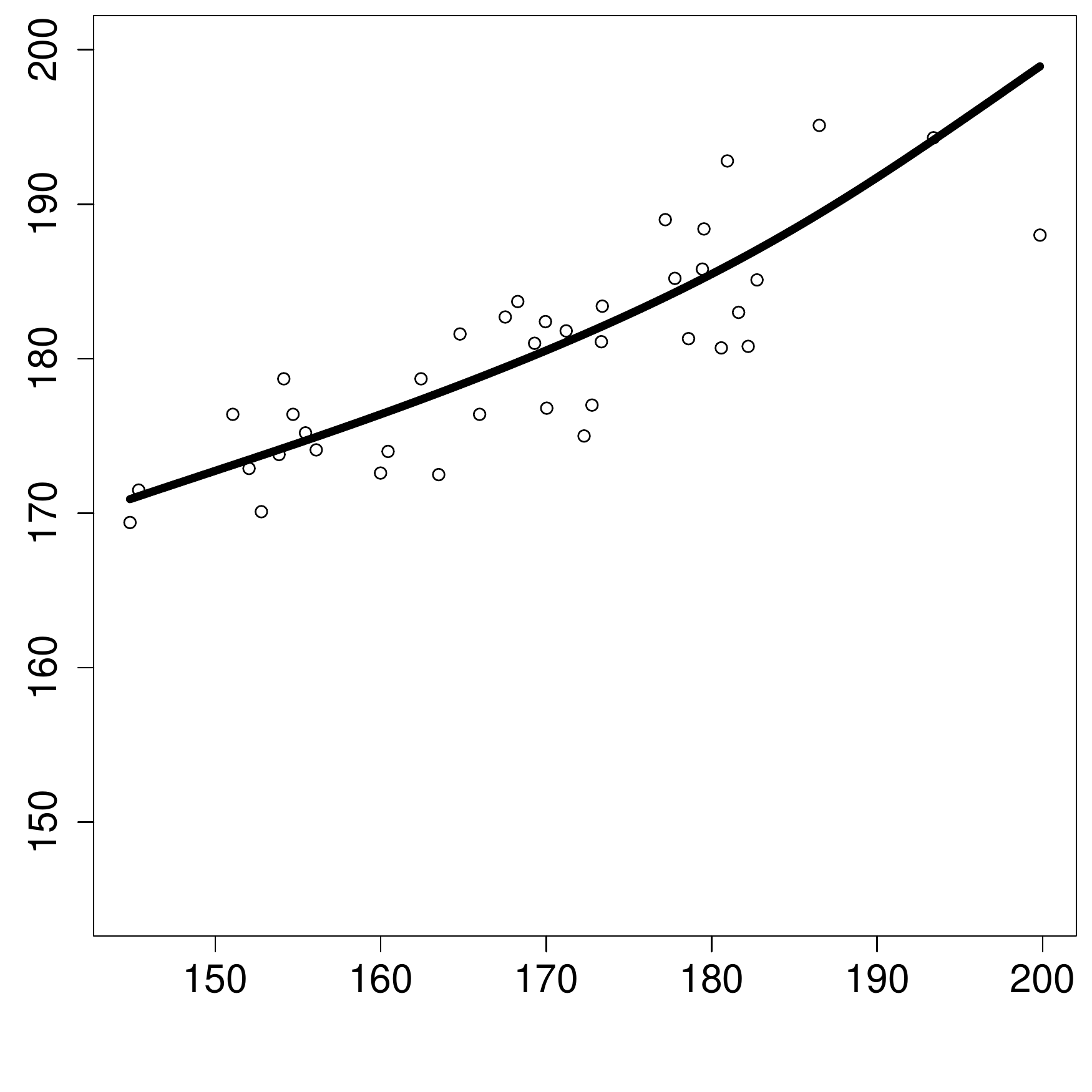}
&
\includegraphics[scale = 0.25]{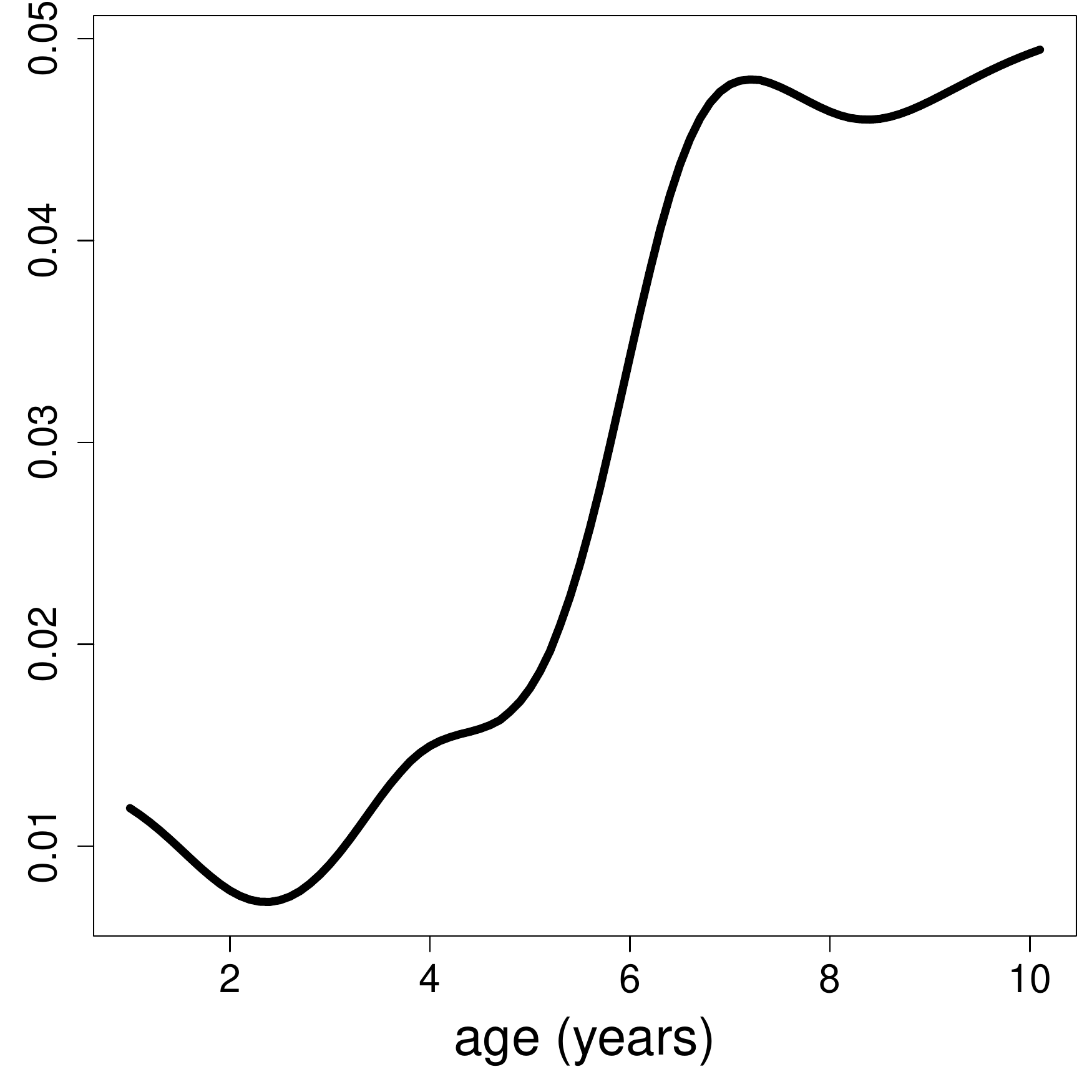}
&
\includegraphics[scale = 0.25]{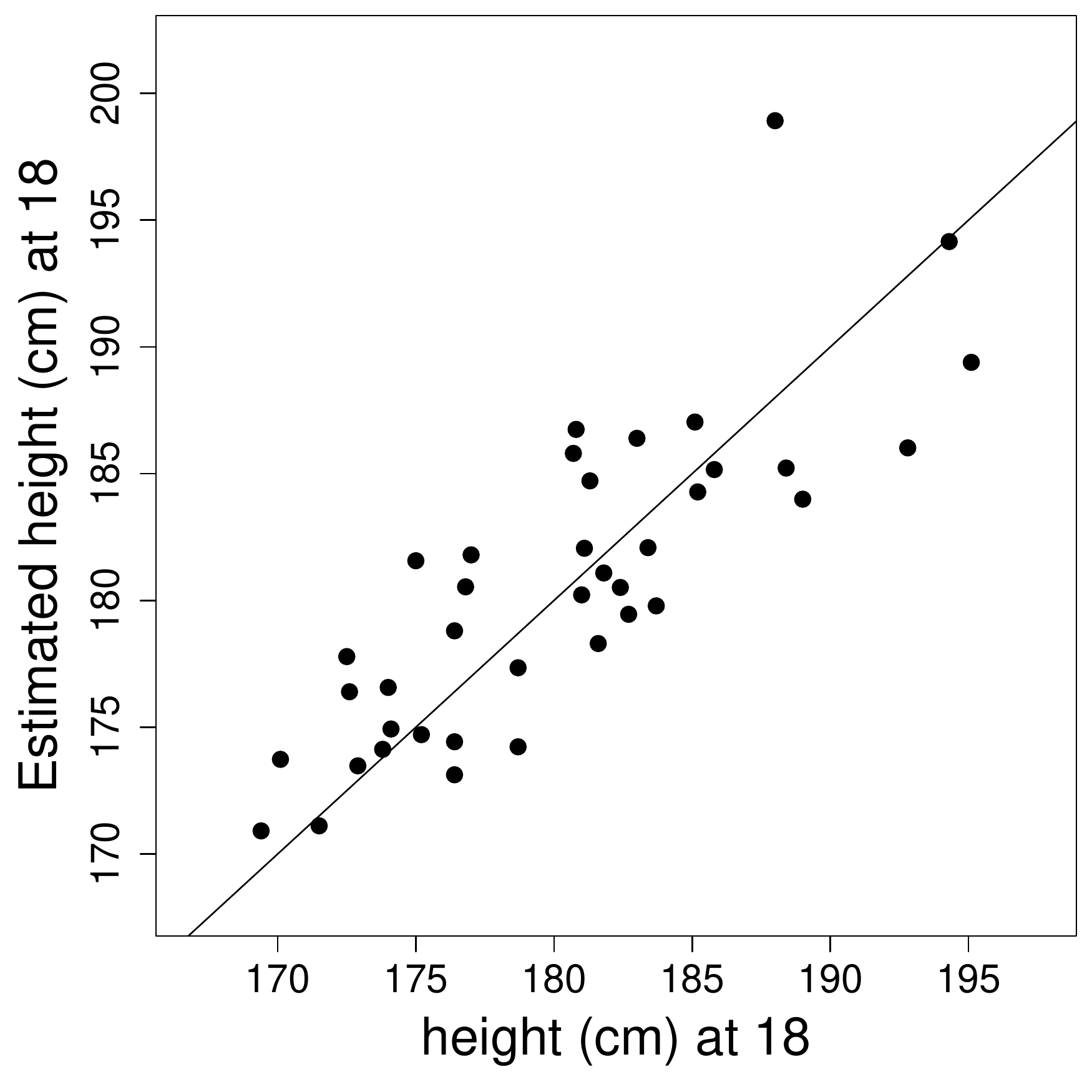}
\end{tabular}
\caption{Berkeley growth dataset: (a) link function $\widehat{g}$ estimated by local linear regression; (b) estimated functional index $\widehat{\beta}$; (c) observed adult height versus its estimates.}
\label{fig:growth_dataset_sfim}
\end{figure}

\subsection*{Acknowledgments} The work of S. Nagy was supported by the grant 19-16097Y of the Czech Science Foundation, and by the PRIMUS/17/SCI/3 project of Charles University.

\appendixpage
\appendix

Throughout the appendices, we use several conventions. Sums indicated by $\sum_i$ are always meant with $i$ from $1$ to $n$, and sums indicated by $\sum_{j \geq J}$ mean sums with $j$ from $J$ to infinity, for $J$ given; $\b1$ stands for a column vector of ones of appropriate dimension; $1_{[0, u]}(t)$ is the indicator of $t \in [0,u]$, i.e. $1$ if $t \in [0,u]$ and $0$ otherwise. The $(j,k)$-th element of a matrix $\bDelta$ can be denoted either by $\left[ \bDelta \right]_{jk}$, or equivalently by $\Delta_{jk}$.

\section{Details of proofs} \label{app:A}
\begin{proof}[Proof of {\sc \prettyref{lem:bias}}-$(i)$] We know that $T_1 = \be\tr \,\left( \bPhi\tr \widetilde{\bK} \bPhi  \right)^{-1}   \left[ A_0, A_1,  \dots,  A_J \right]\tr$ with
	\[	
	\begin{aligned}
	A_0 & \coloneqq (n \E K_1)^{-1} \sum_i K_i \langle \cP_{\cS_J^\perp} m_x', X_i-x \rangle, \\
	A_j & \coloneqq (n \E K_1)^{-1} \sum_i K_i \langle \cP_{\cS_J^\perp} m_x', X_i-x \rangle  \langle \phi_j, X_i-x \rangle, \quad \mbox{for $j=1,\ldots, J$}, 
	\end{aligned}
	\]
for $i=1,\ldots,n$, $K_i \coloneqq K\left( h^{-1} \| X_i - x\| \right)$ and $\widetilde{\bK} \coloneqq (n \E K_1)^{-1} \bK$. Set $\delta_0\coloneqq (n \, \E K_1)^{-1} \sum_i K_i$, $\delta_j\coloneqq (n \, \E K_1)^{-1} \sum_i \langle \phi_j, X_i-x \rangle K_i$ for $j=1,\ldots,J$ and let $\bDelta$ be the $J\times J$ matrix whose the $(j,k)$-th entry is equal to $\Delta_{jk}\coloneqq(n \, \E K_1)^{-1} \sum_i \langle \phi_j, X_i-x \rangle \langle \phi_k, X_i-x \rangle K_i$ for $j,k=1,\ldots,J$. By using elementary linear algebra,
$$
\bPhi\tr \widetilde{\bK} \bPhi \, = \, \left[ \begin{array}{c|c} \delta_0 & \bdelta\tr \\ \hline & \\[-1.0em] \bdelta & \bDelta \end{array} \right],
$$
where $\bdelta\coloneqq[ \delta_1, \dots, \delta_J]\tr$. According to standard results with respect to the inverse of a 2$\times$2 block matrix (see for instance \cite{LuS02}), 
\begin{equation} \label{eq:inverse}
\left(\bPhi\tr \widetilde{\bK} \bPhi \right)^{-1} =  \left[ \begin{array}{c|c} \mu & - \mu \, \bdelta\tr \bDelta^{-1}\\ \hline & \\[-1.0em] - \mu \, \bDelta^{-1} \bdelta & \bDelta^{-1}\, +\, \mu \, \bDelta^{-1} \bdelta \, \bdelta\tr \bDelta^{-1} \end{array} \right] 
\end{equation}
with $\mu \coloneqq \left( \delta_0 - \bdelta\tr \bDelta^{-1} \bdelta \right)^{-1}$. Then, $
T_1 \, = \, \mu \left( A_0 \, - \, \bdelta\tr \bDelta^{-1} \left[A_1, \dots, A_J \right]\tr \right)$. Before going on, let us focus on $\bDelta^{-1}$, the inverse of $\bDelta$. From  {\sc Lemma} \ref{lem:Delta}, one has $\bDelta = \bDelta_1 + \bDelta_2$ with $\ds \bDelta_1 \coloneqq b_{x,0,1}^{-1} \,  b_{x,2,1} \, h^2 \, \boldsymbol{ \Gamma } \, \left\{ 1 + o(1) \right \}$ and $\ds \bDelta_2 \coloneqq O_P\left( h^2 \{n \pi_x(h)\}^{-1/2} \right) \boldsymbol{ \Lambda }$. As soon as $\bDelta_1^{-1}$ is invertible, $\bDelta^{-1} = \bDelta_1^{-1}(\bI + \bDelta_2 \bDelta_1^{-1})^{-1}$. Let $\| \cdot \|_F$ stand for the Frobenius matrix norm and recall that $\lambda_J$ is the smallest eigenvalue of the $J\times J$ matrix $\boldsymbol{ \Gamma }$. We have that $\| \bDelta_1^{-1} \|_F = O_P\left( \lambda_J^{-1} \, J^{1/2} \, h^{-2} \right)$, and thanks to {\sc \prettyref{lem:gammafunction}}, $\| \bDelta_2 \|_F = O_P\left( h^2 \, \{n \pi_x(h)\}^{-1/2} \right)$ so that $\| \bDelta_2 \bDelta_1^{-1} \|_F = O_P\left( \lambda_J^{-1} \, J^{1/2} \,  \{n \pi_x(h)\}^{-1/2} \right)$. According to (\prettyref{hypo:asymptotics}), for $n$ large enough, the Frobenius norm of $\bDelta_2 \bDelta_1^{-1}$ is smaller than 1. Then $\bDelta^{-1} = \bDelta_1^{-1}$ $\left\{\bI + \sum_{k \geq 1} (-1)^k \left( \bDelta_2 \bDelta_1^{-1} \right)^k  \right\}$ which results in  
\begin{equation} \label{eq:inverseDelta}
\bDelta^{-1} =  b_{x,0,1} \,  b_{x,2,1}^{-1} \, h^{-2} \,  \boldsymbol{ \Gamma }^{-1} \, \left\{ 1 + o_P(1) \right \}.
\end{equation}
{\sc Lemmas} \ref{lem:delta} and \ref{lem:A0Aj} allow us to write
\begin{equation}
\begin{aligned} \label{eq:T11}
T_1 & = \mu \, b_{x,0,1}^{-1} b_{x,1,1} \Big\{ \alpha_{0,x,n}^{bias} h \left\{ 1 + o(1) \right\}  + O_P \left( \|  \cP_{\cS_J^\perp} m_x'  \| h \{n \pi_x(h)\}^{-1/2} \right)   \\
& \phantom{=} - \,   \left(\boldsymbol{\gamma}\tr + \boldsymbol{\theta}\tr O_P\left(\{n \pi_x(h)\}^{-1/2}\right)\right)  \boldsymbol{ \Gamma }^{-1} \boldsymbol{ \alpha_{x,n}^{bias} }   h \left\{ 1 + o_P(1) \right\} \\
& \phantom{=} + \, \left(\boldsymbol{\gamma}\tr + \boldsymbol{\theta}\tr O_P\left(\{n \pi_x(h)\}^{-1/2}\right)\right)  \boldsymbol{ \Gamma }^{-1}  \sqrt{ \boldsymbol{ \alpha_{x,n}^{var} } } \, O_P \left( h \, \{n \pi_x(h)\}^{-1/2} \right)\Big\} 
\end{aligned}
\end{equation}
where we denote $\bgamma \coloneqq [{\gamma_1^1}^{'}(0), \dots, {\gamma_J^1}^{'}(0)]\tr$, $\btheta \coloneqq \left[\sqrt{{\gamma_1^2}^{'}(0)}, \dots, \sqrt{{\gamma_J^2}^{'}(0)}\right]\tr$, $\boldsymbol{ \alpha_{x,n}^{bias} } \coloneqq \left[\alpha_{1,x,n}^{bias}, \dots,  \alpha_{J,x,n}^{bias}\right]\tr$ and $\sqrt{ \boldsymbol{ \alpha_{x,n}^{var} } } \coloneqq \left[\sqrt{\alpha_{1,x,n}^{var}}, \dots,  \sqrt{\alpha_{J,x,n}^{var}}\right]\tr$. From {\sc \prettyref{lem:A0Aj}} we know that $\| \boldsymbol{ \alpha_{x,n}^{bias} } \|_2 \leq \|  \cP_{\cS_J^\perp} m_x'  \|$ and $\| \sqrt{ \boldsymbol{ \alpha_{x,n}^{var} } } \|_2 \leq \|  \cP_{\cS_J^\perp} m_x'  \|$. Here $\| . \|_2$ stands for the Euclidean vector norm. Based on {\sc Lemmas} \ref{lem:gammafunction}, \ref{lem:A0Aj}, and \ref{lem:innerproductupperbound},
\begin{equation}\label{eq:T12}
\begin{aligned}
\left(\boldsymbol{\gamma}\tr + \boldsymbol{\theta}\tr O_P\left(\{n \pi_x(h)\}^{-1/2}\right)\right)  \boldsymbol{ \Gamma }^{-1} \boldsymbol{ \alpha_{x,n}^{bias} } & = O_P\left( \lambda_J^{-1} \, \left\|  \cP_{\cS_J^\perp} m_x'  \right\| \right), \\
\left(\boldsymbol{\gamma}\tr + \boldsymbol{\theta}\tr O_P\left(\{n \pi_x(h)\}^{-1/2}\right)\right) \boldsymbol{ \Gamma }^{-1}  \sqrt{ \boldsymbol{ \alpha_{x,n}^{var} } } & = O_P\left( \lambda_J^{-1} \, \left\|  \cP_{\cS_J^\perp} m_x'  \right\| \right),
\end{aligned}
\end{equation}
where $\lambda_J$ is the smallest eigenvalue of the $J\times J$ matrix $\boldsymbol{ \Gamma }$. {\sc \prettyref{lem:gammafunction}} indicates that the Frobenius norm of $\boldsymbol{ \Gamma }$ is finite and thus $\lambda_J$ converges to 0 with $n$. Let us now focus on the term $\mu$. Firstly, one has $\E \delta_0 =1$ and $\var(\delta_0)= n^{-1} (\E K_1)^{-2} \var(K_1)$. Thanks to  {\sc Corollary} \ref{cor:expectationkernel}, it is easy to see that $\var(\delta_0)= O\left( 1 / \{n \, \pi_x(h)\} \right)$ which leads to $\delta_0= 1 + O_P\left( 1 / \sqrt{n \, \pi_x(h)} \right)$. Secondly, {\sc Lemma} \ref{lem:delta} and \prettyref{eq:inverseDelta} imply that 
\begin{equation*}
\begin{aligned}
\bdelta\tr  \bDelta^{-1} \bdelta & = b_{x,0,1}^{-1} \, b_{x,1,1}^2 \, b_{x,2,1}^{-1} \,  \bgamma\tr    \boldsymbol{ \Gamma }^{-1}  \bgamma \,  \left\{ 1 + o_P(1) \right\} \\ 
& \phantom{=} + \, O_P\left( \{ n \, \pi_x(h) \}^{-1/2}  \right) \btheta\tr \boldsymbol{ \Gamma }^{-1}  \bgamma  \, + \,  O_P\left( \{ n \, \pi_x(h) \}^{-1}  \right) \btheta\tr \boldsymbol{ \Gamma }^{-1}  \btheta.
\end{aligned}
\end{equation*}
{\sc Lemmas} \ref{lem:gammafunction} and \ref{lem:innerproductupperbound} give $ \btheta\tr \boldsymbol{ \Gamma }^{-1}  \bgamma \, / \, \bgamma\tr    \boldsymbol{ \Gamma }^{-1}  \bgamma = O(\lambda_J^{-1}) = \btheta\tr \boldsymbol{ \Gamma }^{-1}  \btheta \, / \, \bgamma\tr    \boldsymbol{ \Gamma }^{-1}  \bgamma$. Then (\prettyref{hypo:asymptotics}) leads to 
$\{ n \, \pi_x(h) \}^{-1/2} \btheta\tr \boldsymbol{ \Gamma }^{-1}  \bgamma = o_P( \bgamma\tr    \boldsymbol{ \Gamma }^{-1}  \bgamma) = \{ n \, \pi_x(h) \}^{-1} \btheta\tr \boldsymbol{ \Gamma }^{-1}  \btheta$ and finally, 
$\bdelta\tr  \bDelta^{-1} \bdelta \, = \, b_{x,0,1}^{-1} \, b_{x,1,1}^2 \, b_{x,2,1}^{-1} \,  \bgamma\tr    \boldsymbol{ \Gamma }^{-1}  \bgamma \,  \left\{ 1 + o_P(1) \right\}$. Use again {\sc Lemma}  \ref{lem:innerproductupperbound} to get 
	\begin{equation}	\label{eq:mu}
	\bdelta\tr  \bDelta^{-1} \bdelta \, = \, O_P\left( \lambda_J^{-1} \right), \mbox{ which results in } \mu = O_P\left( \lambda_J \right).
	\end{equation}
This last result combined with \eqref{eq:T11} and \eqref{eq:T12} gives $T_1 \, = \, O_P\left( h \, \|  \cP_{\cS_J^\perp} m_x'  \| \right)$.
\end{proof}

\begin{proof}[Proof of {\sc \prettyref{lem:bias}}-$(ii)$] We have that $T_2 = T_{21} + T_{22}$ where
\begin{itemize}
\item[$\bullet$] $T_{21} \coloneqq \be\tr \,\left( \bPhi\tr \widetilde{\bK} \bPhi  \right)^{-1} \left[ B_0, B_1, \dots,  B_J \right]\tr$ with $B_0  \coloneqq (n \E K_1)^{-1} \sum_i K_i R_{x, x, i}$, 
and for $j=1,\ldots, J$, $B_j \coloneqq (n \E K_1)^{-1} \sum_i K_i R_{x, x, i} \langle \phi_j, X_i-x \rangle$, where $R_{x,x,i}$ is a term involved in \eqref{eq:expansion} with $\zeta$ replaced by $x$;
\item[$\bullet$] $T_{22} \coloneqq \be\tr \,\left( \bPhi\tr \widetilde{\bK} \bPhi  \right)^{-1} \left[ C_0, C_1, \dots, C_J \right]\tr$ where $C_0  \coloneqq (n \E K_1)^{-1} \sum_i K_i \left(R_{\zeta, x, i} \right.$ $ - \left. R_{x, x, i} \right)$,  and $C_j \coloneqq (n \E K_1)^{-1} \sum_i K_i \left(R_{\zeta, x, i} - R_{x, x, i} \right) \langle \phi_j, X_i-x \rangle$ for $j=1,\ldots, J$.
\end{itemize}
{\em About $T_{21}$}. {\sc Lemmas} \ref{lem:delta} and \eqref{lem:B0Bj} allow to write 
\begin{equation} \label{eq:T211}
\begin{aligned}
T_{21} & = \mu \, b_{x,0,1}^{-1} \Big\{  b_{x,2,1} \, \beta_{0,x}^{bias} \, h^2 \left\{ 1 + o(1) \right\}  + O_P \left(  h^2 \{n \, \pi_x(h)\}^{-1/2} \right)  \\
& \phantom{=} - \, b_{x,1,1} b_{x,2,1}^{-1} b_{x,3,1} \left( \boldsymbol{\gamma}\tr + \boldsymbol{\theta}\tr O_P\left( \{n \pi_x(h)\}^{-1/2} \right) \right) \boldsymbol{ \Gamma }^{-1} \boldsymbol{ \beta_{x}^{bias} }  \, h^2 \left\{ 1 + o_P(1) \right\} \\ 
& \phantom{=} \left. + \, \left( \boldsymbol{\gamma}\tr + \boldsymbol{\theta}\tr O_P\left( \{n \pi_x(h)\}^{-1/2} \right) \right) \boldsymbol{ \Gamma }^{-1}  \sqrt{ \boldsymbol{ \beta_{x}^{var} } } \, O_P \left( h^2 \{n \pi_x(h)\}^{-1/2} \right)\right\}. 
\end{aligned}
\end{equation}
Based on  {\sc Lemmas} \ref{lem:gammafunction}, \ref{lem:innerproductupperbound} and \ref{lem:B0Bj},
\begin{equation}\label{eq:T212}
\begin{aligned}
\left( \boldsymbol{\gamma}\tr + \boldsymbol{\theta}\tr O_P\left( \{n \pi_x(h)\}^{-1/2} \right) \right) \boldsymbol{ \Gamma }^{-1} \boldsymbol{ \beta_{x}^{bias} } \, & = \, O_P\left( \lambda_J^{-1}  \right), \\
\left( \boldsymbol{\gamma}\tr + \boldsymbol{\theta}\tr O_P\left( \{n \pi_x(h)\}^{-1/2} \right) \right)  \boldsymbol{ \Gamma }^{-1}  \sqrt{ \boldsymbol{ \beta_{x}^{var} } }& = \, O_P\left( \lambda_J^{-1}  \right).
\end{aligned}
\end{equation}
Now, \eqref{eq:mu}, \eqref{eq:T211}, and \eqref{eq:T212} result in $T_{21} \, = \, O_P\left( h^2 \right)$.

~\\ {\em About $T_{22}$}. Thanks to (\prettyref{hypo:taylor}), it is easy to show that $C_0 = O_P(h^3)$ and $C_j = O_P(h^4)$ for $j=1,\ldots, J$. Consequently, 
	\[	
	\begin{aligned}
T_{22} & = \be\tr \,\left( \bPhi\tr \widetilde{\bK} \bPhi  \right)^{-1}\left[ 1, h, \dots, h \right]\tr O_P(h^3)\\
& = \mu \left\{ 1 - h \, \bdelta\tr \bDelta^{-1} \b1  \right\} O_P(h^3).
	\end{aligned}
	\]
According to \eqref{eq:mu} and {\sc Lemmas} \ref{lem:Delta}, \ref{lem:delta} and \ref{lem:innerproductupperbound}, $T_{22} = O_P\left( h^3 \sqrt{J} \right)$. 

\smallskip

{\em Back to $T_2$}. Because $h \, \sqrt{J} = o(1)$ thanks to (\prettyref{hypo:asymptotics}), $T_{22} = o_P(h^2)$ and $T_2 = O_P(h^2)$.
\end{proof}

Let us now focus on the conditional variance of $\widehat{m}(x)$. With the notations introduced in the proof of  {\sc \prettyref{lem:bias}}-$(i)$, the next results provide a decomposition of the conditional variance with the asymptotic behavior of each term. 
\begin{lem}\label{lem:variance} As soon as conditions  (\prettyref{hypo:taylor})--(\prettyref{hypo:condvar}) are fulfilled, 
	\begin{enumerate}[label=(\roman*)]
	\item $\varX \left\{ \widehat{m}(x) \right\} = \{ \sigma^2(x) \, + \, o(1) \} \, \mu^2 \, \left\{\widetilde{\delta}_0 - \bdelta\tr \bDelta^{-1} \widetilde{\bdelta} -  \widetilde{\bdelta}\tr  \bDelta^{-1} \bdelta     + \bdelta\tr \bDelta^{-1} \widetilde{\bDelta} \bDelta^{-1}  \bdelta \right\},$\\ 
where $\widetilde{\delta}_0 \coloneqq (n \, \E K_1)^{-2} \, \sum_i K_i^2$, $\widetilde{\delta}_j \coloneqq (n \, \E K_1)^{-2} \sum_i \langle \phi_j, X_i - x \rangle \, K_i^2$ for $j=1,\ldots,J$, $\widetilde{\bdelta} \coloneqq [\widetilde{\delta}_1, \dots, \widetilde{\delta}_J]\tr$,  $\widetilde{\Delta}_{jk} \coloneqq (n \, \E K_1)^{-2}\sum_i \langle \phi_j, X_i - x \rangle \, \langle \phi_k, X_i - x \rangle \, K_i^2$ for $j, k=1,\ldots,J$, and $\widetilde{\bDelta}$ is the $J\times J$ matrix such that $[\widetilde{\bDelta}]_{jk} \coloneqq \widetilde{\Delta}_{jk}$,
	\item $\widetilde{\delta}_0 = b_{x,0,1}^{-2} \, b_{x,0,2} \, \{ n \, \pi_x(h) \}^{-1} \{ 1 + o_P(1) \},$  
$ \bdelta\tr \bDelta^{-1} \widetilde{\bdelta} = O_P\left( \, \{  \lambda_J \, n \, \pi_x(h) \}^{-1} \right),$ and \linebreak $ \bdelta\tr \bDelta^{-1} \widetilde{\bDelta} \bDelta^{-1} \bdelta = O_P \left( \{ \lambda_J \, n \, \pi_x(h) \}^{-1} \right)$.
	\end{enumerate}
\end{lem}

\begin{proof}[Proof of {\sc \prettyref{lem:variance}}-$(i)$] The main term of $\varX \left\{ \widehat{m}(x) \right\}$ is given by 
	\[	\be\tr \, \left( \bPhi\tr \widetilde{\bK} \bPhi  \right)^{-1} \, \bPhi\tr \, \widetilde{\bK}^2 \,  \bPhi \, \left( \bPhi\tr \widetilde{\bK} \bPhi  \right)^{-1} \be	\]
where $\bPhi\tr \widetilde{\bK}^2 \bPhi \, = \, \left[ \begin{array}{c|c} \widetilde{\delta}_0 & \widetilde{\bdelta}\tr \\ \hline & \\[-1.0em] \widetilde{\bdelta} & \widetilde{\bDelta} \end{array} \right].$
Formula~\prettyref{eq:inverse} with elementary linear algebra result in the claimed decomposition.
\end{proof}

\begin{proof}[Proof of {\sc \prettyref{lem:variance}}-$(ii)$] By {\sc Corollary} \ref{cor:expectationkernel}, $\E K_1^2 = b_{x,0,2} \, \pi_x(h) \, \{ 1 + o(1) \}$ and in the same vein of {\sc Lemmas} \ref{lem:Delta} and  \ref{lem:delta}, it is easy to see that 
\[	
\begin{aligned}
\widetilde{\delta}_0 & = \, b_{x,0,1}^{-2} \, b_{x,0,2} \, \{ n \, \pi_x(h) \}^{-1} \, \{ 1 + o_P(1) \}, \\
\widetilde{\delta}_j & = \, b_{x,0,1}^{-2} \, b_{x,1,2} \, {\gamma_j^1}^{'} (0) \, h \, \left\{ n \,  \pi_x(h) \right\}^{-1} \, \left\{ 1 + o_P(1) \right\} + \, O_P\left( h \, \{ n \, \pi_x(h) \}^{-3/2} \right) \sqrt{ {\gamma_j^2}^{'}(0) }, \\
\widetilde{\Delta}_{jk}  & = \, b_{x,0,1}^{-2} \, b_{x,2,2} \, {\gamma_{j,k}^{1,1}}^{'} (0) \, h^2 \, \left\{ n  \pi_x(h) \right\}^{-1} \, \left\{ 1 + o_P(1) \right\} \\
& \phantom{=} + O_P\left( h^2 \, \{ n \, \pi_x(h) \}^{-3/2}   \right) \sqrt{ {\gamma_{j,k}^{2,2}}^{'}(0) }. 
\end{aligned}
\]
As a by-product, one has 
\begin{equation} \label{eq:tildedelta}
\widetilde{\bdelta} \, = \, b_{x,0,1}^{-2} \, b_{x,1,2} \, \bgamma \, h \, \left\{ n \,  \pi_x(h) \right\}^{-1} \, \left\{ 1 + o_P(1) \right\} \, + \, O_P\left( h \, \{ n \, \pi_x(h) \}^{-3/2} \right) \, \btheta,
\end{equation}
and 
\begin{equation} \label{eq:tildeDelta}
\widetilde{\bDelta}  \, = \, b_{x,0,1}^{-2} \, b_{x,2,2} \, \bGamma \, h^2 \, \left\{ n  \pi_x(h) \right\}^{-1} \, \left\{ 1 + o_P(1) \right\} \,  + \, O_P\left( h^2 \, \{ n \, \pi_x(h) \}^{-3/2}   \right) \bLambda.
\end{equation}
These two last equations with \prettyref{eq:inverseDelta} lead to 
\begin{equation}	\label{eq:eq17}
\begin{aligned}
\bdelta\tr \bDelta^{-1} \widetilde{\bdelta} & = b_{x,0,1}^{-2} \, b_{x,1,1} \, b_{x,2,1}^{-1} \, b_{x,2,2}\, \bgamma\tr \bGamma^{-1} \bgamma \, \{ n \, \pi_x(h) \}^{-1} \{ 1 + o_P(1) \} \\ 
	& \phantom{=} +  O_P\left( \{ n \, \pi_x(h) \}^{-3/2}  \right) \btheta\tr \boldsymbol{ \Gamma }^{-1}  \bgamma  \, + \,  O_P\left( \{ n \, \pi_x(h) \}^{-2}  \right) \btheta\tr \boldsymbol{ \Gamma }^{-1}  \btheta, 
\end{aligned}
\end{equation}
and 
\begin{equation} \label{eq:lem3ii2}
\begin{aligned}
\bdelta\tr \bDelta^{-1} \widetilde{\bDelta} \bDelta^{-1} \bdelta & = b_{x,0,1}^{-2} \, b_{x,1,1}^2 \, b_{x,2,1}^{-2} \, b_{x,2,2} \, \bgamma\tr \bGamma^{-1} \bgamma \, \{ n \, \pi_x(h) \}^{-1} \, \{ 1 + o_P(1) \} \\ 
& \phantom{=} + \,  O_P\left( \{ n \, \pi_x(h) \}^{-3/2}  \right) \bgamma\tr \bGamma^{-1} \bLambda  \bGamma^{-1} \bgamma.
\end{aligned}
\end{equation}
Combination of {\sc Lemmas}~\ref{lem:gammafunction}-$(ii)$ and \ref{lem:innerproductupperbound} gives $\bgamma\tr \bGamma^{-1} \bLambda  \bGamma^{-1} \bgamma = O(\lambda_J^{-2})$, $\bgamma\tr \bGamma^{-1} \bgamma = O(\lambda_J^{-1})$ and the claimed assertions hold. Finally, the rate of convergence of the conditional variance is derived from {\sc \prettyref{lem:variance}} and the fact that $\mu = O_P(\lambda_J)$.
\end{proof}

\begin{lem}\label{lem:biasderiv} As soon as conditions  (\prettyref{hypo:taylor})--(\prettyref{hypo:sbp1}) are fulfilled,
	\begin{enumerate}[label=(\roman*)]
	\item $\| Q_1 \| \, = \, O_P\left( \lambda_J^{-1} \|  \cP_{\cS_J^\perp} m_x'  \| \right)$,
	\item $\| Q_2 \| \, = \, O_P \left(   \lambda_J^{-1} \, h   \right)$.
	\end{enumerate}
\end{lem}

\begin{proof}[Proof of {\sc \prettyref{lem:biasderiv}}-$(i)$] With the notations introduced in the proof {\sc \prettyref{lem:bias}}-$(i)$ and \prettyref{eq:inverse}, $Q_1  =   -\mu \, A_0 \, \bphi\tr \bDelta^{-1}  \bdelta  +  \bphi\tr ( \bDelta^{-1} +  \mu \, \bDelta^{-1} \, \bdelta  \, \bdelta\tr  \, \bDelta^{-1}  )  \left[ A_1, \dots, A_J \right]\tr$. According to {\sc \prettyref{lem:delta}}, {\sc \prettyref{lem:A0Aj}} and \prettyref{eq:inverseDelta}, $\mu \, A_0 \, \bphi\tr \bDelta^{-1}  \bdelta  \, = \, O_P \left( \lambda_J \|  \cP_{\cS_J^\perp} m_x'  \| \, \bphi\tr \bGamma^{-1} \bgamma \right)$, $\bphi\tr \bDelta^{-1} \left[ A_1, \dots,  A_J \right]\tr =  O_P\left( \bphi\tr \bGamma^{-1}  \boldsymbol{\alpha_{x,n}^{bias}} \right)$ and $\mu \, \bphi\tr \bDelta^{-1} \, \bdelta  \, \bdelta\tr  \, \bDelta^{-1} \left[ A_1, \dots, A_J \right]\tr = $ \linebreak $ O_P \left( \|  \cP_{\cS_J^\perp} m_x'  \| \, \bphi\tr \bGamma^{-1} \bgamma \right)$. Set $\bu \coloneqq  \lambda_J \bGamma^{-1} \bgamma$. For any
  $J$, $\| \bphi\tr \bGamma^{-1} \bgamma \| = \lambda_J^{-1}  \| \bphi\tr \bu \| = \lambda_J^{-1} \| \bu \|_2 \leq \lambda_J^{-1}$. Consequently  $ \| \bphi\tr \bGamma^{-1} \bgamma \| = O_P\left( \lambda_J^{-1} \right)$. Now, define the vector $\bv \coloneqq  \lambda_J  \|  \cP_{\cS_J^\perp} m_x'  \|^{-1} \bGamma^{-1} \boldsymbol{\alpha_{x,n}^{bias}}$ with $\| \bv \|_2 \leq 1$. From the definition of $\bv$ we have $\| \bphi\tr \bGamma^{-1} \boldsymbol{\alpha_{x,n}^{bias}} \| = \lambda_J^{-1} \|  \cP_{\cS_J^\perp} m_x'  \| \, \| \bphi\tr \bv \|$. Because of the orthonormality of the basis, $ \| \bphi\tr \bv \| = \| \bv \|_2$,   $\| \bphi\tr \bGamma^{-1} \boldsymbol{\alpha_{x,n}^{bias}} \| = \lambda_J^{-1} \|  \cP_{\cS_J^\perp} m_x'  \| \, \| \bv \|_2 \leq \lambda_J^{-1} \|  \cP_{\cS_J^\perp} m_x'  \|$ which gives $ \| \bphi\tr \bGamma^{-1} \boldsymbol{\alpha_{x,n}^{bias}} \| = O_P\left( \lambda_J^{-1} \|  \cP_{\cS_J^\perp} m_x'  \| \right)$ and the claimed rate of convergence holds.
\end{proof}

\begin{proof}[Proof of {\sc \prettyref{lem:biasderiv}}-$(ii)$] Similarly to the notations and the proof of {\sc \prettyref{lem:bias}}-$(ii)$, $Q_2 = Q_{21} + Q_{22}$ with 
	\[
	\begin{aligned}
	Q_{21} & \coloneqq -\mu \, B_0 \, \bphi\tr \bDelta^{-1}  \bdelta  +  \bphi\tr ( \bDelta^{-1} +  \mu \, \bDelta^{-1} \, \bdelta  \, \bdelta\tr  \, \bDelta^{-1}  )  \left[ B_1, \dots, B_J \right]\tr, \\
	Q_{22} & \coloneqq -\mu \, C_0 \, \bphi\tr \bDelta^{-1}  \bdelta  +  \bphi\tr ( \bDelta^{-1} +  \mu \, \bDelta^{-1} \, \bdelta  \, \bdelta\tr  \, \bDelta^{-1}  )  \left[ C_1, \dots, C_J \right]\tr.
	\end{aligned}
	\]
By following the proof of {\sc \prettyref{lem:biasderiv}}-$(i)$ and replacing {\sc \prettyref{lem:A0Aj}} with {\sc \prettyref{lem:B0Bj}}, $\mu \, B_0 \, \| \bphi\tr \bDelta^{-1}  \bdelta \| = O_P\left( h  \right)$, $\| \bphi\tr \bDelta^{-1} \left[ B_1, \dots, B_J \right]\tr \| =  O_P\left(  \lambda_J^{-1} \, h  \right)$ and \linebreak $\mu \, \| \bphi\tr \bDelta^{-1} \, \bdelta  \, \bdelta\tr  \, \bDelta^{-1} \left[ B_1, \dots, B_J \right]\tr
  \| = O_P\left(  \lambda_J^{-1} \, h  \right)$ so that $\| Q_{21} \| = O_P \left(   \lambda_J^{-1} \, h   \right)$. For studying $Q_{22}$, recall that $C_0 =O_P(h^3)$ and $C_j=O_P(h^4)$ for $j=1,\ldots, J$ (see again the proof of {\sc \prettyref{lem:bias}}-$(ii)$). Then, we obtain that $\mu \, C_0 \, \| \bphi\tr \bDelta^{-1}  \bdelta \| = O_P\left( h^2  \right)$, $\| \bphi\tr \bDelta^{-1} \left[ C_1, \dots, C_J \right]\tr \| =  O_P\left(  \lambda_J^{-1} \, \sqrt{J} \, h^2  \right)$, $\mu \, \| \bphi\tr \bDelta^{-1} \, \bdelta  \, \bdelta\tr  \, \bDelta^{-1} \left[ C_1, \dots, C_J \right]\tr \| = $ \linebreak $O_P\left(  \lambda_J^{-1} \, \sqrt{J} \, h^2  \right)$ and $\| Q_{22} \| = o_P \left(   \lambda_J^{-1} \, h   \right)$ from (\prettyref{hypo:asymptotics}). This is enough to get the claimed result.
\end{proof}

\begin{proof}[Details of the proof of {\sc \prettyref{lem:biasvarderiv}}-$(ii)$] Set $\bM_0 \coloneqq \bDelta^{-1} \bdelta \bdelta\tr \bDelta^{-1}$. Thanks to \prettyref{eq:inverse}, elementary linear algebra and notations introduced in {\sc \prettyref{lem:variance}}, the following decomposition of the conditional variance 
$$ 
\EX \left( \| \widehat{m}_x' - \EX\left( \widehat{m}_x' \right)  \|^2 \right) \,  = \, \sigma^2(x) \int_0^1  \bphi(t)\tr \left\{\sum_{j=1}^8 \bM_j \right\} \bphi(t) \dd t \, \{ 1 + o(1) \}
$$
holds with 
\begin{center}
\begin{tabular}{ll}
$\bM_1 \coloneqq \mu^2 \widetilde{\delta}_0 \bM_0$, & $\bM_2 \coloneqq - \mu \bDelta^{-1} \widetilde{\bdelta} \bdelta\tr\bDelta^{-1}$, \\
$\bM_3 \coloneqq \bM_2\tr$, & $\bM_4 \coloneqq - 2\mu^2 \left( \bdelta\tr \bDelta^{-1} \widetilde{\bdelta} \right) \bM_0$,  \\
$\bM_5 \coloneqq  \bDelta^{-1} \widetilde{\bDelta}  \bDelta^{-1}$, & $\bM_6 \coloneqq \mu \bDelta^{-1} \widetilde{\bDelta} \bM_0$, \\
$\bM_7 \coloneqq \bM_6\tr$, & $\bM_8 \coloneqq \mu^2 \left( \bdelta\tr \bDelta^{-1} \widetilde{\bDelta} \bDelta^{-1} \bdelta \right) \bM_0$.
\end{tabular}
\end{center}
According to \prettyref{eq:inverse}, \prettyref{eq:mu}, \prettyref{eq:tildedelta}, \prettyref{eq:tildeDelta}, \prettyref{eq:eq17}, and \prettyref{eq:lem3ii2}, {\sc Lemmas} \ref{lem:variance}, \ref{lem:Delta}, \ref{lem:delta}, and \ref{lem:innerproductupperbound}, with the fact that $\phi_1, \phi_2,\ldots$ is an orthonormal basis,
 \begin{center}
 \begin{tabular}{ll}
	$\int_0^1  \bphi(t)\tr \bM_0 \bphi(t) \dd t = O_P\left( \lambda_J^{-2} h^{-2}  \right)$, & $ \int_0^1  \bphi(t)\tr \bM_1 \bphi(t) \dd t = O_P\left( u_n \right)$,\\ 
$ \int_0^1  \bphi(t)\tr \bM_2 \bphi(t) \dd t = O_P\left( \lambda_J^{-1} u_n  \right)$, & $ \int_0^1  \bphi(t)\tr \bM_4 \bphi(t) \dd t = O_P\left( \lambda_J^{-1} u_n  \right)$, \\ 
$ \int_0^1  \bphi(t)\tr \bM_5 \bphi(t) \dd t = O_P\left( J \lambda_J^{-1} u_n \right)$, & 
$\int_0^1  \bphi(t)\tr \bM_6 \bphi(t) \dd t = O_P\left(  \lambda_J^{-2} u_n \right)$, 
\end{tabular}
\end{center}
and $\int_0^1  \bphi(t)\tr \bM_8 \bphi(t) \dd t = O_P\left( \lambda_J^{-1} u_n \right)$, where $u_n \coloneqq  h^{-2}  \{n  \pi_x(h)\}^{-1}$. In conclusion, $ 
\EX \left( \| \widehat{m}_x' - \EX\left( \widehat{m}_x' \right)  \|^2 \right) \,  = \, O_P\left( \lambda_J^{-2} u_n \right) + O_P\left( J \lambda_J^{-1} u_n \right)$ which is exactly the claimed rate of convergence.
\end{proof}

\begin{proof}[Proof of {\sc \prettyref{thm:ddb}}] Let us start with the conditional bias. Similarly to (\ref{eq:bias}), $\EX \left\{ \widehat{\widehat{m}}(x) \right\} = m(x) + \widehat{T}_1 + \widehat{T}_2 + \widehat{T}_3$ where $\widehat{T}_1$ (resp. $\widehat{T}_2$) corresponds to $T_1$ (resp. $T_2$) when replacing $\bPhi$ with $\widehat{\bPhi}$, and $\widehat{T}_3 \coloneqq \be\tr \,\left( \widehat{\bPhi} \tr \bK \widehat{\bPhi}  \right)^{-1} \widehat{\bPhi}  \tr \bK (\bPhi - \widehat{\bPhi}) \left[ m(x) | \nabla {m_x} \tr  \right] \tr$. 

\smallskip

{\em About $\widehat{T}_1$}. Similarly to the proof of {\sc \prettyref{lem:bias}}-$(i)$, $\widehat{T}_1 = \widehat{\mu}\left(A_0 - \widehat{\bdelta} \tr \widehat{\bDelta}^{-1} [\widehat{A}_1 \ldots \widehat{A}_J] \tr  \right)$ where $\widehat{\bdelta}$ (resp. $\widehat{\bDelta}$, $\widehat{A}_1,\ldots, \widehat{A}_J$) are defined as $\bdelta$ (resp. $\bDelta$, $A_1, \ldots,A_J$) but $\phi_1,\ldots, \phi_J$ are  replaced with $\widehat{\phi}_1,\ldots, \widehat{\phi}_J$ and where $\widehat{\mu} \coloneqq \left(\delta_0 - \widehat{\bdelta} \tr \widehat{\bDelta}^{-1} \widehat{\bdelta}\right)^{-1}$. Because $A_0 = O_P\left( \| \cP_{\cS_J^\perp} m_x' \| \, h \right)$ and $\widehat{A}_j = O_P\left(\| \cP_{\cS_J^\perp} m_x' \| \, h^2 \right)$ for $j=1,\ldots,J$,  
\begin{equation} \label{eq:hatT1.1}
\widehat{T}_1 = \widehat{\mu}  \left(1 - h \, \widehat{\bdelta} \tr  \widehat{\bDelta}^{-1}  \b1 \right) O_P\left( \| \cP_{\cS_J^\perp} m_x' \| \, h \right).
\end{equation}
According to the definitions of $\widehat{\bdelta}$ and $\widehat{\bDelta}$, it is easy to see that $\widehat{\bdelta} = \bdelta \left( 1 + U_n \right)$ and $\widehat{\bDelta}^{-1} = \bDelta^{-1} \left( 1 + U_n \right)$ with $ U_n \coloneqq O_P \left( a_J^{-1} \, n^{-1/2} \right) $ being the rate of convergence of $\max_{j \leq J} \| \phi_j - \widehat{\phi}_j \|$. On one side, $ \widehat{\bdelta} \tr \widehat{\bDelta}^{-1} \widehat{\bdelta} = \bdelta \tr \bDelta^{-1} \bdelta \left( 1 + U_n \right) $, and on the other side, $ \widehat{\bdelta} \tr  \widehat{\bDelta}^{-1}  \b1 = \bdelta \tr \bDelta^{-1} \b1 \left( 1 + U_n \right)$. According to \eqref{eq:mu}, $\bdelta \tr \bDelta^{-1} \bdelta = O_P\left(\lambda_J^{-1}\right) $ and {\sc Lemmas} \ref{lem:Delta}, \ref{lem:delta} and \ref{lem:innerproductupperbound} result in $ \widehat{\bdelta} \tr  \widehat{\bDelta}^{-1}  \b1 =  O_P\left(h^{-1} \, \lambda_J^{-1} \, J^{1/2}\right) $. Finally, $ \widehat{\mu} = O_P\left(\lambda_J\right) $,  and with  \eqref{eq:hatT1.1}
\begin{equation} \label{eq:hatT1.2}
\widehat{T}_1 \, = \, O_P\left(J^{1/2} \, \| \cP_{\cS_J^\perp} m_x' \| \, h \right).
\end{equation} 

\smallskip

{\em About $\widehat{T}_2$}. Following the guidelines in the proof of {\sc \prettyref{lem:bias}}-$(ii)$, $\widehat{T}_2 = \widehat{T}_{21} + \widehat{T}_{22}$ where $\widehat{T}_{21} \coloneqq \be\tr \,\left( \widehat{\bPhi} \tr \widetilde{\bK} \widehat{\bPhi}  \right)^{-1} \left[ B_0, \widehat{B}_1, \dots,  \widehat{B}_J \right]\tr$ and $\widehat{T}_{22} \coloneqq \be\tr \,\left( \widehat{\bPhi} \tr \widetilde{\bK} \widehat{\bPhi} \right)^{-1}$ $\left[ C_0, \widehat{C}_1, \dots, \widehat{C}_J \right]\tr$ with $\widehat{B}_1, \ldots, \widehat{B}_J$ (resp. $\widehat{C}_1,\ldots,\widehat{C}_J$) defined as $B_1,\ldots,B_J$ (resp. $C_1,\ldots,C_J$) but $\phi_1,\ldots, \phi_J$ are  replaced with $\widehat{\phi}_1,\ldots, \widehat{\phi}_J$. One has $B_0 = O_P(h^2)$, $\widehat{B}_j = O_P(h^3)$ for $j=1,\ldots,J$, and with (\prettyref{hypo:taylor}), $C_0 = O_P(h^3)$ and $\widehat{B}_j = O_P(h^4)$ for $j=1,\ldots,J$. Consequently, $\widehat{T}_{22}$ is negligible with respect to $\widehat{T}_{21}$. Similarly to $\widehat{T}_1$, $\widehat{T}_{21} =  \widehat{\mu}  \left(1 - h \, \widehat{\bdelta} \tr  \widehat{\bDelta}^{-1}  \b1 \right) O_P\left(  h^2 \right)$  so that 
\begin{equation} \label{eq:hatT2}
\widehat{T}_2 \, = \, O_P\left(J^{1/2} \, h^2 \right).
\end{equation} 

\smallskip

{\em About $\widehat{T}_3$}. By the definition of $\widehat{T}_3$ and since $\max_{j \leq J} \| \phi_j - \widehat{\phi}_j \| = O_P\left( a_J^{-1} \, n^{-1/2} \right) $, 
\begin{eqnarray*}
\widehat{T}_3 & = & \be\tr \,\left( \widehat{\bPhi} \tr \widetilde{\bK} \widehat{\bPhi}  \right)^{-1} \widehat{\bPhi}  \tr \widetilde{\bK} \left[ \delta_0 \, | \, \widehat{\bdelta} \tr \right] \tr O_P\left(a_J^{-1} \, n^{-1/2} \, J^{1/2} h\right)\\
& = & \widehat{\mu}  \left(\delta_0 - \widehat{\bdelta} \tr  \widehat{\bDelta}^{-1}  \widehat{\bdelta} \right)  O_P\left(a_J^{-1} \, n^{-1/2} \, J^{1/2} h\right). 
\end{eqnarray*} 
Since $\widehat{\mu}  \left(\delta_0 - \widehat{\bdelta} \tr  \widehat{\bDelta}^{-1}  \widehat{\bdelta} \right) = 1$, we have that 
\begin{equation} \label{eq:hatT3}
\widehat{T}_3 \, = \, O_P\left(a_J^{-1} \, n^{-1/2} \, J^{1/2} \, h^2 \right).
\end{equation} 

\smallskip

Now it is enough to combine the decomposition of $\EX \left\{\widehat{\widehat{m}}(x) \right\}$ with \eqref{eq:hatT1.2}, \eqref{eq:hatT2}, and \eqref{eq:hatT3} to get the claimed conditional bias. 

Let us now focus on the conditional variance. Similarly to  {\sc \prettyref{lem:variance}}-$(i)$,
	\[	\varX \left\{ \widehat{\widehat{m}}(x) \right\} = \{ \sigma^2(x) \, + \, o(1) \} \, \widehat{\mu}^2 \, \left( \widetilde{\delta}_0 - \widehat{\bdelta} \tr \widehat{\bDelta}^{-1} \widehat{\widetilde{\bdelta}} -  \widehat{\widetilde{\bdelta}} \tr  \widehat{\bDelta}^{-1} \widehat{\bdelta}     + \widehat{\bdelta} \tr \widehat{\bDelta}^{-1} \widehat{\widetilde{\bDelta}} \widehat{\bDelta}^{-1}  \widehat{\bdelta} \right)	\] where $\widehat{\widetilde{\bdelta}}$ (resp. $\widehat{\widetilde{\bDelta}}$) is defined as $\widetilde{\bdelta}$ (resp. $\widetilde{\bDelta}$) but $\phi_1,\ldots, \phi_J$ are  replaced with $\widehat{\phi}_1,\ldots, \widehat{\phi}_J$. From the definition of $\widehat{\widetilde{\bdelta}}$ and $\widehat{\widetilde{\bDelta}}$, it is easy to state that $\widehat{\widetilde{\bdelta}} = \widetilde{\bdelta} \left( 1 + U_n \right)$ and $\widehat{\widetilde{\bDelta}}  = \widetilde{\bDelta} \left( 1 + U_n \right)$ which results in $\varX \left\{ \widehat{\widehat{m}}(x) \right\} = \varX \left\{ \widehat{m}(x) \right\} \left(1 + U_n \right) +$ $ \{ \sigma^2(x) \, + \, o(1) \} \, \widehat{\mu}^2 \, \left( - \bdelta \tr \bDelta^{-1} \widetilde{\bdelta} -  \widetilde{\bdelta} \tr  \bDelta^{-1} \bdelta     + \bdelta \tr \bDelta^{-1} \widetilde{\bDelta} \bDelta^{-1}  \bdelta \right)  U_n$. Let us now use \mbox{{\sc Lemma~\ref{lem:variance}}-$(ii)$} with the fact that $\widehat{\mu} = O_P(\lambda_J)$. This leads to $\varX \left\{ \widehat{\widehat{m}}(x) \right\} = O_P\left( \{ n \, \pi_x(h) \}^{-1} \right) + O_P\left( a_J^{-1} \, n^{-1/2} \, \lambda_J \{ n \, \pi_x(h) \}^{-1} \right)$. Since $a_J^{-1} \, n^{-1/2} = o(1)$, the claimed conditional variance holds. 
\end{proof}
\section{Technical results}	\label{app:B}

\begin{lem}\label{lem:expectation} Let $p \geq 0$ and $q > 0$. As soon as (\prettyref{hypo:kernel}) and (\prettyref{hypo:sbp1}) are fulfilled, $$\int_0^1  t^p K^q(t) \dd P^{\|X_1-x\|/h}(t) = b_{x,p,q} \, \pi_x(h) \, \left\{ 1 + o(1) \right\},$$
where $b_{x,p,q} \coloneqq K^q(1)  - \int_0^1 \left\{u^p K^q(u)\right\}' \tau_x(u) \dd u$.
\end{lem}
A useful by-product of this lemma is the following corollary. 

\begin{cor} \label{cor:expectationkernel} Under (\prettyref{hypo:kernel}) and (\prettyref{hypo:sbp1}), we have that $\E K_1^q \, = \, b_{x,0,q} \, \pi_x(h) \, \left\{ 1 + o(1) \right\}$ for any $q > 0$.
\end{cor}

\begin{proof}[Proof of \sc \prettyref{lem:expectation}.] In order to shorten the notations, set $Z_h=\|X_1-x\|/h$. From the differentiability of $K$ it comes that $t^p K^q(t) = K^q(1) - \int_t^1 \left\{u^p K^q(u)\right\}' \dd u$ and 
	\[	
	\begin{aligned}
\int_0^1  t^p K^q(t) \dd P^{Z_h}(t) & = K^q(1) \int_0^1 \dd P^{Z_h}(t) - \int_0^1 \left( \int_t^1 \left\{u^p K^q(u)\right\}' \dd u \right)  \dd P^{Z_h}(t) \\ 
& = K^q(1) \pi_x(h) - \int_0^1 \left( \int_0^1 \left\{u^p K^q(u)\right\}' 1_{[0, u]}(t) \dd u \right) \dd P^{Z_h}(t) \\  
& = K^q(1) \pi_x(h) - \int_0^1 \left( \int_0^1 1_{[0, u]}(t) \dd P^{Z_h}(t) \right) \left\{u^p K^q(u)\right\}' \dd u \\ 
& = K^q(1) \, \pi_x(h) - \int_0^1 \left\{u^p K^q(u)\right\}' \, \pi_x(h u) \dd u.
	\end{aligned}
	\]
Finally, $\int_0^1  t^p K^q(t) \dd P^{Z_h}(t)  =  \pi_x(h) \left\{K^q(1) - \int_0^1 \left\{u^p K^q(u)\right\}' \, \left\{\pi_x(h u)/\pi_x(h) \right\} \dd u \right\}$ and thanks to (\prettyref{hypo:sbp1}) the claimed result holds.
\end{proof}

\begin{lem}\label{lem:gammafunction} Under (\prettyref{hypo:kernel})--(\prettyref{hypo:sbp1}) one has  
	\begin{enumerate}[label=(\roman*)]
	\item $ \E \left\{ K_1^a \, \gamma_{j_1,\ldots,j_K}^{p_1,\ldots,p_K}\left( \| X_1 - x \|^{p_+} \right) \right\} \, = \, b_{x, p_+, a} \, {\gamma_{j_1,\ldots,j_K}^{p_1,\ldots,p_K} }^{'}(0) \, h^{p_+} \, \pi_x(h) \, \left\{ 1 + o(1) \right\}$ for any $a\geq 0$, where $p_+ = p_1+\cdots + p_K$,
	\item $\ds \sum_{j \geq 1} \left\{ {\gamma_j^1}^{'}(0) \right\}^2 \leq 1$, $\ds \sum_{j_1, j_2 \geq 1} \left\{ {\gamma_{j_1, j_2}^{1,1}}^{'}(0) \right\}^2 \leq 1$, $\ds \sum_{j_1, j_2, j_3 \geq 1} \left\{ {\gamma_{j_1, j_2, j_3}^{1,1,1}}^{'}(0) \right\}^2 \leq 1$, $\ds \sum_{j_1, \ldots, j_4 \geq 1} \left\{ {\gamma_{j_1, \ldots, j_4}^{1,\ldots,1}}^{'}(0) \right\}^2 \leq 1$, $\ds \sum_{j \geq 1} {\gamma_j^2}^{'}(0)  \leq 1$, and $\ds \sum_{j_1, j_2 \geq 1} {\gamma_{j_1, j_2}^{2,2}}^{'}(0)  \leq 1$,
	\item the $J\times J$ matrix $\boldsymbol{ \Gamma }$  is positive semi-definite.
	\end{enumerate}
\end{lem}
An interesting by-product of this lemma indicates that the Frobrenius norm of $\bGamma$ is bounded from above by 1, and so is its largest eigenvalue.
\begin{proof}[Proof of \sc \prettyref{lem:gammafunction}-$(i)$.] We can write
\[
\begin{aligned}
\E \left\{ K_1^a \, \gamma_{j_1,\ldots,j_K}^{p_1,\ldots,p_K}\left( \| X_1 - x \|^{p_+} \right) \right\}  
& = \int_0^h \gamma_{j_1,\ldots,j_K}^{p_1,\ldots,p_K} \left( t^{p_+} \right) K^a\left(h^{-1} t\right) \dd P^{\| X_1 - x \|}(t) \\ 
& = \int_0^1 \gamma_{j_1,\ldots,j_K}^{p_1,\ldots,p_K} \left( h^{p_+} t^{p_+} \right) K^a( t ) \dd P^{Z_h}(t) \\ 
& = \int_0^1 (h \, t)^{p_+}   {\gamma_{j_1,\ldots,j_K}^{p_1,\ldots,p_K}}^{'} (0) \left\{1 + \eta(\epsilon) \right\} K^a( t ) \dd P^{Z_h}(t),
\end{aligned}
\]
where $\eta(\epsilon) = {{\gamma_{j_1,\ldots,j_K}^{p_1,\ldots,p_K}}^{'} (\epsilon) }/ { {\gamma_{j_1,\ldots,j_K}^{p_1,\ldots,p_K}}^{'} (0) } - 1$ with $0< \epsilon < h^{p_+} t^{p_+}$. Because $\gamma_{j_1,\ldots,j_K}^{p_1,\ldots,p_K}$ is continuously differentiable, $\sup_\epsilon |\eta(\epsilon)| = o(1)$, which implies 
\begin{equation*} 
\E \left\{ K_1^a \, \gamma_{j_1,\ldots,j_K}^{p_1,\ldots,p_K} \left( \| X_1 - x \|^{p_+} \right) \right\} =  {\gamma_{j_1,\ldots,j_K}^{p_1,\ldots,p_K}}^{'} (0) \, h^{p_+} \int_0^1  t^{p_+} K^a( t ) \dd P^{Z_h}(t) \, \left\{ 1 + o(1) \right\}.
\end{equation*}
Use of {\sc \prettyref{lem:expectation}} gives in the claimed result.
\end{proof}

\begin{proof}[Proof of {\sc \prettyref{lem:gammafunction}}-$(ii)$.] Let us first remark that  $$\sum_{j \geq 1} \left\{ \gamma_j^1(t) \right\}^2 \ \leq \ \sum_{j \geq 1} \E\left( \langle \phi_j, X_1 - x \rangle^2  \, | \, \| X_1 - x  \| = t \right)$$ by using the definition of $\gamma_j^1(t)$. Because $\sum_{j \geq 1} \langle \phi_j, X_1 - x \rangle^2 = \| X_1 - x  \|^2$, \linebreak $\sum_{j \geq 1} \left\{ \gamma_j^1(t) \right\}^2 \leq \E\left( \| X_1 - x  \|^2  \, | \, \| X_1 - x  \| = t \right) = t^2$. Now, by the definition of the derivative, 
\begin{equation} \label{eq:propgammaproperties1}
\sum_{j \geq 1} \left\{ {\gamma_j^1}^{'}(0) \right\}^2 \, = \,  \sum_{j\geq 1} \, \lim_{t\rightarrow 0} \, \frac{ \left\{ \gamma_j^1(t) \right\}^2 }{t^2}.
\end{equation}
So, for any $d \geq 1$, $\left| S_d(t) - S_\infty(t) \right| \leq \E\left( \zeta_d  \, | \, \| X_1 - x  \| = t \right)$ where we write $S_d(t) \coloneqq   \sum_{j=1}^d \left\{ \gamma_j^1(t) \right\}^2$ and $\zeta_d \coloneqq \sum_{j >d} \, \langle \phi_j, X_1 - x \rangle^2$. It is clear that $\left\{ \zeta_d \right\}_d$ is a non-increasing sequence of random variables that converges almost surely to zero with $d \to \infty$. The monotone convergence theorem implies that for any $t$, the sequence $\E\left( \zeta_d \, | \, \| X_1 - x  \| = t \right)$ is also a non-increasing sequence of random variables converging almost surely to zero with $d$. Consequently, $\sum_{j \geq 1} t^{-2}\left\{ \gamma_j^1(t) \right\}^2$ converges uniformly on $(0, h)$. By remarking that for any $j$, $\lim_{t\rightarrow 0} t^{-2}\left\{ \gamma_j^1(t) \right\}^2 = \left\{ {\gamma_j^1}^{'}(0) \right\}^2$, one can exchange the limit with the sum in \prettyref{eq:propgammaproperties1}
	\[	\sum_{j\geq 1} \left\{ {\gamma_j^1}^{'}(0) \right\}^2 = \lim_{t\rightarrow 0} \sum_{j\geq 1}  \frac{\left\{ \gamma_j^1(t) \right\}^2 }{ t^2 }
 \leq  \lim_{t\rightarrow 0} \frac{1}{t^2} \sum_{j\geq 1}  \E \left( \langle \phi_j, X_1 - x \rangle^2 \, | \, \| X_1 - x  \| = t  \right).	\] 
The expression $\sum_{j > d}  \langle \phi_j, X_1 - x \rangle^2  = \|  \cP_{\cS_d^\perp} (X_1 - x) \|^2 $ converges almost surely to 0 when $d$ tends to infinity, where $\E \left( \|  \cP_{\cS_d^\perp} (X_1 - x) \|^2  \, | \, \| X_1 - x  \| = t  \right) \leq t^2$. Thanks to \cite{Bour06} (see Corollary 2 - INT IV.37), one can put the infinite summation into the expectation 
\begin{eqnarray*} 
 \sum_{j\geq 1} \left\{ {\gamma_j^1}^{'}(0) \right\}^2 & \leq & \lim_{t\rightarrow 0} \, \frac{1}{t^2} \, \E \left(  \sum_{j\geq 1} \langle \phi_j, X_1 - x \rangle^2  \, | \, \| X_1 - x  \| = t  \right)\\ 
& \leq & \lim_{t\rightarrow 0} \, \frac{1}{t^2} \, \E \left( \| X_1 - x  \|^2  \, | \, \| X_1 - x  \| = t  \right) \leq 1,
\end{eqnarray*}
which corresponds to the claimed first assertion. The other ones can be obtained by using similar arguments. 
\end{proof}

\begin{proof}[Proof of \sc \prettyref{lem:gammafunction}-$(iii)$.] For any $J$-dimensional vector $\bu$ and $t \geq 0$, 
	\[	
	\begin{aligned}
\bu\tr \, \boldsymbol{\Gamma} \, \bu & = \sum_{j,k =1}^J u_j u_k {\gamma_{j,k}^{1,1}}^{'} (0) = \lim_{t \rightarrow 0} \, t^{-1} \sum_{j,k =1}^J u_j u_k \gamma_{j,k}^{1,1}(t) \\ 
& = \lim_{t \rightarrow 0} \, t^{-1} \E\left\{ \sum_{j,k =1}^J u_j u_k  \langle \phi_j, X_1 - x \rangle \langle \phi_k, X_1 - x \rangle \, | \, \| X_1 - x \|^2 = t \right\} \\ 
& = \lim_{t \rightarrow 0} \, t^{-1} \E\left\{ \left( \sum_{j =1}^J u_j \langle \phi_j, X_1 - x \rangle \right)^2 \, | \, \| X_1 - x \|^2 = t \right\} \geq 0. 
	\end{aligned}
	\]
\end{proof}

\begin{lem}\label{lem:Delta} Under (\prettyref{hypo:kernel})--(\prettyref{hypo:sbp1}), one has
\[
\bDelta = b_{x,0,1}^{-1} \,  b_{x,2,1} \, h^2 \, \boldsymbol{ \Gamma } \, \left\{ 1 + o_P(1) \right \}  + \, O_P\left( \frac{h^2}{\sqrt{n \pi_x(h)}} \right) \boldsymbol{ \Lambda },
\]
where $ \boldsymbol{  \Lambda }$ is the $J\times J$ matrix such that $\left[ \boldsymbol{  \Lambda } \right]_{j,k} \coloneqq \sqrt{ {\gamma_{j,k}^{2,2}}^{'} (0) }$.
\end{lem}

\begin{proof}[Proof of \sc \prettyref{lem:Delta}.] From the definition of $\Delta_{jk}$, 
	\[	
	\begin{aligned}
\E \Delta_{jk} & = (\E K_1)^{-1} \, \E \left( K_1 \langle \phi_j, X_1 - x \rangle   \langle \phi_k, X_1 - x  \rangle \right) \\
& = (\E K_1)^{-1} \, \E \left\{ K_1 \, \gamma_{j,k}^{1,1}\left( \| X_1 - x \|^2 \right)  \right\}.
	\end{aligned}
	\]
Use of {\sc \prettyref{lem:gammafunction}} and {\sc Corollary~\ref{cor:expectationkernel}} results in 
\begin{equation} \label{eq:Delta1} 
\E \Delta_{jk} \, = \, b_{x,0,1}^{-1} \, b_{x,2,1} \, {\gamma_{j,k}^{1,1}}^{'} (0) \, h^2 \, \left\{ 1 + o(1) \right\}. 
\end{equation}
One can also derive the asymptotic behavior of $\var\left( \Delta_{jk} \right)$
\begin{equation} \label{eq:Delta2} 
\begin{aligned}
\var\left( \Delta_{jk} \right) & = n^{-1} (\E K_1)^{-2} \var\left( K_1 \langle \phi_j, X_1 - x \rangle   \langle \phi_k, X_1 - x  \rangle \right) \\ 
& = n^{-1} (\E K_1)^{-2} \, \E \left\{ K_1^2 \, \gamma_{j,k}^{2,2}\left( \| X_1 - x \|^4 \right) \right\} \\ 
& \phantom{=} - \, n^{-1} (\E K_1)^{-2} \, \left(\E \left\{ K_1 \, \gamma_{j,k}^{1,1}\left( \| X_1 - x \|^2 \right) \right\} \right)^2 \\ 
& = O\left( h^4 \{ n \, \pi_x(h) \}^{-1} \right) {\gamma_{j,k}^{2,2}}^{'}(0) \, + \, o\left( \left\{ {\gamma_{j,k}^{1,1}}^{'} (0) \right\}^2 \, h^4 \right),
\end{aligned}
\end{equation}
where the last equality comes from {\sc \prettyref{lem:gammafunction}} and {\sc Corollary~\ref{cor:expectationkernel}}. 
Equations \eqref{eq:Delta1} and \eqref{eq:Delta2} result in 
$$
\Delta_{jk}  \, = \, b_{x,0,1}^{-1} \, b_{x,2,1} \, {\gamma_{j,k}^{1,1}}^{'} (0) \, h^2 \, \left\{ 1 + o_P(1) \right\} \, + \, O_P\left( h^2 \{ n \, \pi_x(h) \}^{-1/2}   \right) \sqrt{ {\gamma_{j,k}^{2,2}}^{'}(0) },
$$
which is the element-wise version of the claimed result. 
\end{proof}

\begin{lem}\label{lem:delta} As soon as (\prettyref{hypo:kernel})--(\prettyref{hypo:sbp1}) are fulfilled, 
\begin{eqnarray*}
\bdelta  &= & b_{x,0,1}^{-1} \, b_{x,1,1} \, \bgamma \, h  \, \left\{ 1 + o_P(1) \right\}\, + \, O_P\left( h \, \{ n \, \pi_x(h) \}^{-1/2}  \right) \btheta,
\end{eqnarray*}
where $\btheta$ is the $J$-dimensional vector such that $\left[ \btheta \right]_j \coloneqq \sqrt{ {\gamma_j^2}^{'} (0) }$.
\end{lem}

\begin{proof}[Proof of \sc \prettyref{lem:delta}.] This proof is shortened since it follows the same lines as the previous one. From the definition of $\delta_{j}$, 
$\E \delta_{j} \, = \, (\E K_1)^{-1} \, \E \left\{ K_1 \, \gamma_j^1\left( \| X - x \| \right)  \right\}$ and use of {\sc \prettyref{lem:gammafunction}} and {\sc Corollary~\ref{cor:expectationkernel}} results in 
\begin{equation} \label{eq:delta1} 
\E \delta_{j} \, = \,  b_{x,0,1}^{-1} \, b_{x,1,1} \,  {\gamma_j^1}^{'} (0) \, h \, \left\{ 1 + o(1) \right\}. 
\end{equation}
Let us now focus on the variance of the $\delta_j$. 
	\[	
	\begin{aligned}
	\var\left( \delta_{j} \right) & = n^{-1} (\E K_1)^{-2} \var \left( K_1 \langle \phi_j, X_1 - x \rangle  \right) \\
& = n^{-1} (\E K_1)^{-2} \, \E \left\{ K_1^2 \, \gamma_j^2\left( \| X - x \|^2 \right) \right\} \\
& \phantom{=} - \, n^{-1} (\E K_1)^{-2} \, \left(\E \left\{ K_1 \, \gamma_j^1\left( \| X - x \| \right) \right\} \right)^2 \\ 
& = O\left( h^2 \, \{ n \, \pi_x(h) \}^{-1} \right) {\gamma_j^2}^{'}(0) \, + \, o\left( \left\{ {\gamma_j^1}^{'} (0) \right\}^2 \, h^2 \right),
\end{aligned}
\]
the last equality resulting from {\sc \prettyref{lem:gammafunction}} and {\sc Corollary~\ref{cor:expectationkernel}}. By combining this last equation with \prettyref{eq:delta1},
$$
\delta_{j} \, = \, b_{x,0,1}^{-1} \, b_{x,1,1} \, {\gamma_j^1}^{'} (0) \, h \,  \left\{ 1 + o_P(1) \right\} \, + \, O_P\left( h \, \{ n \, \pi_x(h) \}^{-1/2} \right) \sqrt{ {\gamma_j^2}^{'}(0) },
$$
which is the element-wise version of the claimed result. 
\end{proof}

\begin{lem}\label{lem:A0Aj} Under (\prettyref{hypo:kernel})--(\prettyref{hypo:sbp1}), one has
	\begin{enumerate}[label=(\roman*)]
	\item $\ds A_0 \, = \, b_{x,0,1} ^{-1} \, b_{x,1,1} \,  \alpha_{0,x,n}^{bias}  \, h \, \left\{ 1 + o(1) \right\} \, + \, O_P\left( \| \cP_{\cS_J^\perp} m_x' \| \, h \, \{ n \, \pi_x(h) \}^{-1/2}  \right)$, \\ where the sequence  $\alpha_{0,x,n}^{bias}$ is upper bounded by $ \| \cP_{\cS_J^\perp} m_x' \|$,
	\item $\ds A_j \, = \, b_{x,0,1} ^{-1} \, b_{x,2,1} \,  \alpha_{j,x,n}^{bias}  \, h^2 \, \left\{ 1 + o(1) \right\} \, + \, O_P\left( h^2 \sqrt{ \alpha_{j,x,n}^{var} } \{ n \, \pi_x(h) \}^{-1/2}  \right)$  for \linebreak  $j=1,\ldots,J$, where $ \sum_{j=1}^J \left\{ \alpha_{j,x,n}^{bias} \right\}^2 \leq \| \cP_{\cS_J^\perp} m_x' \|^2$ and $ \sum_{j=1}^J \alpha_{j,x,n}^{var} \leq  \| \cP_{\cS_J^\perp} m_x' \|^2 $.
	\end{enumerate}
\end{lem}

\begin{proof}[Proof of \sc \prettyref{lem:A0Aj}-$(i)$.] We can write
	\[	
	\begin{aligned}
\E A_0 & = (\E K_1)^{-1} \, \E \left( K_1 \langle \cP_{\cS_J^\perp} m_x', X_1-x \rangle \right) \\ 
& = (\E K_1)^{-1} \, \sum_{j > J} \langle \phi_j, m_x' \rangle \E \left( K_1 \langle\phi_j, X_1-x \rangle \right) \\
& = (\E K_1)^{-1} \, \sum_{j > J} \langle \phi_j, m_x' \rangle \E \left\{ K_1 \,  \gamma_j^1\left( \|X_1-x\| \right) \right\} \\
& = b_{x,0,1} ^{-1} \, b_{x,1,1} \,  \alpha_{0,x,n}^{bias}  \, h \, \left\{ 1 + o(1) \right\},
	\end{aligned}
	\]
where $ \alpha_{0,x,n}^{bias} = \sum_{j > J} \langle \phi_j, m_x' \rangle \,  {\gamma_j^1}^{'}(0)$, the last equality coming from  \mbox{{\sc \prettyref{lem:gammafunction}}-$(i)$}. Moreover, the Cauchy-Schwartz inequality and {\sc \prettyref{lem:gammafunction}}-$(ii)$ imply that $ | \alpha_{0,x,n}^{bias} | \leq  \| \cP_{\cS_J^\perp} m_x' \|$. In the same way, one has
	\[	
	\begin{aligned}
\var (A_0) & = n^{-1} \, (\E K_1)^{-2} \, \var \left\{ K_1  \sum_{j > J} \langle \phi_j, m_x' \rangle  \langle\phi_j, X_1-x \rangle \right\} \\
& \leq n^{-1} \, (\E K_1)^{-2} \sum_{j > J} \sum_{k > J} \langle \phi_j, m_x' \rangle  \langle \phi_k, m_x' \rangle \E \left\{ K_1^2 \gamma_{j,k}^{1,1}\left( \|X_1-x\|^2 \right) \right\} \\
& = O\left( \alpha_{0,x,n}^{var} \, \frac{h^2}{n \, \pi_x(h)} \right),
\end{aligned}
\]
where $ \alpha_{0,x,n}^{var} = \sum_{j > J} \sum_{k > J} \langle \phi_j, m_x' \rangle \langle \phi_k, m_x' \rangle \, {\gamma_{j,k}^{1,1}}^{'}(0)$, the last equality using again {\sc \prettyref{lem:gammafunction}}-$(i)$. Moreover, 
$$
\alpha_{0,x,n}^{var} \, \leq  \, \left\{  \sum_{j > J} \sum_{k > J} \langle \phi_j, m_x' \rangle^2 \langle \phi_k, m_x' \rangle^2 \right\}^{1/2} \left\{  \sum_{j > J} \sum_{k > J}  \left\{ {\gamma_{j,k}^{1,1}}^{'}(0)\right\}^2 \right\}^{1/2},
$$
which results in $ \alpha_{0,x,n}^{var} \leq  \| \cP_{\cS_J^\perp} m_x' \|^2 $ (use again {\sc \prettyref{lem:gammafunction}}-$(ii)$). Then $\var(A_0) \, = \, O\left( \| \cP_{\cS_J^\perp} m_x' \|^2 \, h^2 \{n \, \pi_x(h)\}^{-1} \right)$ and
the claimed result holds.
\end{proof}

\begin{proof}[Proof of \sc \prettyref{lem:A0Aj}-$(ii)$.] We have
\begin{equation} \label{eq2:A0Aj}
\begin{aligned}
\E A_j & = (\E K_1)^{-1} \, \E \left( K_1 \langle \cP_{\cS_J^\perp} m_x', X_1-x \rangle \langle\phi_j, X_1-x \rangle \right)\\ 
& = (\E K_1)^{-1} \, \sum_{k > J} \langle \phi_k, m_x' \rangle \E \left( K_1 \langle\phi_k, X_1-x \rangle \langle\phi_j, X_1-x \rangle \right) \\
& = (\E K_1)^{-1} \, \sum_{k > J} \langle \phi_k, m_x' \rangle \E \left\{ K_1 \,  \gamma_{j,k}^{1,1}\left( \|X_1-x\|^2 \right) \right\} \\ 
& = b_{x,0,1} ^{-1} \, b_{x,2,1} \,  \alpha_{j,x,n}^{bias}  \, h^2 \, \left\{ 1 + o(1) \right\},
\end{aligned}
\end{equation}
where $ \alpha_{j,x,n}^{bias} = \sum_{k > J} \langle \phi_k, m_x' \rangle \,  {\gamma_{j,k}^{1,1}}^{'}(0)$, the last equality coming from the use of  {\sc \prettyref{lem:gammafunction}}-$(i)$. In addition, 
	\[	
	\begin{aligned}
 \sum_{j=1}^J \left\{\alpha_{j,x,n}^{bias}\right\}^2 & = \sum_{j=1}^J  \sum_{k,\ell > J} \langle \phi_k, m_x' \rangle \langle \phi_\ell, m_x' \rangle \, {\gamma_{j,k}^{1,1}}^{'}(0) \, {\gamma_{j,\ell}^{1,1}}^{'}(0) \\
& \leq \left\{ \sum_{k,\ell > J} \langle \phi_k, m_x' \rangle^2 \langle \phi_\ell, m_x' \rangle^2 \right\}^{1/2} \left\{  \sum_{k,\ell > J} \left( \sum_{j=1}^J {\gamma_{j,k}^{1,1}}^{'}(0) \, {\gamma_{j,\ell}^{1,1}}^{'}(0) \right)^2 \right\}^{1/2}\\
& \leq \| \cP_{\cS_J^\perp} m_x' \|^2 \left\{ \left(\sum_{j,k \geq 1} \left[{\gamma_{j,k}^{1,1}}^{'}(0)\right]^2 \right) \left(\sum_{j,\ell \geq 1} \left[ {\gamma_{j,\ell}^{1,1}}^{'}(0) \right]^2 \right)  \right\}^{1/2}\\ 
& \leq \| \cP_{\cS_J^\perp} m_x' \|^2.
\end{aligned}
\]
To derive the asymptotic behavior of the variance of $A_j$, we follow similar arguments
\begin{equation} \label{eq3:A0Aj}
\begin{aligned}
\var (A_j) & = n^{-1} \, (\E K_1)^{-2} \, \var \left\{ K_1  \left( \sum_{k > J} \langle \phi_k, m_x' \rangle  \langle\phi_k, X_1-x \rangle \right) \langle\phi_j, X_1-x \rangle \right\} \\
& \leq n^{-1} \, (\E K_1)^{-2} \sum_{k > J} \sum_{\ell > J} \langle \phi_k, m_x' \rangle  \langle \phi_\ell, m_x' \rangle \E \left\{ K_1^2 \gamma_{j,k,\ell}^{2,1,1}\left( \|X_1-x\|^4 \right) \right\} \\ 
& = O\left( \alpha_{j,x,n}^{var} \, h^4 \, \{n \, \pi_x(h)\}^{-1} \right),
\end{aligned}
\end{equation}
where $ \alpha_{j,x,n}^{var} = \sum_{k > J} \sum_{\ell > J} \langle \phi_k, m_x' \rangle \langle \phi_\ell, m_x' \rangle \, {\gamma_{x,j,k,\ell}^{2,1,1}}^{'}(0)$. Now, by involving arguments similar to those used in {\sc \prettyref{lem:gammafunction}}-$(ii)$ one is able to show $ \sum_{j=1}^J \alpha_{j,x,n}^{var} \leq  \| \cP_{\cS_J^\perp} m_x' \|^2 $. Just combine \eqref{eq2:A0Aj} and \eqref{eq3:A0Aj} to get the claimed result.
\end{proof}

\begin{lem}\label{lem:innerproductupperbound} For any $\bu, \bv \in \R^J$, $\left| \bu\tr  \, \boldsymbol{ \Gamma }^{-1} \, \bv \right| \, \leq \, \lambda_J^{-1} \, \| \bu \|_2 \, \| \bv \|_2$, where  $\lambda_J$ is the smallest eigenvalue of $ \boldsymbol{ \Gamma }$.
\end{lem}

\begin{proof} Let us remark that $\lambda_J >0$ according to (\prettyref{hypo:gammafunction}). This result involves the Cauchy-Schwartz inequality and the Rayleigh quotient of the inverse of the $J\times J$ matrix $ \boldsymbol{ \Gamma } $:
\begin{eqnarray*}
\left| \bu\tr \, \boldsymbol{ \Gamma }^{-1} \, \bv \right| & \leq &  \| \bu \|_{  \boldsymbol{ \Gamma }^{-1} } \, \| \bv \|_{  \boldsymbol{ \Gamma }^{-1} }\\ 
& \leq & \| \bu \|_2 \, \| \bv \|_2 \, \left( \frac{\bu\tr \, \boldsymbol{ \Gamma }^{-1} \, \bu}{\bu\tr \bu} \right)^{1/2}  \, \left( \frac{\bv\tr \, \boldsymbol{ \Gamma }^{-1} \, \bv}{\bv\tr \bv} \right)^{1/2} \\
& \leq & \| \bu \|_2 \, \| \bv \|_2 \, R_1^{1/2} \, R_2^{1/2},
\end{eqnarray*}
where $R_1$ and $R_2$ are the Rayleigh quotients of $ \boldsymbol{ \Gamma }^{-1}$. Let $\lambda_J$ be the smallest eigenvalue of $\boldsymbol{ \Gamma }$. Then $\lambda_J^{-1}$ is the greatest eigenvalue of $ \boldsymbol{ \Gamma }^{-1}$ with $R_1 \leq \lambda_J^{-1}$ and $R_2 \leq \lambda_J^{-1}$, and the claimed result holds.
\end{proof}

\begin{lem}\label{lem:B0Bj} Under (\prettyref{hypo:kernel})--(\prettyref{hypo:sbp1}), one has
	\begin{enumerate}[label=(\roman*)]
	\item $\ds B_0 \, = \, b_{x,0,1} ^{-1}  b_{x,2,1}   \beta_{0,x}^{bias}  \, h^2   \left\{ 1 + o(1) \right\} \, + \, O_P\left( h^2 \, \{n \, \pi_x(h)\}^{-1/2}   \right)$ with $\beta_{0,x}^{bias} = O(1)$,
	\item $\ds B_j \, = \, b_{x,0,1} ^{-1} \, b_{x,3,1} \,  \beta_{j,x}^{bias}  \, h^3 \, \left\{ 1 + o(1) \right\} \, + \, \sqrt{\beta_{j,x}^{var}} \, O_P\left( h^3 \, \{ n \, \pi_x(h) \}^{-1/2}  \right)$  for all $j=1,\ldots,J$, where $\sum_{j=1}^J \left\{\beta_{j,x}^{bias}\right\}^2 = O(1)$ and $\sum_{j=1}^J \beta_{j,x}^{var} = O(1)$.
	\end{enumerate}
\end{lem}

\begin{proof}[Proof of \sc \prettyref{lem:B0Bj}-$(i)$.] A standard expansion of the linear operator $m_x''$ results in
\[
\begin{aligned}
\E B_0 & = (\E K_1)^{-1} \, \sum_{j \geq 1} \sum_{k \geq 1} \langle m_x'' \phi_j, \phi_k \rangle \E \left( K_1 \langle\phi_j, X_1-x \rangle \langle\phi_k, X_1-x \rangle \right) \\
& = (\E K_1)^{-1} \, \sum_{j \geq 1} \sum_{k \geq 1}  \langle m_x'' \phi_j, \phi_k \rangle \E \left\{ K_1 \,  \gamma_{j,k}^{1,1}\left( \|X_1-x\|^2 \right) \right\} \\
& = b_{x,0,1} ^{-1} \, b_{x,2,1} \,  \beta_{0,x}^{bias}  \, h^2 \, \left\{ 1 + o(1) \right\},
\end{aligned}
\]
where $ \beta_{0,x}^{bias} \coloneqq\sum_{j \geq 1} \sum_{k \geq 1}  \langle m_x'' \phi_j, \phi_k \rangle \,  {\gamma_{j,k}^{1,1}}^{'}(0)$, the last equality by \mbox{{\sc \prettyref{lem:gammafunction}}-$(i)$}. Moreover, the Cauchy-Schwartz inequality combined with (\prettyref{hypo:taylor}) and {\sc \prettyref{lem:gammafunction}}-$(ii)$ imply that $\beta_{0,x}^{bias} = O(1)$. In the same way, one has
\[
\begin{aligned}
\var (B_0) & = n^{-1} \, (\E K_1)^{-2} \, \var \left\{ K_1  \langle m_x''(X_1 - x), X_1 - x \rangle  \right\} \\
& \leq n^{-1} \, (\E K_1)^{-2} \sum_{j_1,\ldots,j_4 \geq 1} \langle m_x'' \phi_{j_1}, \phi_{j_2} \rangle  \langle m_x'' \phi_{j_3}, \phi_{j_4} \rangle \E \left\{ K_1^2 \gamma_{j_1,\ldots,j_4}^{1,1,1,1}\left( \|X_1-x\|^4 \right) \right\} \\
& = O\left( \beta_{0,x}^{var} \, h^4 \, \{n \, \pi_x(h)\}^{-1} \right),
\end{aligned}
\]
where $ \beta_{0,x}^{var} \coloneqq \sum_{j_1,\ldots,j_4 \geq 1} \langle m_x'' \phi_{j_1}, \phi_{j_2} \rangle  \langle m_x'' \phi_{j_3}, \phi_{j_4} \rangle \, {\gamma_{j_1,\ldots,j_4}^{1,1,1,1}}^{'}(0)$, the last equality using again {\sc \prettyref{lem:gammafunction}}-$(i)$. Moreover, 
$$  
\beta_{0,x}^{var} \, \leq  \, \left\{  \sum_{j_1,\ldots,j_4 \geq 1} \langle m_x'' \phi_{j_1}, \phi_{j_2} \rangle^2  \langle m_x'' \phi_{j_3}, \phi_{j_4} \rangle^2 \right\}^{1/2} \left\{   \sum_{j_1,\ldots,j_4 \geq 1}  \left\{ {\gamma_{j_1,\ldots,j_4}^{1,1,1,1}}^{'}(0)\right\}^2 \right\}^{1/2},
$$
which gives that $\beta_{0,x}^{var}$ is a finite quantity (use again (\prettyref{hypo:taylor}) and {\sc \prettyref{lem:gammafunction}}-$(ii)$). Then $\var(B_0) \, = \, O\left( h^4 \{n \, \pi_x(h)\}^{-1} \right)$ and
the claimed result holds.
\end{proof}

\begin{proof}[Proof of \sc \prettyref{lem:B0Bj}-$(ii)$.] Similarly, 

\[
\begin{aligned}
\E B_j & = (\E K_1)^{-1} \, \sum_{k \geq 1} \sum_{\ell \geq 1} \langle m_x'' \phi_k, \phi_\ell \rangle \E \left( K_1 \langle \phi_j, X_1-x \rangle \langle \phi_k, X_1-x \rangle \langle \phi_\ell, X_1-x \rangle \right) \\
& = (\E K_1)^{-1} \, \sum_{k \geq 1} \sum_{\ell \geq 1}  \langle m_x'' \phi_k, \phi_\ell \rangle \E \left\{ K_1 \,  \gamma_{j,k,\ell}^{1,1,1}\left( \|X_1-x\|^3 \right) \right\} \\
& = b_{x,0,1} ^{-1} \, b_{x,3,1} \,  \beta_{j,x}^{bias}  \, h^3 \, \left\{ 1 + o(1) \right\},
\end{aligned}
\]
where $ \beta_{j,x}^{bias} \coloneqq\sum_{k \geq 1} \sum_{\ell \geq 1}  \langle m_x'' \phi_k, \phi_\ell \rangle \,  {\gamma_{j,k,\ell}^{1,1,1}}^{'}(0)$. Moreover, 
\[
\begin{aligned}
\sum_{j=1}^J \left\{ \beta_{j,x}^{bias} \right\}^2 & = \sum_{j=1}^J  \sum_{k,\ell,p,q \geq 1} \langle m_x'' \phi_k, \phi_\ell \rangle \langle m_x'' \phi_p, \phi_q \rangle {\gamma_{j,k,\ell}^{1,1,1}}^{'}(0) \, {\gamma_{j,p,q}^{1,1,1}}^{'}(0) \\
& \leq \left\{ \sum_{k,\ell,p,q \geq 1} \langle m_x'' \phi_k, \phi_\ell \rangle^2 \langle m_x'' \phi_p, \phi_q \rangle^2 \right\}^{1/2} \\
& \phantom{=} \times  \left\{  \sum_{k,\ell,p,q \geq 1} \left( \sum_{j=1}^J {\gamma_{j,k,\ell}^{1,1,1}}^{'}(0) \, {\gamma_{j,p,q}^{1,1,1}}^{'}(0) \right)^2 \right\}^{1/2}\\
& \leq \| m_x'' \|_{HS}^2 \left\{ \left(\sum_{j,k,\ell \geq 1} \left[{\gamma_{j,k,\ell}^{1,1,1}}^{'}(0)\right]^2 \right) \left(\sum_{j,p,q \geq 1} \left[ {\gamma_{j,p,q}^{1,1,1}}^{'}(0) \right]^2 \right)  \right\}^{1/2}\\ 
& \leq \| m_x'' \|_{HS}^2,
\end{aligned}
\]
where $\left\Vert \cdot \right\Vert_{HS}$ denotes the Hilbert-Schmidt norm of an operator. So, $\sum_{j=1}^J \left\{ \beta_{j,x}^{bias} \right\}^2 = O(1)$. Let us  now focus on the variance of $B_j$
\[
\begin{aligned}
\var (B_j) & = n^{-1} \, (\E K_1)^{-2} \, \var \left\{ K_1  \langle m_x''(X_1 - x), X_1 - x \rangle \langle \phi_j, X_1-x \rangle \right\} \\
& \leq n^{-1} \, (\E K_1)^{-2} \sum_{k,\ell,p,q \geq 1} \langle m_x'' \phi_k, \phi_\ell \rangle  \langle m_x'' \phi_p, \phi_q \rangle \E \left\{ K_1^2 \gamma_{j,k,\ell,p,q}^{2,1,1,1,1}\left( \|X_1-x\|^6 \right) \right\} \\
& = O\left( \beta_{j,x}^{var} \, \frac{h^6}{n \, \pi_x(h)} \right),
\end{aligned}
\]
where $ \beta_{j,x}^{var} \coloneqq \sum_{k,\ell,p,q \geq 1} \langle m_x'' \phi_k, \phi_\ell \rangle  \langle m_x'' \phi_p, \phi_q \rangle \, {\gamma_{j,k,\ell,p,q}^{2,1,1,1,1}}^{'}(0)$, the last equality using again {\sc \prettyref{lem:gammafunction}}-$(i)$. Using similar arguments as those involved to show the second assertion of {\sc \prettyref{lem:gammafunction}}-$(ii)$,
\[
\begin{aligned}
\sum_{j=1}^J \beta_{j,x}^{var} & \leq \sum_{j,k,\ell,p,q \geq 1} \langle m_x'' \phi_k, \phi_\ell \rangle  \langle m_x'' \phi_p, \phi_q \rangle \left\{ \lim_{t \rightarrow 0} \, t^{-1} \, \gamma_{j,k,\ell,p,q}^{2,1,1,1,1}(t) \right\} \\ 
& = \sum_{k,\ell,p,q \geq 1} \langle m_x'' \phi_k, \phi_\ell \rangle  \langle m_x'' \phi_p, \phi_q \rangle \left\{ \lim_{t \rightarrow 0} \, t^{-2/3} \, \gamma_{k,\ell,p,q}^{1,1,1,1}(t) \right\}\\ 
& \leq \| m_x'' \|_{HS}^2 \lim_{t \rightarrow 0} \, t^{-2/3} \, \left\{  \sum_{k,\ell,p,q \geq 1}   \left[\gamma_{k,\ell,p,q}^{1,1,1,1}(t)\right]^2 \right\}^{1/2}.
\end{aligned}
\]
Consequently, 
\[
\begin{aligned}
\sum_{j=1}^J \beta_{j,x}^{var} 
& \leq \| m_x'' \|_{HS}^2 \lim_{t \rightarrow 0} \, t^{-2/3} \left\{ \sum_{k,\ell,p,q \geq 1} \gamma_{k,\ell,p,q}^{2,2,2,2}\left( t^{4/3} \right)  
\right\}^{1/2} \leq \| m_x'' \|_{HS}^2,
\end{aligned}
\]
where, of course, $\gamma_{k,\ell,p,q}^{2,2,2,2}\left( t^{4/3} \right)$ is the same as 
	\[	\E\left(\langle \phi_k, X_1 - x \rangle^2  \langle \phi_\ell, X_1 - x \rangle^2 \langle \phi_p, X_1 - x \rangle^2  \langle \phi_q, X_1 - x \rangle^2 | \|X_1- x\|^6 = t  \right).	\]
\end{proof}

\section{Practical aspects} \label{app:C}

\subsection{{\em fllr} R package}
An efficient, fully documented \texttt{R} implementation of all the considered estimating procedures, including the automated selection of all their parameters, is freely available as a part of \texttt{R} package \texttt{fllr}. The package can be downloaded from \url{https://bitbucket.org/StanislavNagy/fllr}. Using the procedures from \texttt{fllr} and the source codes accompanying the present manuscript, all the results from the main paper, and the simulation studies presented in the supplementary material, can be replicated in full.

\subsection{Bandwidth selection}
{\sc Functional derivative: bootstrap bandwidth selection.} The bootstrap procedure introduces additional randomness into our local linear estimation method. A natural question raises: is the proposed bandwidth selection stable? In other words, if the functional derivatives are estimated several times on the same dataset, are the resulting bandwidths similar? A related issue concerns the number $B$ of bootstrap repetitions set by the user: is the procedure sensitive to this parameter, and how to choose it? To address both these concerns, one dataset is simulated according to (M1). For different $B$ and learning sample sizes $n$, our estimating algorithm is launched 100 times. Table~\ref{tab:stability} displays the mean and standard deviation (in brackets) of $h_{deriv}$ and $ORMSEP_{deriv}$ in (a) the worst case, and (b) the most favorable case in the simulation study from the main document. 

\begin{table}[ht]
\centering
\resizebox{\textwidth}{!}{\begin{tabular}{cc}
(a) $n = 100$ and $nsr = 0.4$ & (b)  $n = 500$ and $nsr = 0.05$\\
\begin{tabular}{c|cc}
$B$  & $h_{deriv}$ & $ORMSEP_{deriv}$ \\ 
  \hline
	50 & 33.070 \,(7.395) & 0.218 \,(0.050) \\ 
  100 & 35.710 \,(9.401) & 0.239 \,(0.075) \\ 
  500 & 34.820 \,(6.967) & 0.225 \,(0.048) \\ 
  1000 & 34.310 \,(6.872) & 0.217 \,(0.046) \\ 
  \end{tabular}
&
\begin{tabular}{c|cc}
	$B$  & $h_{deriv}$ & $ORMSEP_{deriv}$ \\ 
  \hline
	50 & 33.680 \,(4.537) & 0.043 \,(0.003) \\ 
  100 & 33.800 \,(4.837) & 0.043 \,(0.003) \\ 
  500 & 32.720 \,(3.893) & 0.044 \,(0.004) \\ 
  1000 & 33.760 \,(4.209) & 0.044 \,(0.003) \\ 
  \end{tabular}
\end{tabular}}
\caption{Stability of the bootstrap bandwidth selection.}\label{tab:stability}
\end{table}

The variability of the selected bandwidth is smaller in the most favorable case as expected. The bootstrap bandwidth selection remains stable in both cases despite the additional randomness introduced. In both situations the variability of $ORMSEP_{deriv}$ is close to 0. The default value ($B=100$) used in our procedure seems to be large enough to ensure stability as well as accuracy for predictions.

\bigskip

{\sc Raw results.} Table~\ref{tab:der_bandwidth_choice}  gives the mean and the standard deviation (in brackets) of the computed bandwidths over 100 runs in 27 situations (9 learning sample sizes $\times$ 3 noise-to-signal ratios) and corresponds to Figure~\ref{fig:der_bandwidth_choice} in the main document. 

\begin{table}[ht]
\centering
\begin{tabular}{ccccc}
   \hline
    $n$ & $nsr$      & \multicolumn{1}{l}{         $h_{reg}$} & \multicolumn{1}{l}{       $h_{deriv}$} & \multicolumn{1}{l}{ $h_{deriv}^{oracle}$} \\ 
   \hline
100 & 0.05 & 11.150 \,(1.755) & 15.060 \,(5.007) &    18.610 \,(2.357) \\ 
      & 0.2  & 14.010 \,(2.560) & 26.560 \,(6.191) &    22.990 \,(2.834) \\ 
      & 0.4  & 17.910 \,(4.147) & 35.700 \,(6.882) &    26.490 \,(3.680) \\ 
   \hline
150 & 0.05 & 11.330 \,(1.551) & 17.240 \,(4.490) &    21.810 \,(2.485) \\ 
      & 0.2  & 14.170 \,(2.396) & 30.810 \,(5.144) &    28.110 \,(2.998) \\ 
      & 0.4  & 17.160 \,(3.369) & 37.500 \,(6.308) &    33.540 \,(3.471) \\ 
   \hline
200 & 0.05 & 11.560 \,(1.321) & 19.460 \,(4.988) &    24.800 \,(2.361) \\ 
      & 0.2  & 14.960 \,(2.247) & 34.840 \,(5.158) &    33.400 \,(3.260) \\ 
      & 0.4  & 18.600 \,(2.878) & 41.500 \,(5.402) &    39.440 \,(4.246) \\ 
   \hline
250 & 0.05 & 11.460 \,(1.234) & 21.840 \,(5.502) &    27.620 \,(2.440) \\ 
      & 0.2  & 15.760 \,(2.151) & 37.260 \,(5.010) &    38.080 \,(3.617) \\ 
      & 0.4  & 18.920 \,(3.080) & 45.280 \,(5.760) &    45.320 \,(4.387) \\ 
   \hline
300 & 0.05 & 11.870 \,(1.368) & 24.200 \,(4.767) &    30.590 \,(2.575) \\ 
      & 0.2  & 16.490 \,(1.957) & 40.880 \,(4.326) &    41.810 \,(3.786) \\ 
      & 0.4  & 19.820 \,(2.823) & 49.220 \,(4.683) &    51.650 \,(4.961) \\ 
   \hline
350 & 0.05 & 12.140 \,(1.524) & 26.360 \,(5.262) &    32.930 \,(2.786) \\ 
      & 0.2  & 16.220 \,(2.116) & 42.980 \,(5.067) &    46.910 \,(3.728) \\ 
      & 0.4  & 20.870 \,(2.707) & 52.730 \,(4.750) &    56.090 \,(5.071) \\ 
   \hline
400 & 0.05 & 12.360 \,(1.150) & 28.080 \,(4.618) &    35.320 \,(3.011) \\ 
      & 0.2  & 17.360 \,(2.067) & 46.680 \,(4.759) &    49.640 \,(3.509) \\ 
      & 0.4  & 21.360 \,(3.227) & 55.960 \,(5.895) &    60.640 \,(4.394) \\ 
   \hline
450 & 0.05 & 13.000 \,(0.000) & 31.240 \,(5.190) &    37.600 \,(2.566) \\ 
      & 0.2  & 17.760 \,(1.944) & 48.320 \,(4.479) &    53.400 \,(4.566) \\ 
      & 0.4  & 22.920 \,(3.240) & 59.880 \,(6.125) &    65.000 \,(5.714) \\ 
   \hline
500 & 0.05 & 13.520 \,(1.306) & 33.840 \,(4.720) &    40.200 \,(3.028) \\ 
      & 0.2  & 18.120 \,(2.006) & 51.120 \,(4.959) &    56.800 \,(4.158) \\ 
      & 0.4  & 23.280 \,(3.358) & 61.800 \,(5.222) &    69.920 \,(6.016) \\ 
   \hline
\end{tabular}
\caption{Estimated bandwidths.} 
\label{tab:der_bandwidth_choice}
\end{table}

\subsection{Robustness to the model complexity}
We now provide several additional tables of results to assess the robustness of our estimating procedure with respect to the complexity of the simulated model (see Section~\ref{subsec:ParameterSelection} in the main paper). Similarly to model (M2), $Y \coloneqq \sum_{j=1}^J \exp(-U_j^2) + \varepsilon$ with $X \coloneqq  \sum_{j=1}^J U_j \, \phi_j \, + \, \eta$ and where $\eta  \coloneqq  \sum_{j=J+1}^{2J} V_j \, \phi_j (t)$ is a structural perturbation acting on the functional predictor $X$. Variables $U_j$ (resp. $V_j$) are iid uniform on $[-1,\, 1]$ (resp. $[-b,\, b]$). Again, the structural perturbation is controlled by the ratio $\rho  \coloneqq b^2 / (1 + b^2)$. The noise-to-signal ratio ($nsr$) of the regression model is set to 0.05 and 0.4.  Tables~\ref{tab:dimension_choice_J2_nsr005}--\ref{tab:ormsep_J4_nsr04} display, for $J$ respectively set to 2, 3 and 4 ($J = 4$ only when $nsr=0.4$), i) the number of times, out of 100, that the dimension is correctly selected (cf. Table~\ref{tab:dimension_choice} in the main document), and ii) the corresponding $ORMSEP_{reg}$ and $ORMSEP_{deriv}$ averaged over 100 runs with standard deviation in brackets (cf.~Table~\ref{tab:ormsep_J4_nsr005} in the main document).

\clearpage

\begin{table}[t]
\centering
\caption{Number of times, out of $100$, that the dimension is correctly selected with $J=2$ and $nsr=0.05$.} 
\label{tab:dimension_choice_J2_nsr005}
\begin{tabular}{c|cccc}
$n$  & $\rho = 0.05$ & $\rho = 0.1$ & $\rho = 0.2$ & $\rho = 0.4$ \\ 
  \hline
100 & 99 & 100 & 100 & 99 \\ 
  150 & 100 & 98 & 100 & 99 \\ 
  200 & 100 & 100 & 100 & 100 \\ 
  250 & 100 & 100 & 100 & 99 \\ 
  300 & 100 & 100 & 100 & 100 \\ 
  350 & 100 & 100 & 100 & 99 \\ 
  400 & 100 & 100 & 100 & 100 \\ 
  450 & 100 & 100 & 100 & 100 \\ 
  500 & 100 & 99 & 100 & 100 \\ 
  \end{tabular}
\end{table}

\begin{table}[b]
\centering
\begin{tabular}{cc|cccc}
                                                                                  & $n$            & \multicolumn{1}{l}{   $\rho = 0.05$} & \multicolumn{1}{l}{    $\rho = 0.1$} & \multicolumn{1}{l}{    $\rho = 0.2$} & \multicolumn{1}{l}{    $\rho = 0.4$} \\ 
   \hline
\parbox[t]{2mm}{\multirow{9}{*}{\rotatebox[origin=c]{90}{$ORMSEP_{reg}$}}}   & $100$ & 0.031 \,(0.008) & 0.038 \,(0.009) & 0.059 \,(0.012) & 0.116 \,(0.021) \\ 
                                                                                  & $150$ & 0.021 \,(0.006) & 0.029 \,(0.006) & 0.040 \,(0.006) & 0.080 \,(0.014) \\ 
                                                                                  & $200$ & 0.017 \,(0.004) & 0.023 \,(0.004) & 0.035 \,(0.006) & 0.066 \,(0.010) \\ 
                                                                                  & $250$ & 0.014 \,(0.003) & 0.020 \,(0.005) & 0.029 \,(0.005) & 0.056 \,(0.008) \\ 
                                                                                  & $300$ & 0.012 \,(0.003) & 0.017 \,(0.003) & 0.026 \,(0.004) & 0.047 \,(0.006) \\ 
                                                                                  & $350$ & 0.011 \,(0.003) & 0.015 \,(0.003) & 0.023 \,(0.004) & 0.043 \,(0.006) \\ 
                                                                                  & $400$ & 0.010 \,(0.002) & 0.014 \,(0.003) & 0.021 \,(0.004) & 0.038 \,(0.005) \\ 
                                                                                  & $450$ & 0.009 \,(0.002) & 0.013 \,(0.002) & 0.020 \,(0.003) & 0.035 \,(0.005) \\ 
                                                                                  & $500$ & 0.009 \,(0.002) & 0.012 \,(0.002) & 0.019 \,(0.003) & 0.033 \,(0.004) \\ 
   \hline
\parbox[t]{2mm}{\multirow{9}{*}{\rotatebox[origin=c]{90}{$ORMSEP_{deriv}$}}} & $100$ & 0.033 \,(0.011) & 0.037 \,(0.008) & 0.050 \,(0.010) & 0.108 \,(0.023) \\ 
                                                                                  & $150$ & 0.024 \,(0.005) & 0.028 \,(0.006) & 0.037 \,(0.007) & 0.080 \,(0.017) \\ 
                                                                                  & $200$ & 0.020 \,(0.004) & 0.023 \,(0.004) & 0.031 \,(0.005) & 0.067 \,(0.013) \\ 
                                                                                  & $250$ & 0.018 \,(0.003) & 0.020 \,(0.003) & 0.027 \,(0.004) & 0.063 \,(0.015) \\ 
                                                                                  & $300$ & 0.016 \,(0.003) & 0.018 \,(0.003) & 0.024 \,(0.003) & 0.052 \,(0.013) \\ 
                                                                                  & $350$ & 0.015 \,(0.004) & 0.016 \,(0.003) & 0.022 \,(0.003) & 0.050 \,(0.014) \\ 
                                                                                  & $400$ & 0.014 \,(0.002) & 0.016 \,(0.003) & 0.021 \,(0.003) & 0.045 \,(0.010) \\ 
                                                                                  & $450$ & 0.013 \,(0.003) & 0.015 \,(0.003) & 0.020 \,(0.002) & 0.042 \,(0.011) \\ 
                                                                                  & $500$ & 0.012 \,(0.002) & 0.014 \,(0.002) & 0.019 \,(0.002) & 0.039 \,(0.009) \\ 
  \end{tabular}
\caption{Average and standard deviation (in brackets) of $ORMSEP$ with $J=2$ and $nsr=0.05$.} 
\label{tab:ormsep_J2_nsr005}
\end{table}

\begin{table}[htpb]
\centering
\caption{Number of times, out of $100$, that the dimension is correctly selected with $J=3$ and $nsr=0.05$.} 
\label{tab:dimension_choice_J3_nsr005}
\begin{tabular}{c|cccc}
$n$  & $\rho = 0.05$ & $\rho = 0.1$ & $\rho = 0.2$ & $\rho = 0.4$ \\ 
  \hline
100 & 96 & 100 & 91 & 82 \\ 
  150 & 100 & 98 & 100 & 90 \\ 
  200 & 100 & 100 & 100 & 91 \\ 
  250 & 99 & 100 & 100 & 99 \\ 
  300 & 99 & 100 & 100 & 100 \\ 
  350 & 100 & 100 & 100 & 100 \\ 
  400 & 100 & 100 & 100 & 100 \\ 
  450 & 100 & 100 & 100 & 100 \\ 
  500 & 100 & 100 & 100 & 100 \\ 
  \end{tabular}
\end{table}

\begin{table}[ht]
\centering
\begin{tabular}{cc|cccc}
                                                                                  & $n$            & \multicolumn{1}{l}{   $\rho = 0.05$} & \multicolumn{1}{l}{    $\rho = 0.1$} & \multicolumn{1}{l}{    $\rho = 0.2$} & \multicolumn{1}{l}{    $\rho = 0.4$} \\ 
   \hline
\parbox[t]{2mm}{\multirow{9}{*}{\rotatebox[origin=c]{90}{$ORMSEP_{reg}$}}}   & $100$ & 0.137 \,(0.038) & 0.142 \,(0.027) & 0.208 \,(0.043) & 0.369 \,(0.056) \\ 
                                                                                  & $150$ & 0.086 \,(0.015) & 0.100 \,(0.018) & 0.147 \,(0.020) & 0.285 \,(0.044) \\ 
                                                                                  & $200$ & 0.064 \,(0.009) & 0.077 \,(0.012) & 0.117 \,(0.015) & 0.242 \,(0.036) \\ 
                                                                                  & $250$ & 0.054 \,(0.014) & 0.066 \,(0.010) & 0.102 \,(0.013) & 0.207 \,(0.023) \\ 
                                                                                  & $300$ & 0.046 \,(0.007) & 0.058 \,(0.008) & 0.091 \,(0.010) & 0.181 \,(0.021) \\ 
                                                                                  & $350$ & 0.040 \,(0.007) & 0.053 \,(0.007) & 0.084 \,(0.010) & 0.169 \,(0.020) \\ 
                                                                                  & $400$ & 0.036 \,(0.005) & 0.048 \,(0.007) & 0.078 \,(0.009) & 0.157 \,(0.016) \\ 
                                                                                  & $450$ & 0.033 \,(0.004) & 0.044 \,(0.005) & 0.073 \,(0.009) & 0.143 \,(0.015) \\ 
                                                                                  & $500$ & 0.030 \,(0.004) & 0.041 \,(0.005) & 0.068 \,(0.008) & 0.134 \,(0.013) \\ 
   \hline
\parbox[t]{2mm}{\multirow{9}{*}{\rotatebox[origin=c]{90}{$ORMSEP_{deriv}$}}} & $100$ & 0.097 \,(0.102) & 0.087 \,(0.014) & 0.153 \,(0.129) & 0.300 \,(0.121) \\ 
                                                                                  & $150$ & 0.056 \,(0.009) & 0.075 \,(0.074) & 0.086 \,(0.013) & 0.218 \,(0.097) \\ 
                                                                                  & $200$ & 0.045 \,(0.006) & 0.052 \,(0.009) & 0.074 \,(0.010) & 0.185 \,(0.085) \\ 
                                                                                  & $250$ & 0.046 \,(0.073) & 0.044 \,(0.006) & 0.063 \,(0.008) & 0.139 \,(0.033) \\ 
                                                                                  & $300$ & 0.042 \,(0.072) & 0.040 \,(0.004) & 0.057 \,(0.007) & 0.125 \,(0.017) \\ 
                                                                                  & $350$ & 0.032 \,(0.003) & 0.036 \,(0.003) & 0.051 \,(0.007) & 0.112 \,(0.013) \\ 
                                                                                  & $400$ & 0.029 \,(0.003) & 0.033 \,(0.003) & 0.048 \,(0.005) & 0.106 \,(0.014) \\ 
                                                                                  & $450$ & 0.027 \,(0.002) & 0.031 \,(0.003) & 0.044 \,(0.005) & 0.099 \,(0.015) \\ 
                                                                                  & $500$ & 0.025 \,(0.002) & 0.029 \,(0.002) & 0.041 \,(0.003) & 0.093 \,(0.011) \\ 
  \end{tabular}
\caption{Average and standard deviation (in brackets) of $ORMSEP$ with $J=3$ and $nsr=0.05$.} 
\label{tab:ormsep_J3_nsr005}
\end{table}

\begin{table}[htpb]
\centering
\caption{Number of times, out of $100$, that the dimension is correctly selected with $J=2$ and $nsr=0.4$.} 
\label{tab:dimension_choice_J2_nsr04}
\begin{tabular}{c|cccc}
$n$  & $\rho = 0.05$ & $\rho = 0.1$ & $\rho = 0.2$ & $\rho = 0.4$ \\ 
  \hline
100 & 91 & 97 & 98 & 91 \\ 
  150 & 96 & 95 & 93 & 98 \\ 
  200 & 94 & 97 & 97 & 98 \\ 
  250 & 99 & 99 & 100 & 98 \\ 
  300 & 99 & 100 & 98 & 99 \\ 
  350 & 99 & 97 & 100 & 99 \\ 
  400 & 98 & 99 & 100 & 100 \\ 
  450 & 99 & 100 & 100 & 100 \\ 
  500 & 97 & 100 & 100 & 99 \\ 
  \end{tabular}
\end{table}

\begin{table}[ht]
\centering
\begin{tabular}{cc|cccc}
                                                                                  & $n$            & \multicolumn{1}{l}{   $\rho = 0.05$} & \multicolumn{1}{l}{    $\rho = 0.1$} & \multicolumn{1}{l}{    $\rho = 0.2$} & \multicolumn{1}{l}{    $\rho = 0.4$} \\ 
   \hline
\parbox[t]{2mm}{\multirow{9}{*}{\rotatebox[origin=c]{90}{$ORMSEP_{reg}$}}}   & $100$ & 0.107 \,(0.043) & 0.111 \,(0.041) & 0.143 \,(0.040) & 0.235 \,(0.070) \\ 
                                                                                  & $150$ & 0.076 \,(0.026) & 0.085 \,(0.030) & 0.114 \,(0.040) & 0.172 \,(0.042) \\ 
                                                                                  & $200$ & 0.063 \,(0.024) & 0.067 \,(0.023) & 0.095 \,(0.024) & 0.143 \,(0.030) \\ 
                                                                                  & $250$ & 0.049 \,(0.017) & 0.056 \,(0.017) & 0.079 \,(0.021) & 0.125 \,(0.023) \\ 
                                                                                  & $300$ & 0.042 \,(0.013) & 0.050 \,(0.013) & 0.072 \,(0.015) & 0.110 \,(0.021) \\ 
                                                                                  & $350$ & 0.039 \,(0.012) & 0.047 \,(0.013) & 0.064 \,(0.014) & 0.103 \,(0.019) \\ 
                                                                                  & $400$ & 0.036 \,(0.013) & 0.045 \,(0.012) & 0.058 \,(0.011) & 0.095 \,(0.016) \\ 
                                                                                  & $450$ & 0.032 \,(0.008) & 0.041 \,(0.009) & 0.055 \,(0.011) & 0.085 \,(0.013) \\ 
                                                                                  & $500$ & 0.030 \,(0.009) & 0.038 \,(0.009) & 0.052 \,(0.011) & 0.082 \,(0.016) \\ 
   \hline
\parbox[t]{2mm}{\multirow{9}{*}{\rotatebox[origin=c]{90}{$ORMSEP_{deriv}$}}} & $100$ & 0.116 \,(0.117) & 0.094 \,(0.072) & 0.111 \,(0.070) & 0.197 \,(0.089) \\ 
                                                                                  & $150$ & 0.072 \,(0.061) & 0.077 \,(0.072) & 0.094 \,(0.091) & 0.146 \,(0.068) \\ 
                                                                                  & $200$ & 0.071 \,(0.086) & 0.058 \,(0.056) & 0.064 \,(0.017) & 0.121 \,(0.027) \\ 
                                                                                  & $250$ & 0.045 \,(0.014) & 0.050 \,(0.052) & 0.055 \,(0.013) & 0.101 \,(0.021) \\ 
                                                                                  & $300$ & 0.040 \,(0.011) & 0.041 \,(0.009) & 0.049 \,(0.011) & 0.091 \,(0.016) \\ 
                                                                                  & $350$ & 0.038 \,(0.012) & 0.044 \,(0.054) & 0.044 \,(0.008) & 0.081 \,(0.015) \\ 
                                                                                  & $400$ & 0.035 \,(0.015) & 0.036 \,(0.009) & 0.042 \,(0.008) & 0.077 \,(0.014) \\ 
                                                                                  & $450$ & 0.034 \,(0.014) & 0.032 \,(0.007) & 0.040 \,(0.008) & 0.071 \,(0.012) \\ 
                                                                                  & $500$ & 0.035 \,(0.054) & 0.031 \,(0.007) & 0.038 \,(0.006) & 0.066 \,(0.013) \\ 
  \end{tabular}
\caption{Average and standard deviation (in brackets) of $ORMSEP$ with $J=2$ and $nsr=0.4$.} 
\label{tab:ormsep_J2_nsr04}
\end{table}

\begin{table}[htpb]
\centering
\caption{Number of times, out of $100$, that the dimension is correctly selected with $J=3$ and $nsr=0.4$.} 
\label{tab:dimension_choice_J3_nsr04}
\begin{tabular}{c|cccc}
$n$  & $\rho = 0.05$ & $\rho = 0.1$ & $\rho = 0.2$ & $\rho = 0.4$ \\ 
  \hline
100 & 55 & 58 & 60 & 37 \\ 
  150 & 67 & 71 & 65 & 51 \\ 
  200 & 69 & 76 & 71 & 57 \\ 
  250 & 81 & 75 & 71 & 68 \\ 
  300 & 86 & 88 & 87 & 67 \\ 
  350 & 89 & 89 & 85 & 73 \\ 
  400 & 87 & 90 & 78 & 78 \\ 
  450 & 94 & 93 & 85 & 81 \\ 
  500 & 98 & 94 & 92 & 82 \\ 
  \end{tabular}
\end{table}

\begin{table}[ht]
\centering
\begin{tabular}{cc|cccc}
                                                                                  & $n$            & \multicolumn{1}{l}{   $\rho = 0.05$} & \multicolumn{1}{l}{    $\rho = 0.1$} & \multicolumn{1}{l}{    $\rho = 0.2$} & \multicolumn{1}{l}{    $\rho = 0.4$} \\ 
   \hline
\parbox[t]{2mm}{\multirow{9}{*}{\rotatebox[origin=c]{90}{$ORMSEP_{reg}$}}}   & $100$ & 0.287 \,(0.085) & 0.310 \,(0.076) & 0.344 \,(0.075) & 0.503 \,(0.077) \\ 
                                                                                  & $150$ & 0.202 \,(0.054) & 0.223 \,(0.052) & 0.276 \,(0.054) & 0.426 \,(0.062) \\ 
                                                                                  & $200$ & 0.168 \,(0.042) & 0.189 \,(0.040) & 0.235 \,(0.050) & 0.368 \,(0.048) \\ 
                                                                                  & $250$ & 0.141 \,(0.041) & 0.163 \,(0.040) & 0.205 \,(0.037) & 0.324 \,(0.047) \\ 
                                                                                  & $300$ & 0.120 \,(0.025) & 0.134 \,(0.030) & 0.182 \,(0.032) & 0.298 \,(0.043) \\ 
                                                                                  & $350$ & 0.108 \,(0.024) & 0.123 \,(0.023) & 0.167 \,(0.027) & 0.274 \,(0.040) \\ 
                                                                                  & $400$ & 0.099 \,(0.023) & 0.113 \,(0.023) & 0.157 \,(0.028) & 0.252 \,(0.030) \\ 
                                                                                  & $450$ & 0.089 \,(0.018) & 0.104 \,(0.019) & 0.147 \,(0.024) & 0.239 \,(0.032) \\ 
                                                                                  & $500$ & 0.083 \,(0.015) & 0.098 \,(0.018) & 0.138 \,(0.018) & 0.232 \,(0.029) \\ 
   \hline
\parbox[t]{2mm}{\multirow{9}{*}{\rotatebox[origin=c]{90}{$ORMSEP_{deriv}$}}} & $100$ & 0.321 \,(0.227) & 0.300 \,(0.204) & 0.315 \,(0.179) & 0.513 \,(0.142) \\ 
                                                                                  & $150$ & 0.236 \,(0.212) & 0.207 \,(0.157) & 0.267 \,(0.187) & 0.422 \,(0.160) \\ 
                                                                                  & $200$ & 0.229 \,(0.234) & 0.179 \,(0.164) & 0.217 \,(0.171) & 0.351 \,(0.140) \\ 
                                                                                  & $250$ & 0.159 \,(0.186) & 0.182 \,(0.190) & 0.200 \,(0.161) & 0.287 \,(0.133) \\ 
                                                                                  & $300$ & 0.130 \,(0.151) & 0.121 \,(0.136) & 0.133 \,(0.099) & 0.272 \,(0.134) \\ 
                                                                                  & $350$ & 0.103 \,(0.112) & 0.106 \,(0.114) & 0.130 \,(0.107) & 0.244 \,(0.131) \\ 
                                                                                  & $400$ & 0.121 \,(0.170) & 0.102 \,(0.123) & 0.154 \,(0.138) & 0.217 \,(0.116) \\ 
                                                                                  & $450$ & 0.085 \,(0.123) & 0.083 \,(0.082) & 0.120 \,(0.110) & 0.205 \,(0.123) \\ 
                                                                                  & $500$ & 0.063 \,(0.073) & 0.076 \,(0.076) & 0.093 \,(0.078) & 0.192 \,(0.112) \\ 
  \end{tabular}
\caption{Average and standard deviation (in brackets) of $ORMSEP$ with $J=3$ and $nsr=0.4$.} 
\label{tab:ormsep_J3_nsr04}
\end{table}

\begin{table}[htpb]
\centering
\caption{Number of times, out of $100$, that the dimension is correctly selected with $J=4$ and $nsr=0.4$.} 
\label{tab:dimension_choice_J4_nsr04}
\begin{tabular}{c|cccc}
$n$  & $\rho = 0.05$ & $\rho = 0.1$ & $\rho = 0.2$ & $\rho = 0.4$ \\ 
  \hline
100 & 6 & 3 & 6 & 9 \\ 
  150 & 8 & 7 & 5 & 3 \\ 
  200 & 7 & 6 & 7 & 3 \\ 
  250 & 9 & 8 & 8 & 4 \\ 
  300 & 15 & 12 & 7 & 3 \\ 
  350 & 17 & 10 & 11 & 1 \\ 
  400 & 15 & 7 & 5 & 2 \\ 
  450 & 24 & 12 & 10 & 0 \\ 
  500 & 19 & 16 & 12 & 1 \\ 
  \end{tabular}
\end{table}

\begin{table}[ht]
\centering
\begin{tabular}{cc|cccc}
                                                                                  & $n$            & \multicolumn{1}{l}{   $\rho = 0.05$} & \multicolumn{1}{l}{    $\rho = 0.1$} & \multicolumn{1}{l}{    $\rho = 0.2$} & \multicolumn{1}{l}{    $\rho = 0.4$} \\ 
   \hline
\parbox[t]{2mm}{\multirow{9}{*}{\rotatebox[origin=c]{90}{$ORMSEP_{reg}$}}}   & $100$ & 0.516 \,(0.101) & 0.502 \,(0.075) & 0.549 \,(0.081) & 0.701 \,(0.087) \\ 
                                                                                  & $150$ & 0.394 \,(0.066) & 0.416 \,(0.071) & 0.451 \,(0.059) & 0.608 \,(0.064) \\ 
                                                                                  & $200$ & 0.337 \,(0.057) & 0.345 \,(0.054) & 0.405 \,(0.051) & 0.550 \,(0.055) \\ 
                                                                                  & $250$ & 0.298 \,(0.051) & 0.310 \,(0.048) & 0.351 \,(0.045) & 0.512 \,(0.052) \\ 
                                                                                  & $300$ & 0.258 \,(0.037) & 0.270 \,(0.044) & 0.329 \,(0.048) & 0.480 \,(0.052) \\ 
                                                                                  & $350$ & 0.239 \,(0.031) & 0.250 \,(0.043) & 0.305 \,(0.041) & 0.453 \,(0.046) \\ 
                                                                                  & $400$ & 0.211 \,(0.032) & 0.229 \,(0.032) & 0.285 \,(0.033) & 0.428 \,(0.037) \\ 
                                                                                  & $450$ & 0.199 \,(0.029) & 0.212 \,(0.024) & 0.271 \,(0.031) & 0.409 \,(0.038) \\ 
                                                                                  & $500$ & 0.190 \,(0.032) & 0.205 \,(0.024) & 0.259 \,(0.032) & 0.396 \,(0.033) \\ 
   \hline
\parbox[t]{2mm}{\multirow{9}{*}{\rotatebox[origin=c]{90}{$ORMSEP_{deriv}$}}} & $100$ & 0.561 \,(0.176) & 0.577 \,(0.156) & 0.559 \,(0.135) & 0.707 \,(0.114) \\ 
                                                                                  & $150$ & 0.479 \,(0.180) & 0.505 \,(0.169) & 0.508 \,(0.137) & 0.646 \,(0.111) \\ 
                                                                                  & $200$ & 0.451 \,(0.183) & 0.446 \,(0.152) & 0.478 \,(0.161) & 0.602 \,(0.104) \\ 
                                                                                  & $250$ & 0.435 \,(0.207) & 0.413 \,(0.158) & 0.401 \,(0.122) & 0.542 \,(0.100) \\ 
                                                                                  & $300$ & 0.354 \,(0.151) & 0.363 \,(0.143) & 0.403 \,(0.133) & 0.529 \,(0.107) \\ 
                                                                                  & $350$ & 0.333 \,(0.157) & 0.353 \,(0.143) & 0.362 \,(0.117) & 0.505 \,(0.091) \\ 
                                                                                  & $400$ & 0.350 \,(0.178) & 0.367 \,(0.145) & 0.368 \,(0.104) & 0.479 \,(0.085) \\ 
                                                                                  & $450$ & 0.294 \,(0.158) & 0.318 \,(0.114) & 0.356 \,(0.117) & 0.484 \,(0.087) \\ 
                                                                                  & $500$ & 0.319 \,(0.171) & 0.295 \,(0.117) & 0.319 \,(0.089) & 0.462 \,(0.081) \\ 
  \end{tabular}
\caption{Average and standard deviation (in brackets) of $ORMSEP$ with $J=4$ and $nsr=0.4$.} 
\label{tab:ormsep_J4_nsr04}
\end{table}


\clearpage

\subsection{Complement to the study from Section~\ref{subsec:ComparativeStudy}}
Tables~\ref{tab:R-M2-2} and \ref{tab:R-M2-3} supplement results given in Tables~\ref{tab:R-M2-1} and \ref{tab:R-M2-4} for two additional perturbation levels and $500$ learning functions, and Tables~\ref{tab:R-M2_100-1}--\ref{tab:R-M2_100-4} provide analogous results for samples of $100$ learning functions. The conclusions from the main paper all hold true. Especially in the situation when the learning sample size is low, the method (MY) seems to be numerically unstable. Remarkably, the local linear estimation method does not appear to suffer from such drawbacks.

\subsubsection{Results for 500 learning functions}

\begin{table}[htpb]
\centering
\caption{Model (M3) with $nsr = 0.1$ and $\rho=0.1$.} 
\label{tab:R-M2-2}
\resizebox{\textwidth}{!}{\begin{tabular}{c|c|ccccc}
&   & $a = 0$ & $a = 0.25$ & $a = 0.5$ & $a = 0.75$ & $a = 1$ \\ 
  \hline
\parbox[t]{2mm}{\multirow{4}{*}{\rotatebox[origin=c]{90}{Reg.}}}& L & 0.001 \,(0.001) & 0.015 \,(0.001) & 0.113 \,(0.009) & 0.532 \,(0.031) & 1.007 \,(0.011) \\ 
&   LC & 0.061 \,(0.008) & 0.064 \,(0.008) & 0.087 \,(0.011) & 0.187 \,(0.021) & 0.292 \,(0.032) \\ 
&   LL & 0.035 \,(0.028) & 0.034 \,(0.024) & 0.048 \,(0.021) & 0.091 \,(0.020) & 0.142 \,(0.018) \\ 
&   MY & 0.026 \,(0.022) & 0.035 \,(0.024) & 0.724 \,(6.494) & 0.207 \,(0.291) & 0.271 \,(0.051) \\ 
   \hline
\parbox[t]{2mm}{\multirow{3}{*}{\rotatebox[origin=c]{90}{Deriv.}}}& L & \graytext{1.847 \,(0.029)} & 1.093 \,(0.096) & 1.026 \,(0.031) & 1.014 \,(0.016) & 1.004 \,(0.007) \\ 
&   LL & \graytext{0.228 \,(0.213)} & 8.999 \,(8.559) & 1.191 \,(1.158) & 0.210 \,(0.243) & 0.124 \,(0.116) \\ 
&   MY & \graytext{0.407 \,(0.116)} & 5.040 \,(6.388) & 0.766 \,(0.577) & 0.287 \,(0.110) & 0.221 \,(0.062) \\ 
  \end{tabular}}
\end{table}

\begin{table}[htpb]
\centering
\caption{Model (M3) with $nsr = 0.2$ and $\rho=0.2$.} 
\label{tab:R-M2-3}
\resizebox{\textwidth}{!}{\begin{tabular}{c|c|ccccc}
&   & $a = 0$ & $a = 0.25$ & $a = 0.5$ & $a = 0.75$ & $a = 1$ \\ 
  \hline
\parbox[t]{2mm}{\multirow{4}{*}{\rotatebox[origin=c]{90}{Reg.}}}& L & 0.003 \,(0.002) & 0.017 \,(0.002) & 0.113 \,(0.010) & 0.534 \,(0.032) & 1.013 \,(0.018) \\ 
&   LC & 0.085 \,(0.014) & 0.092 \,(0.013) & 0.125 \,(0.017) & 0.260 \,(0.031) & 0.402 \,(0.040) \\ 
&   LL & 0.036 \,(0.036) & 0.050 \,(0.036) & 0.071 \,(0.030) & 0.141 \,(0.020) & 0.220 \,(0.023) \\ 
&   MY & 0.062 \,(0.210) & 0.051 \,(0.028) & 0.161 \,(0.769) & 0.191 \,(0.039) & 0.653 \,(3.516) \\ 
   \hline
\parbox[t]{2mm}{\multirow{3}{*}{\rotatebox[origin=c]{90}{Deriv.}}}& L & \graytext{1.852 \,(0.031)} & 1.133 \,(0.115) & 1.024 \,(0.021) & 1.011 \,(0.011) & 1.005 \,(0.008) \\ 
&   LL & \graytext{0.514 \,(0.200)} & 9.041 \,(8.499) & 1.203 \,(1.066) & 0.333 \,(0.193) & 0.279 \,(0.109) \\ 
&   MY & \graytext{0.602 \,(0.893)} & 12.163 \,(24.987) & 1.125 \,(0.857) & 0.505 \,(1.284) & 0.326 \,(0.479) \\ 
  \end{tabular}}
\end{table}

\FloatBarrier
\clearpage 

\subsubsection{Results for 100 learning functions}

\begin{table}[htpb]
\centering
\caption{Model (M3) with $nsr = 0.05$ and $\rho=0.05$.} 
\label{tab:R-M2_100-1}
\resizebox{\textwidth}{!}{\begin{tabular}{c|c|ccccc}
&   & $a = 0$ & $a = 0.25$ & $a = 0.5$ & $a = 0.75$ & $a = 1$ \\ 
  \hline
\parbox[t]{2mm}{\multirow{4}{*}{\rotatebox[origin=c]{90}{Reg.}}}& L & 0.003 \,(0.002) & 0.017 \,(0.003) & 0.119 \,(0.010) & 0.567 \,(0.041) & 1.029 \,(0.043) \\ 
&   LC & 0.134 \,(0.030) & 0.138 \,(0.028) & 0.188 \,(0.040) & 0.376 \,(0.053) & 0.593 \,(0.088) \\ 
&   LL & 0.071 \,(0.061) & 0.084 \,(0.061) & 0.110 \,(0.062) & 0.239 \,(0.045) & 0.382 \,(0.064) \\ 
&   MY & 0.446 \,(3.538) & 65.192 \,(651.124) & 0.236 \,(0.593) & 0.814 \,(4.602) & 1.609 \,(11.202) \\ 
   \hline
\parbox[t]{2mm}{\multirow{3}{*}{\rotatebox[origin=c]{90}{Deriv.}}}& L & \graytext{1.865 \,(0.101)} & 1.229 \,(0.354) & 1.069 \,(0.133) & 1.045 \,(0.084) & 1.012 \,(0.049) \\ 
&   LL & \graytext{0.760 \,(0.376)} & 10.463 \,(8.701) & 1.424 \,(1.113) & 0.472 \,(0.292) & 0.413 \,(0.204) \\ 
&   MY & \graytext{0.680 \,(0.668)} & 50.611 \,(277.149) & 5.677 \,(27.481) & 3.090 \,(11.623) & 0.369 \,(0.359) \\ 
  \end{tabular}}
\end{table}

\begin{table}[htpb]
\centering
\caption{Model (M3) with $nsr = 0.1$ and $\rho=0.1$.} 
\label{tab:R-M2_100-2}
\resizebox{\textwidth}{!}{\begin{tabular}{c|c|ccccc}
&   & $a = 0$ & $a = 0.25$ & $a = 0.5$ & $a = 0.75$ & $a = 1$ \\ 
  \hline
\parbox[t]{2mm}{\multirow{4}{*}{\rotatebox[origin=c]{90}{Reg.}}}& L & 0.007 \,(0.005) & 0.021 \,(0.006) & 0.123 \,(0.014) & 0.570 \,(0.049) & 1.038 \,(0.055) \\ 
&   LC & 0.151 \,(0.031) & 0.148 \,(0.034) & 0.198 \,(0.045) & 0.416 \,(0.064) & 0.637 \,(0.089) \\ 
&   LL & 0.083 \,(0.070) & 0.095 \,(0.065) & 0.120 \,(0.066) & 0.257 \,(0.051) & 0.411 \,(0.064) \\ 
&   MY & 0.079 \,(0.080) & 0.135 \,(0.488) & 14.058 \,(139.119) & 0.622 \,(2.011) & 3.253 \,(24.778) \\ 
   \hline
\parbox[t]{2mm}{\multirow{3}{*}{\rotatebox[origin=c]{90}{Deriv.}}}& L & \graytext{1.869 \,(0.098)} & 1.404 \,(0.554) & 1.098 \,(0.151) & 1.053 \,(0.086) & 1.014 \,(0.050) \\ 
&   LL & \graytext{0.770 \,(0.329)} & 10.960 \,(8.586) & 1.397 \,(1.064) & 0.533 \,(0.271) & 0.418 \,(0.177) \\ 
&   MY & \graytext{2.244 \,(7.015)} & 19007.294 \,(188516.083) & 24.681 \,(201.610) & 23.785 \,(126.218) & 1.217 \,(3.787) \\ 
  \end{tabular}}
\end{table}

\begin{table}[htpb]
\centering
\caption{Model (M3) with $nsr = 0.2$ and $\rho=0.2$.} 
\label{tab:R-M2_100-3}
\resizebox{\textwidth}{!}{\begin{tabular}{c|c|ccccc}
&   & $a = 0$ & $a = 0.25$ & $a = 0.5$ & $a = 0.75$ & $a = 1$ \\ 
  \hline
\parbox[t]{2mm}{\multirow{4}{*}{\rotatebox[origin=c]{90}{Reg.}}}& L & 0.014 \,(0.008) & 0.030 \,(0.012) & 0.131 \,(0.016) & 0.573 \,(0.046) & 1.045 \,(0.055) \\ 
&   LC & 0.182 \,(0.038) & 0.197 \,(0.044) & 0.252 \,(0.052) & 0.470 \,(0.079) & 0.754 \,(0.090) \\ 
&   LL & 0.105 \,(0.075) & 0.111 \,(0.078) & 0.180 \,(0.071) & 0.323 \,(0.066) & 0.496 \,(0.068) \\ 
&   MY & 0.449 \,(2.173) & 0.218 \,(0.674) & 0.221 \,(0.206) & 0.617 \,(1.084) & 5.052 \,(43.219) \\ 
   \hline
\parbox[t]{2mm}{\multirow{3}{*}{\rotatebox[origin=c]{90}{Deriv.}}}& L & \graytext{1.877 \,(0.077)} & 1.551 \,(0.691) & 1.077 \,(0.076) & 1.035 \,(0.043) & 1.017 \,(0.040) \\ 
&   LL & \graytext{0.935 \,(0.276)} & 8.854 \,(8.173) & 1.804 \,(0.958) & 0.639 \,(0.244) & 0.507 \,(0.149) \\ 
&   MY & \graytext{116.568 \,(1067.638)} & 429.089 \,(3417.084) & 14.088 \,(42.670) & 10.561 \,(81.778) & 62.822 \,(575.144) \\ 
  \end{tabular}}
\end{table}

\begin{table}[htpb]
\centering
\caption{Model (M3) with $nsr = 0.4$ and $\rho=0.4$.} 
\label{tab:R-M2_100-4}
\resizebox{\textwidth}{!}{\begin{tabular}{c|c|ccccc}
&   & $a = 0$ & $a = 0.25$ & $a = 0.5$ & $a = 0.75$ & $a = 1$ \\ 
  \hline
\parbox[t]{2mm}{\multirow{4}{*}{\rotatebox[origin=c]{90}{Reg.}}}& L & 0.042 \,(0.025) & 0.059 \,(0.021) & 0.161 \,(0.028) & 0.601 \,(0.060) & 1.058 \,(0.075) \\ 
&   LC & 0.303 \,(0.065) & 0.303 \,(0.066) & 0.359 \,(0.073) & 0.632 \,(0.083) & 0.939 \,(0.101) \\ 
&   LL & 0.191 \,(0.114) & 0.184 \,(0.109) & 0.254 \,(0.110) & 0.451 \,(0.075) & 0.705 \,(0.081) \\ 
&   MY & 0.235 \,(0.191) & 0.215 \,(0.167) & 0.296 \,(0.229) & 0.544 \,(0.200) & 1.255 \,(3.361) \\ 
   \hline
\parbox[t]{2mm}{\multirow{3}{*}{\rotatebox[origin=c]{90}{Deriv.}}}& L & \graytext{1.869 \,(0.049)} & 2.023 \,(0.495) & 1.146 \,(0.075) & 1.044 \,(0.028) & 1.012 \,(0.023) \\ 
&   LL & \graytext{1.331 \,(0.213)} & 10.177 \,(7.576) & 1.902 \,(0.871) & 0.823 \,(0.177) & 0.721 \,(0.116) \\ 
&   MY & \graytext{11.547 \,(84.817)} & 309.921 \,(1678.421) & 2772.892 \,(27642.691) & 10.038 \,(37.057) & 6.362 \,(47.083) \\ 
  \end{tabular}}
\end{table}

\FloatBarrier

\subsection{Complement to the data analysis from Section~\ref{subsec:GrowthDataset}}
As explained in Section~\ref{subsec:GrowthDataset}, restricting the growth velocity profiles from ages 1--10 to 6--10 does not degrade the quality of estimation. Figure~\ref{fig:growth_sfim_obs_vs_estimations_compare_1-10_6-10} displays the observed responses versus their estimates when considering the whole growth velocity profile (1--10), or the restricted one (6--10). To quantify the performance of the estimating procedure, the empirical (Pearson's) correlation coefficient between the observations and their estimates is computed in each situation. When the regression model involves the whole trajectory of the growth velocity, the correlation equals 0.823; in the other case where estimates are based on ages 6--10, one gets 0.819. The accuracy of the estimating procedures are almost the same, which confirms that the behavior of the growth velocity profile under 6 years of age does not influence the adult height at 18.
\begin{figure}[ht]
\centering
\includegraphics[scale = 0.6]{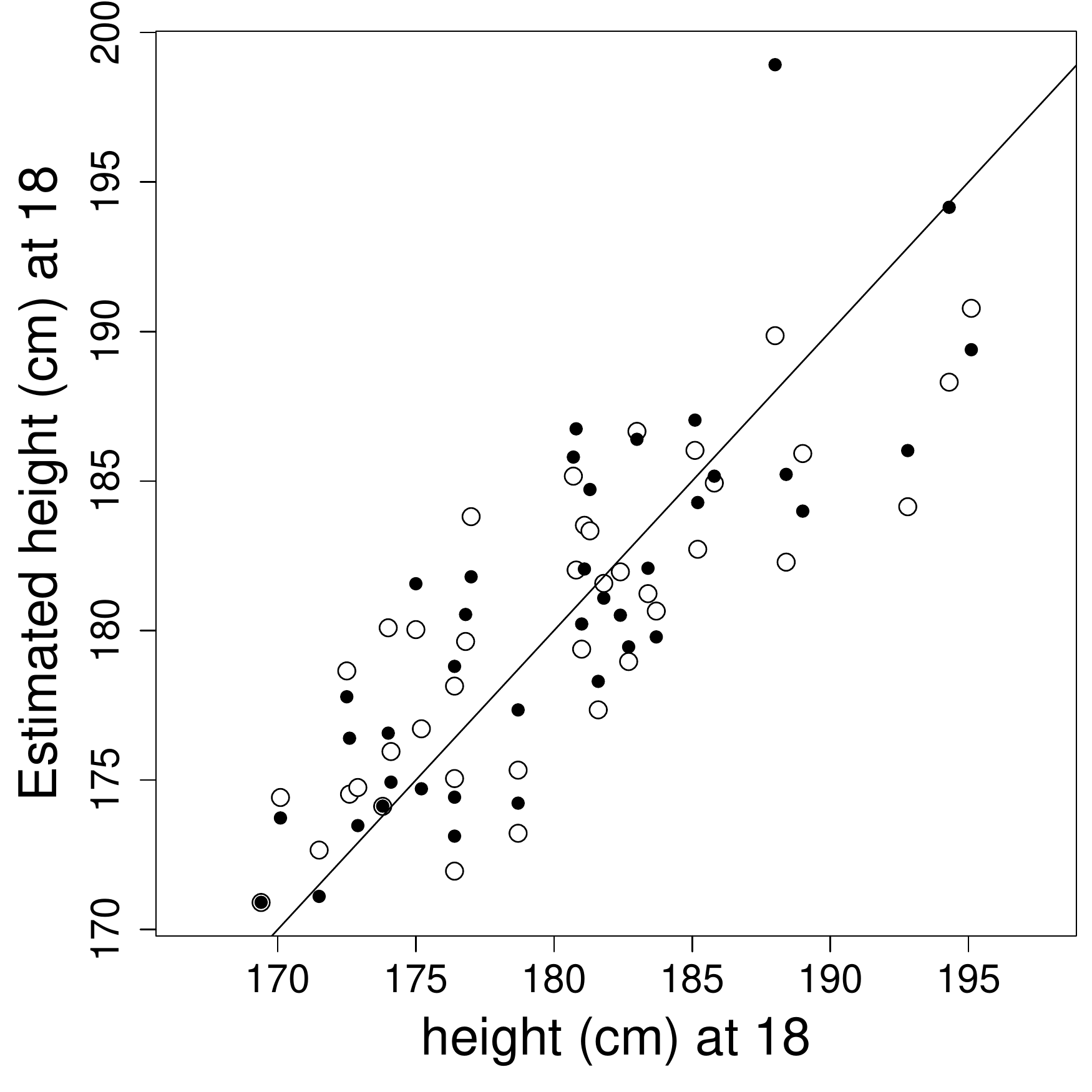}
\caption{Estimates based on the 1--10 growth velocity profiles (black points) and on the 6--10 growth velocity profiles (circles).}
\label{fig:growth_sfim_obs_vs_estimations_compare_1-10_6-10}
\end{figure}

\section{Theoretical complement on (H\ref{hypo:gammafunction})} \label{app:D}
In this section, we investigate hypothesis (H\ref{hypo:gammafunction}) that requires regularity of the family of functions $$\gamma_{j_1,\ldots,j_K}^{p_1,\ldots,p_K}(t) = \E \left( \langle \phi_{j_1}, X_1-x \rangle^{p_1}  \cdots  \langle \phi_{j_K}, X_1-x \rangle^{p_K} | \| X_1 - x \|^{p_1 + \cdots + p_K} = t \right).$$
The next lemma provides a general condition on the functional predictor $X$ in order to fulfill (H\ref{hypo:gammafunction}). 

\begin{lem}\label{lem:aboutH3}
Suppose that for some $x \in H$ the random vector 	
	\[	\left(\langle \phi_{1}, X_1-x \rangle, \dots, \langle \phi_{J}, X_1-x \rangle, \left\Vert X_1 - x \right\Vert - \sqrt{\sum_{i=1}^J \langle \phi_{i}, X_1-x \rangle^2} \right)\tr	\]
is absolutely continuous with a density in $\R^{J+1}$ that is positive at the origin, continuous at the origin in its first $J$ coordinates, and continuous at the origin from the right in its last coordinate. Then condition (H\ref{hypo:gammafunction}) is satisfied for all functions in $H$. 
\end{lem}

\begin{proof}
For $x \in H$ fixed and $j = 1, \dots, J$ denote $Z_j = \langle \phi_{j}, X_1-x \rangle$. We want to establish that the matrix
	\[	\mathbf{\Gamma} = \left[ {\gamma_{j,k}^{1,1}}^{'}(0) \right]_{j,k=1}^J = \lim_{t\to 0} \frac{1}{t} \left[ \E \left( Z_j Z_k | \| X_1 - x \|^2 = t \right) \right]_{j,k=1}^J	\]
is positive definite. That is equivalent with the fact that for any $\bu = \left(u_1, \dots, u_J\right)\tr \in \R^J$, $\left\Vert \bu \right\Vert = 1$
	\[
	\begin{aligned}
	0 < & \bu\tr \bGamma \bu =  \lim_{t\to 0} \frac{1}{t} \sum_{j=1}^J \sum_{k=1}^J \E \left( u_j Z_j Z_k u_k \middle| \| X_1 - x \|^2 = t \right)\\
	& = \lim_{t\to 0} \frac{1}{t} \E \left( \left(\sum_{j=1}^J u_j Z_j\right)^2 \middle| \| X_1 - x \|^2 = t \right).	
	\end{aligned}
	\]
For any univariate random variable $Z$ with variance and $z \in \R$ we know that
	\begin{equation}	\label{eq:variance inequality}
	\E \left(Z - z\right)^2 = \E \left(Z - \E Z\right)^2 + \left(\E Z - z\right)^2 \geq \var Z,	
	\end{equation}
with equality if and only if $z = \E Z$. Use the conditional version of this inequality to obtain
	\begin{equation}	\label{eq:conditional variance}
	\begin{aligned}
	\E & \left( \left(\sum_{j=1}^J u_j Z_j\right)^2 \middle| \| X_1 - x \|^2 = t \right) \\
	& = \E \left( \left(\sum_{j=1}^J u_j Z_j - \E \left( \sum_{j=1}^J u_j Z_j \middle| \left\Vert X_1 - x \right\Vert^2 = t \right) \right)^2 \middle| \| X_1 - x \|^2 = t \right) \\
	& \phantom{=} + \E \left( \left(\E \left( \sum_{j=1}^J u_j Z_j \middle| \left\Vert X_1 - x \right\Vert^2 = t \right) \right)^2 \middle| \| X_1 - x \|^2 = t \right) \\
	& \geq \E \left( \left(\sum_{j=1}^J u_j Z_j - \E \left( \sum_{j=1}^J u_j Z_j \middle| \left\Vert X_1 - x \right\Vert^2 = t \right) \right)^2 \middle| \| X_1 - x \|^2 = t \right) \\
	& = \var \left( \sum_{j=1}^J u_j Z_j \middle| \left\Vert X_1 - x \right\Vert^2 = t \right) = \var \left( \sum_{j=1}^J u_j Z_j \middle| \sum_{i=1}^\infty Z_i^2 = t \right).
	\end{aligned}
	\end{equation}
Therefore, it suffices to show that the conditional variance of no projection of the vector $\left(Z_1, \dots, Z_J\right)\tr$ into a line spanned by a unit vector is of order $o(t)$ with $t \to 0$. 

We assume that the random vector $\bZ = \left(Z_1, \dots, Z_J, \sqrt{\sum_{i=J+1}^\infty Z_i^2} \right)\tr$ is absolutely continuous in $\R^{J+1}$. For an independent Rademacher random variable $R$, i.e. $\PP\left(R = 1\right) = \PP\left( R = -1 \right) = 1/2$, define $\widetilde{\bZ} = \left(Z_1, \dots, Z_J, R \sqrt{\sum_{i=J+1}^\infty Z_i^2} \right)\tr$. This random vector is absolutely continuous, with density $f_J$ positive and continuous at the origin. It differs from the original random vector $\bZ$ only in its last coordinate, and $\bZ\tr \bZ$ has the same distribution as $\widetilde{\bZ}\tr \widetilde{\bZ}$. 
The conditional density of $\widetilde{\bZ}$ given $\sum_{i=1}^\infty Z_i^2 = \widetilde{\bZ}\tr \widetilde{\bZ} = t$ takes the form
	\begin{equation}	\label{eq:conditional density}
	\frac{f_J\left(\bz\right) 1\left[ t = \bz\tr \bz \right]}{\int_{\{\bv\tr \bv = t\}} f_J\left(\bv\right) \dd \bv} \quad \mbox{for }\bz\in\R^{J+1}, \end{equation}
where $1\left[ t = \bz\tr \bz \right]$ is $1$ if $t = \bz\tr \bz$, $0$ otherwise. The integral in \eqref{eq:conditional density}, and in analogous expressions below, is taken with respect to the Hausdorff measure on an appropriate sphere in $\R^{J+1}$. By our assumptions, $f_J$ is positive and continuous in the neighborhood of the origin. Then, for $t$ small enough, $c_{J,t} = \inf_{\{\bz\tr \bz = t\}} f_J(\bz)$ must be positive. Using \eqref{eq:variance inequality} again, we can therefore write
	\[
	\begin{aligned}
	\var & \left( \sum_{j=1}^J u_j Z_j \middle| \sum_{i=1}^\infty Z_i^2 = t \right) \\
	& = \int_{\R^{J+1}} \left( \sum_{j=1}^J u_j z_j - \E\left( \sum_{j=1}^J u_j Z_j \middle| \widetilde{\bZ}\tr \widetilde{\bZ} = t \right) \right)^2 \frac{f_J\left(\bz\right) 1\left[ t = \bz\tr \bz \right]}{\int_{\{\bv\tr \bv = t\}} f_J\left(\bv\right) \dd \bv} \dd \bz \\
	& = \int_{\{ \bz\tr \bz = t \}} \left( \sum_{j=1}^J u_j z_j - \E\left( \sum_{j=1}^J u_j Z_j \middle| \widetilde{\bZ}\tr \widetilde{\bZ} = t \right) \right)^2 \frac{f_J\left(\bz\right)}{\int_{\{\bv\tr \bv = t\}} f_J\left(\bv\right) \dd \bv} \dd \bz \\
	& \geq \int_{\{ \bz\tr \bz = t \}} \left( \sum_{j=1}^J u_j z_j - \E\left( \sum_{j=1}^J u_j Z_j \middle| \widetilde{\bZ}\tr \widetilde{\bZ} = t \right) \right)^2 \frac{c_{J,t}}{\int_{\{\bv\tr \bv = t\}} f_J\left(\bv\right) \dd \bv} \dd \bz \\
	& = \frac{c_{J,t} \int_{\{\bv\tr \bv = t\}} 1 \dd \bv}{\int_{\{\bv\tr \bv = t\}} f_J\left(\bv\right) \dd \bv} \int_{\{ \bz\tr \bz = t \}} \left( \sum_{j=1}^J u_j z_j - \E\left( \sum_{j=1}^J u_j Z_j \middle| \widetilde{\bZ}\tr \widetilde{\bZ} = t \right) \right)^2 g_t(\bz) \dd \bz \\
	& = \frac{c_{J,t} \int_{\{\bv\tr \bv = t\}} 1 \dd \bv}{\int_{\{\bv\tr \bv = t\}} f_J\left(\bv\right) \dd \bv} \E \left( \sqrt{t} \sum_{j=1}^J u_j U_j - \E\left( \sum_{j=1}^J u_j Z_j \middle| \widetilde{\bZ}\tr \widetilde{\bZ} = t \right) \right)^2 \\	
	& \geq \frac{c_{J,t} \int_{\{\bv\tr \bv = t\}} 1 \dd \bv}{\int_{\{\bv\tr \bv = t\}} f_J\left(\bv\right) \dd \bv} \var \left( \sqrt{t} \sum_{j=1}^J u_j U_j \right) \\	
	& = \frac{c_{J,t} \int_{\{\bv\tr \bv = t\}} 1 \dd \bv}{\int_{\{\bv\tr \bv = t\}} f_J\left(\bv\right) \dd \bv} t \var\left(U_1\right) =  \frac{c_{J,t} \int_{\{\bv\tr \bv = t\}} 1 \dd \bv}{\int_{\{\bv\tr \bv = t\}} f_J\left(\bv\right) \dd \bv} \frac{t}{J+1}.
	\end{aligned}
	\]
Here, $g_t(\bz) = \left(\int_{\{\bv\tr \bv = t\}} 1 \dd \bv\right)^{-1}$ is the reciprocal of the Hausdorff measure of a sphere, and $\bU = \left(U_1, \dots, U_{J+1}\right)\tr$ is a random vector distributed uniformly on the unit sphere in $\R^{J+1}$. The second inequality is from \eqref{eq:variance inequality}. The first equality on the last line in the formula above follows from the spherical symmetry of the vector $\bU$ --- any projection of $\bU$ onto a line has the same distribution as $U_1$. The final equality follows from $\var \bU = \bI/(J+1)$, for $\bI$ the $(J+1) \times (J+1)$ identity matrix.

From the continuity of $f_J$ around the origin it follows that 
	\[	\lim_{t \to 0} \frac{c_{J,t} \int_{\{\bv\tr \bv = t\}} 1 \dd \bv}{\int_{\{\bv\tr \bv = t\}} f_J\left(\bv\right) \dd \bv} = \lim_{t\to 0} \frac{ f_J\left(\bzero \right) \int_{\{\bv\tr \bv = t\}} 1 \dd \bv}{\int_{\{\bv\tr \bv = t\}} f_J\left(\bzero \right) \dd \bv} = 1.	\]
Therefore, we obtain
	\[	\lim_{t \to 0} \frac{1}{t} \var \left( \sum_{j=1}^J u_j Z_j \middle| \sum_{i=1}^\infty Z_i^2 = t \right) \geq \frac{1}{J+1},	\]
and also the desired
	\[	\bu\tr \bGamma \bu \geq \frac{1}{J+1} > 0.	\]
Note that the uniformity in $\bu$ is assured by the spherical symmetry of $\bU$ utilized above.

Finally, to see that, for instance, function ${\gamma_{j,k}^{1,1}}$ is continuously differentiable at the origin, note that from \eqref{eq:conditional density} we get
		\[	\gamma_{j,k}^{1,1}(t) = \int_{\{\bz\tr\bz = t\}} \frac{z_j z_k f_J(\bz)}{\int_{\{\bv\tr\bv = t\}} f_J(\bv) \dd \bv} \dd \bz,	\]
for $z_j$ and $z_k$ the $j$-th and $k$-th elements of the vector $\bz \in \R^{J+1}$, respectively. The assertion follows from the continuity of $f_J$, and both integrands above in the formula for $\gamma_{j,k}^{1,1}(t)$, around the origin and the fundamental theorem of calculus.
\end{proof}	

For a non-degenerate Gaussian process we know that each $Z_i$ has a univariate, non-degenerate normal distribution. The distributions of $Z_i$ are generally correlated, with
	\begin{equation}	\label{covariance of Z}
	\cov\left(Z_i, Z_j\right) = \left\langle B \phi_i, \phi_j \right\rangle,	
	\end{equation}
for $B$ the covariance operator of the process $X$. But, for the special case of the $\left\{ \phi_i \right\}$ being the eigenbasis of $B$, $\left\{ Z_i \right\}$ is a sequence of independent random variables with non-degenerate normal distributions. Therefore, in this case the elements of the random vector $\widetilde{\bZ}$ are independent, and absolutely continuous random variables. Absolute continuity of the last element $\sqrt{\sum_{i=J+1}^\infty \langle \phi_{i}, X_1-x \rangle^2}$ follows from the orthonormality of the basis $\left\{\phi_i\right\}$, independence and absolute continuity of the terms of the sequence $\left\{ Z_i \right\} = \left\{ \langle \phi_i, X_1 - x \rangle \right\}$, and absolute continuity of $\left\Vert X_1 - x \right\Vert$. All the marginal densities of vector $\widetilde{\bZ}$ are positive and continuous at the origin. Therefore, for any non-degenerate Gaussian process with $\left\{ \phi_i \right\}$ its collection of eigenfunctions, the assumptions of {\sc Lemma~\ref{lem:aboutH3}} are satisfied, and $X$ satisfies (H\ref{hypo:gammafunction}).

\end{document}